\documentclass{article} 
\usepackage{arxiv_iclr,times}

\usepackage{amsmath,amsfonts,bm}

\def\eqref#1{equation~\ref{#1}}

\def\1{\bm{1}}

\DeclareMathAlphabet{\mathsfit}{\encodingdefault}{\sfdefault}{m}{sl}
\SetMathAlphabet{\mathsfit}{bold}{\encodingdefault}{\sfdefault}{bx}{n}

\usepackage[utf8]{inputenc} 
\usepackage[T1]{fontenc}    
\usepackage{hyperref}       
\usepackage{url}            
\usepackage{booktabs}       
\usepackage{amsfonts}       
\usepackage{nicefrac}       
\usepackage{microtype}      
\usepackage{xcolor}         
\usepackage{amsthm}
\usepackage{algorithmic}
\usepackage{algorithm}
\usepackage{caption}
\usepackage{multirow}
\newtheorem{definition}{Definition}
\newtheorem{theorem}{Theorem}
\usepackage{amsmath}
\usepackage{graphicx}
\usepackage{subcaption}

\title{Align, Then Adapt: Few-Shot fMRI Visual Reconstruction for New Subjects}

\author{%
Christos Zangos\thanks{Equal contribution.
$^{1}$Department of Computer Science, University of California,
Santa Barbara, USA.
$^{2}$Department of Psychological and Brain Sciences,
University of California, Santa Barbara, USA.
Correspondence to: Christos Zangos
\texttt{<christoszangos@ucsb.edu>}.}\hspace{0.40em}$^{1}$
\quad Danish Ebadulla$^{*1}$
\quad Thomas Sprague$^{2}$
\quad Ambuj Singh$^{1}$%
}

\iclrfinalcopy 

\begin{document}

\maketitle

\begin{abstract}
  Reconstructing visual images from fMRI data is a central but challenging problem in neuroscience. Despite recent progress, current methods fall short when data and computation are limited---precisely the conditions under which this task is most critical. We introduce Adapter Alignment, a staged, architecture-agnostic training procedure for adapting a pre-trained reconstruction model to a new subject with little data. Rather than letting a shared representation space emerge implicitly, lightweight subject-specific adapters are first explicitly aligned to a pre-trained reference subject on a small set of shared calibration images, and the full pipeline is then fine-tuned while preserving this alignment. The resulting representations not only place different subjects in a shared space but also align them semantically. Our results demonstrate significant gains in low-data regimes, improving reconstruction metrics by up to 8\%, and pre-diffusion retrieval shows that these gains reflect additional image-specific information in the fMRI embeddings rather than the generative prior. We also propose a coverage-based algorithm to select representative calibration images. Using both techniques together, we surpass baseline models fine-tuned with one hour of data using only 24 minutes of subject-specific data. The procedure is effective across architectures, reference subjects, and datasets, including adapting 3T subjects to a 7T reference.
\end{abstract}

\section{Introduction}\label{sec:intro}

Over the past several decades, the use of machine learning techniques has enabled the decoding and/or reconstruction of information represented in neural activation patterns. These approaches offer a key step toward characterizing and quantifying cognition \citep{naselaris2011encoding}. Beyond establishing the presence of information in brain activity, another important goal is to compare decoding or reconstruction performance across experimental manipulations, such as cued visual attention~\citep{serences2007feature, kamitani2005decoding,scolari2012optimal,sprague2013attention, itthipuripat2019functional, sprague2018dissociable} and visual working memory~\citep{serences2009stimulus, harrison2009decoding, ester2013neural, sprague2014reconstructions, christophel2012decoding, li2021joint}. However, for these approaches to be practical in a cognitive neuroscience laboratory, it is typically necessary to acquire large amounts of data.

Recent efforts have begun to address the challenges associated with decoding and reconstructing brain activity. 
Early approaches, including self-supervised learning frameworks ~\citep{beliy2019voxels,gaziv2022self,mozafari2020reconstructing, shen2019end,seeliger2018generative,ren2021reconstructing} struggled to generate semantically accurate image reconstructions, often resulting in images with limited resemblance to the original stimuli. Significant advancements were achieved by incorporating multimodal data, such as text \citep{lin2022mind}, and then by leveraging the generative power of latent diffusion models (LDMs)~\citep{takagi2023high, scotti2024reconstructing, lu2023minddiffuser,xia2024dream}. A more detailed literature review is presented in Appendix \ref{app:related_work}. However, these state-of-the-art methods require extremely large datasets with dozens of hours of fMRI data per participant. Thus, there is a clear need to maximize the data efficiency of image reconstruction methods so that a model for a new participant can be trained using as little subject-specific data as possible.

Recent works by \citet{scotti2024mindeye2} and \citet{liu2024mindslearningtransferableneural} have reduced the data needed for new subjects by projecting fMRI signals from multiple subjects into a shared representation space. However, in these pipelines the shared space is treated only as a convergence point of subject-specific pathways and is not explicitly aligned. We show in our analysis that though subject representations in this space share a similar structure, manipulations such as a simple orthogonal transformation substantially increase cross-subject similarity. This motivates an explicit alignment strategy: instead of relying on the shared space to emerge implicitly, lightweight subject-specific neural network modules called adapters \citep{liu2024mindslearningtransferableneural} can be trained to align each subject to a reference subject. 

We propose Adapter Alignment (AA), a training procedure for low-data adaptation of new subjects in shared-space reconstruction pipelines. Its components are established \citep{scotti2024mindeye2,ferrante2024through,haxby2011common}; our contribution is their organization into a two-stage procedure, explicit adapter alignment to a pre-trained reference followed by end-to-end fine-tuning that retains the alignment loss, which outperforms alternative alignment strategies (Appendix~\ref{app:train_strat}). AA substantially reduces data and training time while improving reconstruction quality. We also present a data-efficient image selection strategy that identifies the most representative training samples by covering key dimensions of variance in the common space. This method reduces the burden of data collection and further complements Adapter Alignment by making fine-tuning more efficient. Finally, AA is modular and composable: it can be integrated with existing reconstruction architectures (such as \citet{scotti2024reconstructing} and \citet{scotti2024mindeye2}) and complementary alignment techniques (such as \citet{dai2025mindaligner},\citet{wang2024mindbridge}), enabling practitioners to combine multiple strategies for improved performance. All the experiments were run on a single NVIDIA H200 GPU with 141GB of RAM.

\textbf{Scope.} AA is few-shot, not zero-shot: each new subject's adapter is fit to that subject's responses to a small calibration set shared with, or matched to, the reference. Our contributions are:

\begin{itemize}
    \item We show that subject representations in shared fMRI decoding spaces exhibit consistent structure but remain misaligned, and that explicit alignment substantially improves cross-subject correspondence, including on held-out images (Section~\ref{sec:common_visual}, Appendix~\ref{app:common_space}).
    \item We formally define Adapter Alignment and the image selection algorithm, and evaluate them through extensive experiments across subjects, datasets, and training regimes, showing consistently improved performance and faster convergence (up to 200x), particularly in limited-data settings (Sections \ref{sec:adaptalign}, \ref{sec:complete_data_tests}, \ref{sec:imageselection} and \ref{sec:limited_tests}). We also compare against standard selection baselines and evaluate pre-diffusion retrieval (Appendices~\ref{app:increase_evdim} and~\ref{app:retrieval}).
    \item We validate AA on NSD, THINGS ~\citep{hebart2023things}, and BOLD5000 \citep{chang2019bold5000} datasets, showing consistent low-data gains across architectures (Appendix \ref{app:diff_archs}), data acquisition protocols (7T vs. 3T) (Section \ref{sec:things_dataset}, Appendix \ref{app:bold_5k}), training strategies (Section \ref{sec:train_strat}, Appendix \ref{app:train_strat}), and reference--target subject pairs (Appendix \ref{app:diff_refsub}). 
\end{itemize}
\section{A Common Visual Representation Space}\label{sec:common_visual}

At the core of our approach is the hypothesis that each individual’s brain (or at least the tissue involved in processing static visual images) can be viewed as an instance of a shared 'common brain', with visual activity patterns that vary along similar dimensions in a shared representational space. The mapping of each individual's voxels into this space, however, must be learned from that individual's data. This idea is supported by prior neuroscience findings, including consistent category-selective tissue responses (e.g., faces, scenes) and continuous dimensions of high-level visual selectivity in ventral temporal cortex such as animate/inanimate, large/small (real-world size), and spiky/round \citep{huth2012continuous,long2018mid,bao2020map}. Retinotopic organization provides another example, where adaptive models transform individual anatomical mappings from retinal to cortical coordinates into a shared reference space \citep{sprague2013attention,wandell2015computational}. In our work, we analyze embeddings from multiple subjects and provide evidence that although fMRI signals do not naturally align in a shared space, they exhibit a similar underlying structure. We further explore structural and semantic properties of this space through eigenvector analysis, visual embedding comparisons, and alignment experiments, with additional results of semantic categories (e.g. animate / inanimate) provided in Appendix~\ref{app:common_concepts}.

Existing works, such as MindEye2~\citep{scotti2024mindeye2}, employ simple linear adapters to project subject-specific fMRI signals into a common representation space. A key indication of the existence of such a space is precisely this structural similarity across subjects. To investigate it, we use a pre-trained MindEye2 model to extract common-space embeddings for shared images across multiple subjects. Initial cross-subject cosine similarity is low, but as shown in Figure \ref{fig:cosine_evidence}, applying an orthogonal transformation with Procrustes analysis~\citep{schonemann1966generalized} reveals that the embeddings can be structurally aligned. However, this alignment is still incomplete, reflecting both the limitations of orthogonal transformations and the absence of an explicit alignment objective.

\begin{figure}[ht]
\centering
\scalebox{0.7}{
\begin{subfigure}[t]{0.32\columnwidth}
    \centering
    \includegraphics[width=\columnwidth]{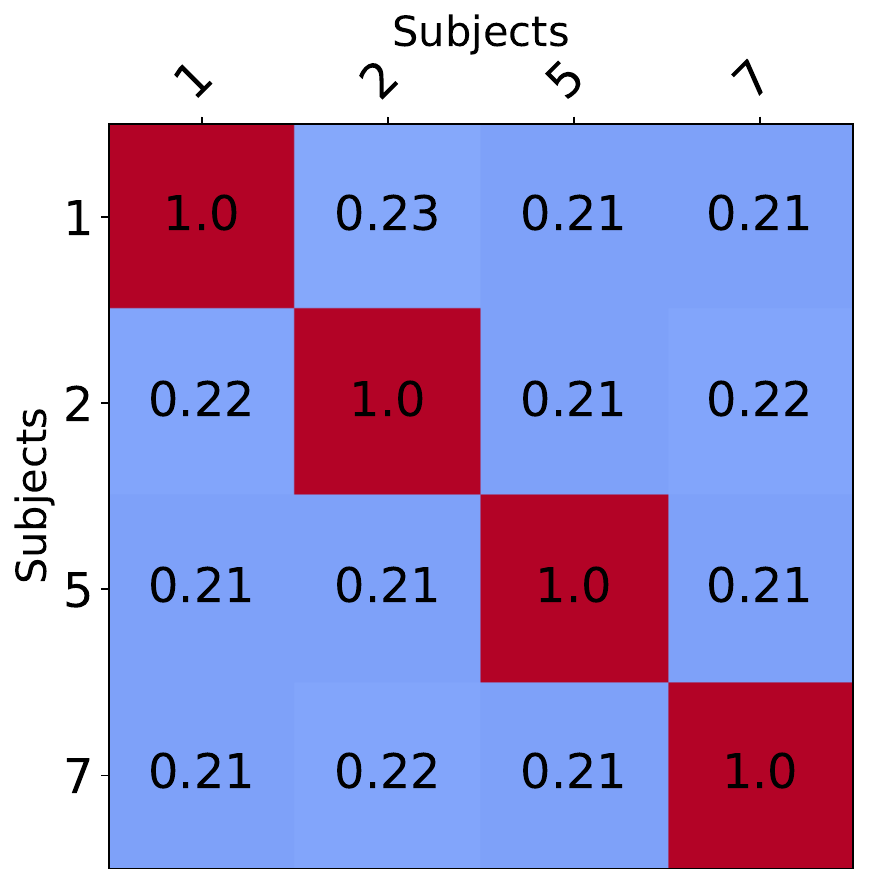}
    \caption{Original alignment.}
\end{subfigure}
\hfill
\begin{subfigure}[t]{0.32\columnwidth}
    \centering
    \includegraphics[width=\columnwidth]{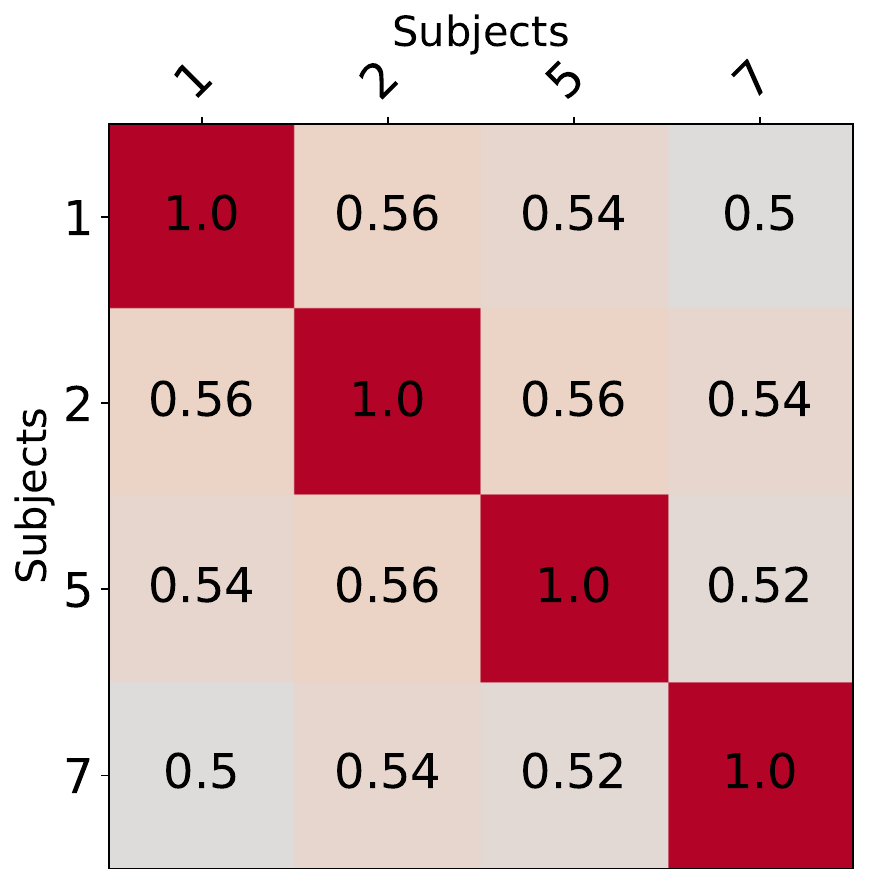}
    \caption{Post orthogonal transformation.}
\end{subfigure}
\hfill
\begin{subfigure}[t]{0.32\columnwidth}
    \centering
    \includegraphics[width=\columnwidth]{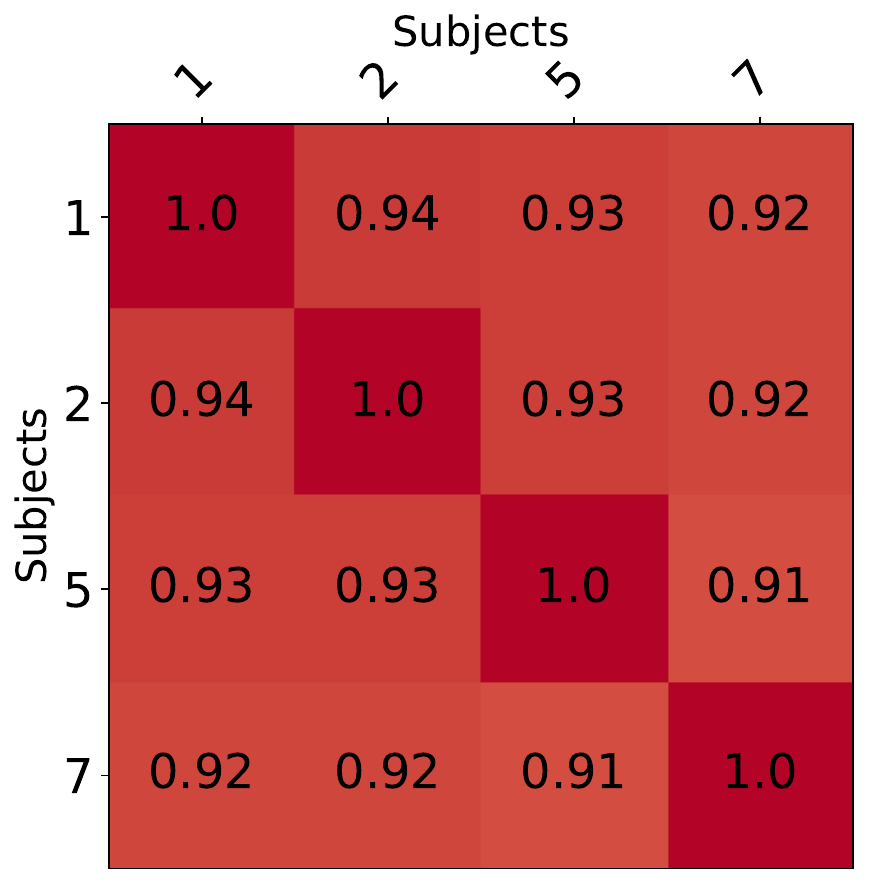}
    \caption{Post Adapter Alignment.}
\end{subfigure}
}
\caption{Cross-subject cosine similarity of fMRI embeddings in the common space extracted from a pre-trained MindEye2 model. (a) Cross-subject alignment is weak before any transformation. (b) Applying an orthogonal transformation substantially improves cross-subject similarity. (c) Explicit alignment through lightweight adapter modules achieves near-perfect alignment. Only Subjects 1, 2, 5, and 7 are shown for consistency throughout the manuscript.}
\label{fig:cosine_evidence}
\end{figure}

As a preliminary experiment, we pre-trained a single-subject model and explicitly mapped subsequent subjects to this reference subject using lightweight subject-specific adapters. Even with a simple adapter design (a linear layer with a GELU), this approach yielded near-perfect alignment in the common space (Figure \ref{fig:cosine_evidence} (c)). This agreement is expected by construction on the aligned images; Appendix~\ref{app:common_space} shows that alignment also improves on held-out images, including within independently defined retinotopic and category-selective regions (functional, not anatomical, correspondence). The choice of adapter is not critical. Any transformation that goes beyond an orthogonal mapping can substantially improve alignment highlighting that it is the explicit alignment objective, rather than architectural complexity, that drives the improvement.

We next develop a training procedure that explicitly enforces such alignment and compare it with standard baseline fine-tuning (further analysis in Appendix~\ref{app:common_space}).

\section{Methods}\label{sec:reconstruction}
We present two methods that leverage the presence of a shared space for fMRI-to-image reconstruction. Our framework introduces Adapter Alignment (AA) as a lightweight strategy for subject-specific alignment, integrates this into a reconstruction pipeline, and augments it with a greedy image selection algorithm to further improve data efficiency. Together, these components form the basis for the experiments presented in the following section.

\subsection{Dataset and Split Strategy}\label{sec:datasets}
We primarily use the Natural Scenes Dataset (NSD) \citep{allen2022massive}, which provides high-SNR 7T fMRI recordings of subjects viewing natural images. We also evaluate transfer to THINGS-fMRI \citep{hebart2023things}, which contains 3T data collected under different stimulus and task protocols, and to BOLD5K \citep{chang2019bold5000}.

For NSD, we adopt a modified split strategy compared to the standard setup. Traditionally, images shared across subjects are reserved for testing, while training consists of subject-specific unique images. However, ideal Adapter Alignment requires a one-to-one correspondence between subjects. We therefore include the shared images in the training set and move an equal number of previously unique training images into the test set. As a result, each subject is evaluated on a distinct set of unseen images. Test splits are constructed to maintain a roughly balanced category distribution. To ensure fairness, all baselines are re-trained and evaluated on this modified split, and results are averaged across three independently generated splits.

A matched rerun on the standard split shows that this test set is harder (Appendix~\ref{app:standard_split}); baselines use their published, untuned hyperparameters, which AAMax inherits. Similar to prior work, we focus on Subjects 1, 2, 5, and 7, as they have complete 40-hour fMRI datasets available in NSD, enabling better reconstruction performance. 

\begin{figure*}[ht]
    \centering
    \includegraphics[width=0.70\textwidth]{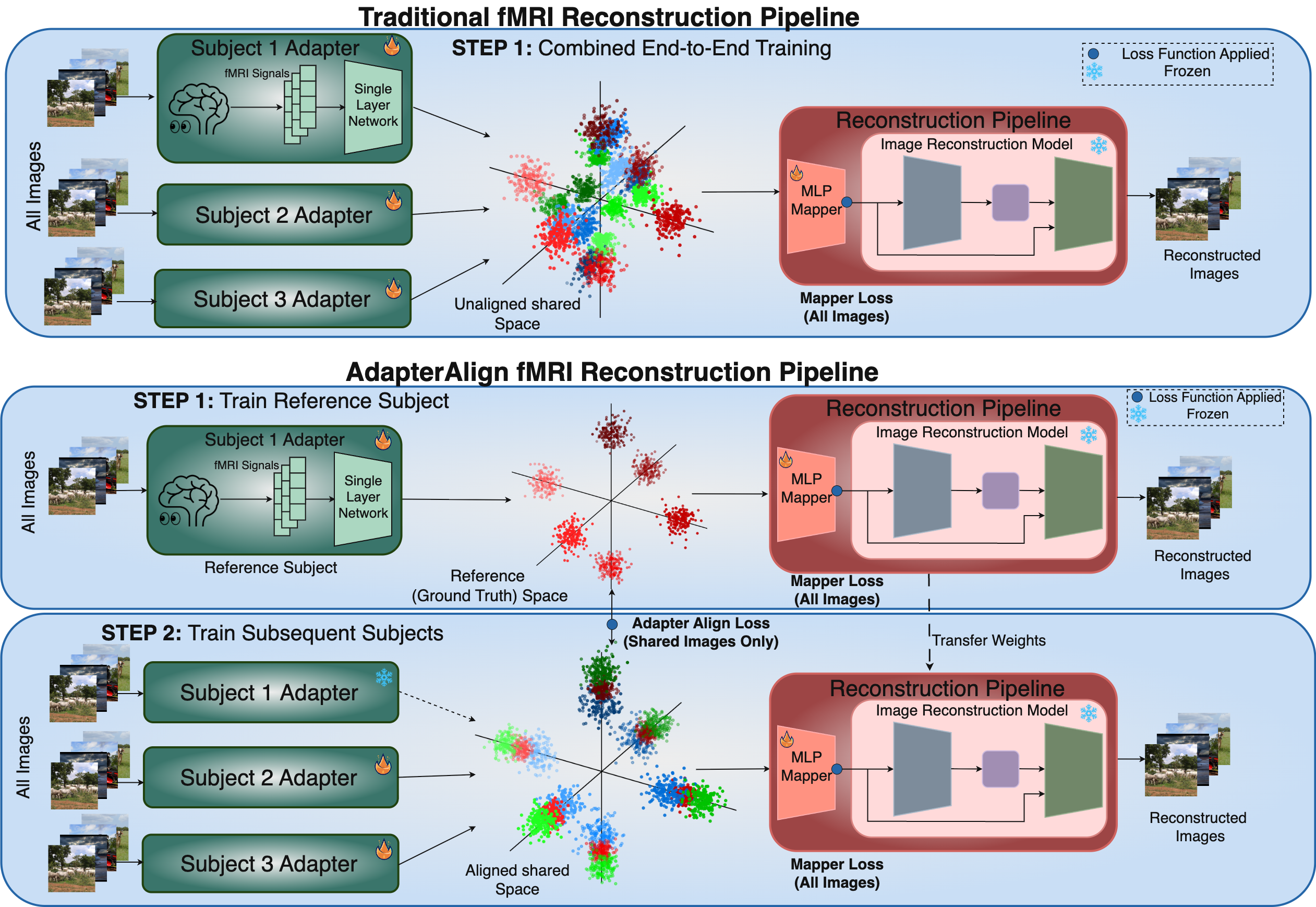}
    \caption{Training procedures for fMRI-Image Reconstruction. \textbf{Top:} Traditional Approach to reconstruction. All subjects are trained simultaneously and loss functions are applied at the MLP Mapper to align the output to pre-trained CLIP space. \textbf{Bottom:} Proposed AdapterAlign training pipeline. First, we train a single reference subject and obtain its embeddings in the shared space. Then we train subsequent subjects to minimize loss at both the shared space and the mapper. The same mapper weights are transferred over in step 2.}
    \label{fig:adapt_align}
\end{figure*}

\subsection{Adapter Alignment}\label{sec:adaptalign}

Current approaches for training multi-subject pipelines with a common representation space typically follow an end-to-end strategy, where fMRI inputs are processed by the model and the loss is computed at the final layer to minimize the discrepancy between the ground-truth and predicted embeddings. Throughout the paper, we refer to this fine-tuning strategy as the baseline. In this setup, the common space emerges organically during training without explicit constraints. However, this approach has two key drawbacks: (1) it requires extended training to form a satisfactory common space and achieve convergence, and (2) the subject-specific embeddings in the common space are not strongly aligned, making it harder for the residual MLP to map multiple subject inputs to a unified output (Figure \ref{fig:adapt_align}).

To address these challenges, we propose Adapter Alignment (AA) training. In AA training, the pipeline is first pre-trained end-to-end on a single reference subject, chosen based on reconstruction performance and data availability (Row 2 of Figure \ref{fig:adapt_align}). Among NSD subjects, Subject 1 consistently exhibited the best reconstruction performance, making it the ideal candidate to construct the reference space. Other reference choices are evaluated in Appendix~\ref{app:diff_refsub}. After pre-training, the output embeddings for the subject-specific adapter, particularly those corresponding to the shared images, are extracted.

When fine-tuning a new subject, we adopt a two-step process. The first is \textbf{Adapter Alignment}, where we train the new subject’s adapter using the shared images, with the objective of aligning its output to the embeddings of Subject 1 for the same images. The adapter is allowed to overfit on these shared images until the loss is minimized. The second step is \textbf{End-to-End Fine-tuning}, where we unfreeze the entire pipeline and resume training, applying the usual losses at the final output (Row 3 of Figure \ref{fig:adapt_align}). Additionally, we introduce an MSE loss at the adapter level specifically for the common images, which allows the adapter to restructure itself to account for subject-specific discrepancies and adapt to the unique images encountered by the new subject. We refer to this process of overfitting the adapter with AA and fine-tuning end-to-end "AAMax". We present additional experiments in Appendix \ref{app:common_space} to show that this two step process improves alignment in the shared space at all stages.

Our training method is effective across both high-data and low-data regimes. With full training data, AAMax achieves performance comparable to traditional end-to-end training while converging more quickly. In limited data settings, AAMax shows a clear advantage, producing higher-quality reconstructions and greater accuracy with substantially fewer training samples. 

\subsection{Reconstruction Pipeline}\label{sec:pipeline}

In general, fMRI-to-image reconstruction pipelines follow a common structure: brain signals are mapped into an intermediate embedding space (semantic or virtual), which is then decoded into an image using a pre-trained generative model such as a diffusion decoder. We discuss multiple previous approaches and their pipelines in Appendix \ref{app:related_work}. Since Adapter Alignment (AA) is robust to architectural choices, we adopt the pipeline used in MindEye1 for our main experiments and refer readers to MindEye1 \citep{scotti2024reconstructing} and MindEye2 \citep{scotti2024mindeye2} for architectural details. We also validate our method on MindEye2's architecture (Appendix \ref{app:diff_archs} and \ref{app:me2_diffsub}) and across varying component choices and shared space dimensionalities (Appendix \ref{app:scale}).

In summary, flattened fMRI voxels from the ``nsdgeneral'' region are passed through a subject-specific adapter (a single linear layer with GELU non-linearity) into a shared representation space. This shared representation is then processed by a residual MLP. The resulting output is mapped into the CLIP \citep{radford2021learning} space via a pre-trained diffusion prior \citep{ramesh2022hierarchical}. This serves as a conditioning signal for a diffusion model to produce the final reconstruction. 

We detail the training losses used in this pipeline for both the high-level (semantic) and low-level (visual) branches below, with full equations provided in Appendices \ref{app:hlp_breakdown} and \ref{app:llp_breakdown}. The overall loss function for the high-level pipeline is:

\[
\mathcal{L}_{\text{new}} =
\left\{
\begin{array}{ll}
\lambda_1 \mathcal{L}_{\text{Prior}}
+ \lambda_2 \mathcal{L}_{\text{CLIP}},
& \text{unique images} \\[0.5em]
\lambda_1 \mathcal{L}_{\text{Prior}}
+ \lambda_2 \mathcal{L}_{\text{CLIP}} \\
\quad + \lambda_3 \mathcal{L}_{\text{MSE}}
(\mathbf{z}_{\text{adapter,new}}, \mathbf{z}_{\text{adapter,ref}}),
& \text{common images}
\end{array}
\right.
\]

where \(\mathbf{z}_{\text{adapter, new}}\) and \(\mathbf{z}_{\text{adapter, ref}}\) are the adapter-level embeddings for the new and reference subjects, respectively, on the shared common images. 

The low-level pipeline also follows a similar structure to the high-level pipeline but maps the fMRI data to latent embeddings derived from a custom autoencoder. This autoencoder is trained on downscaled images from the NSD dataset (64x64x3) and reduces them to a latent space of size 1024. The low-level pipeline learns to map to this latent space; passing these embeddings through a pre-trained decoder produces blurry reconstructions. The loss function for this stage is:
\[
\mathcal{L}_{\text{new}} =
\left\{
\begin{array}{ll}
\mathcal{L}_{\text{MSE}}
(\mathbf{z}_{\text{fMRI,new}}, \mathbf{z}_{\text{AE}}),
& \text{unique images} \\[0.5em]
\mathcal{L}_{\text{MSE}}
(\mathbf{z}_{\text{fMRI,new}}, \mathbf{z}_{\text{AE}}) \\
\quad + \mathcal{L}_{\text{MSE}}
(\mathbf{z}_{\text{adapter,new}}, \mathbf{z}_{\text{adapter,ref}}),
& \text{common images}
\end{array}
\right.
\]

where \(\mathbf{z}_{\text{fMRI, new}}\) is the new subject's output embedding, \(\mathbf{z}_{\text{AE}}\) is the image's AE latent, and the adapter-level embeddings are defined as above. After mapping brain signals to the CLIP space, the mappings are fed as conditioning to a frozen Diffusion Model \citep{xu2023versatile} to generate final reconstructions.

\subsection{Greedy Algorithm for Best Candidate Image Selection}\label{sec:imageselection}
In the context of fMRI-based visual reconstruction, we introduce a greedy algorithm that selects the most representative subset of images from the common space. We first decompose (SVD) the adapter’s weight matrix of the reference subject. Then we project the embeddings of the shared set of 1000 images onto the $d$ principal dimensions corresponding to the principal singular vectors. 
For each dimension $j$, we partition the range into \( B_j \) bins that is determined by the ratio of singular value \( \lambda_j \) to the largest singular value \( \lambda_1 \):
\[
B_j = \left\lfloor w \cdot \frac{\lambda_j}{\lambda_1} \right\rfloor, \quad j = 1, 2, \dots, d
\]
where \( w \) is a predefined scaling parameter. Dimensions with larger singular values are partitioned into more bins, reflecting their greater variance and importance in capturing meaningful data features. 

Each image (and its fMRI) maps to exactly one bin in each dimension and our goal is to find the smallest subset of images that covers all bins in each dimension. This subset is critical in selecting the most representative images for a new subject, ensuring that the selected images capture the key activity patterns across subjects. The problem as stated is an NP-hard set-cover instance. The maximum ``set'' size is $d$, the number of principal eigenvector dimensions. The greedy heuristic below therefore achieves an $H(d)$-approximation~\citep{johnson1974approximation,lovasz1975ratio}, where
$H(d)$ is the harmonic series sum up to $d$. Stopping at a fixed budget, as in our experiments, gives a $(1-1/e)$-approximation for maximum coverage~\citep{nemhauser1978analysis}. 

Given a subset $S$ of images, let $Gap(S,j)$ be the number of empty bins in dimension $j$, and let $Gap(S) = \Sigma_j Gap(S,j)$ be the total number of empty bins across all dimensions. 
At every step, the greedy algorithm chooses image $i \not \in S$ that minimizes $Gap$ when added to $S$. 

After applying the image selection algorithm, the chosen subset simply replaces random sampling for fine-tuning. The training procedure itself is unchanged: the selected images can be used with either baseline fine-tuning or AAMax. The algorithm serves purely as a strategy to identify the most representative subset of images, reducing the data requirement without altering the reconstruction pipeline. Formal proofs for the algorithm are presented in Appendix \ref{app:nphardness}.

\section{Results}
We now evaluate the performance of our proposed framework across a range of experimental conditions. We begin with complete data experiments, where AAMax is compared against baseline training. We then turn to limited data settings, where we report results from both AAMax and AAMax combined with image selection in a unified comparison. We also analyze how AAMax scales with increasing amounts of training data, highlighting that its advantage diminishes as larger datasets allow even unaligned models to generalize. Finally, we present experiments showing that AAMax is effective across architectures using MindEye1 and MindEye2 \citep{scotti2024mindeye2, scotti2024reconstructing} architectures.

\subsection{Tests on Complete Data} \label{sec:complete_data_tests}
We first validate our training method using the full dataset. Subject 1 is used as the reference subject, while fine-tuning is performed on Subjects 2, 5, and 7. The results averaged over all 3 subjects are presented in Tables \ref{tab:fulldatacomparison} and \ref{tab:fulldatacomparison_sum}. AA can be used with any reference subject, but the quality of the reference affects fine-tuning performance (Appendix \ref{app:diff_refsub}). We use the same set of metrics as \citet{scotti2024mindeye2} to evaluate our reconstructions. Pixel Correlation, SSIM and 2-way percent correct for the 2nd and 5th layer of AlexNet are considered low level metrics. 2-way percent correct for Inception\citep{szegedy2016rethinking}, CLIP, EfficientNet\citep{tan2019efficientnet} and SwAV\citep{caron2020unsupervised} score are considered high level metrics. When trained to completion, both baseline fine-tuning and AAMax show similar overall performance. However, the MSE loss at the adapter level is significantly lower (by an order of magnitude as shown in Appendix \ref{sec:mse}) when using adapter alignment. This allows for more precise post-hoc analysis after extracting embeddings from the common space.

\begin{table}[h]

    \centering
    \caption{Quantitative results on training the reconstruction models on 40 hours of data. The fine-tuning results are averaged over subjects 2,5 and 7. AAMax and baseline fine-tuning perform similarly across all metrics. Std. deviations are provided in Appendix \ref{app:sum_stat}.}
    
    \vskip 0.1in
    \setlength{\tabcolsep}{2pt}  
    \renewcommand{\arraystretch}{1.1}  
    \footnotesize  

    \resizebox{0.7\columnwidth}{!}{  
    
    \begin{tabular}{lccccccccc}
        \toprule
        & \multicolumn{4}{c}{\textbf{Low-Level}} & \multicolumn{4}{c}{\textbf{High-Level}} \\
        \cmidrule(lr){2-5} \cmidrule(lr){6-9}
        \textbf{Method} & PixCorr $\uparrow$ & SSIM $\uparrow$ & Alex(2) $\uparrow$ & Alex(5) $\uparrow$ & Incep $\uparrow$ & CLIP $\uparrow$ & Eff $\downarrow$ & SwAV $\downarrow$\\
        \midrule
        
        Subject 1 Pre-train & 0.345 & 0.346 & 92.80\% & 96.88\% & 94.40\% & 90.02\% & 0.692 & 0.399 \\ 
        \cline{2-9}
        Baseline FT AVG & 0.258& \textbf{0.339} & \textbf{89.64\%} & 95.05\%& 92.63\% & \textbf{89.64\%} & 0.717 & 0.427 \\  
        AAMax FT AVG & \textbf{0.259} & 0.337 & 89.20\% & \textbf{95.16\%} & \textbf{92.68\%} & 89.49\% & \textbf{0.713} & \textbf{0.423} \\  
        \bottomrule
    \end{tabular}
    }
    \label{tab:fulldatacomparison}
\end{table}

\begin{table*}[h]
    \centering
    \caption{Limited data fine-tuning results (1 hour, 250 images) using Subject 1 as the reference and Subjects 2, 5, and 7 as fine-tuning subjects. Each subject is evaluated under four settings: Baseline with random sampling, AAMax with random sampling, Baseline with image selection (IS), and AAMax with IS. Results are averaged over three runs with fixed random splits. AAMax consistently outperforms the baseline, IS provides additional gains in both training regimes, and the best performance is achieved by combining AAMax with IS. Pixel-level metrics are shown in gray since they are not optimized in these runs but are still valid.}
    \vskip 0.1in
    \setlength{\tabcolsep}{2pt}  
    \renewcommand{\arraystretch}{1.1}  
    \scriptsize  
    \resizebox{1.0\columnwidth}{!}{  
    \begin{tabular}{lccccccccc}
        \toprule
        & \multicolumn{4}{c}{\textbf{Low-Level}} & \multicolumn{4}{c}{\textbf{High-Level}} \\
        \cmidrule(lr){2-5} \cmidrule(lr){6-9}
        \textbf{Method} & PixCorr $\uparrow$ & SSIM $\uparrow$ & Alex(2) $\uparrow$ & Alex(5) $\uparrow$ & Incep $\uparrow$ & CLIP $\uparrow$ & Eff $\downarrow$ & SwAV $\downarrow$\\
        \midrule
        Subject 2 FT 1HR Baseline      & \color{gray}{0.078$\pm$0.010} & \color{gray}{0.264$\pm$0.007} & 75.33\%$\pm$1.62\% & 84.18\%$\pm$1.40\% & 74.15\%$\pm$0.99\% & 73.66\%$\pm$1.3\% & 0.862$\pm$0.011 & 0.524$\pm$0.009 \\
        Subject 2 FT Selected 1HR Baseline & \color{gray}{0.097$\pm$0.003} & \color{gray}{0.288$\pm$0.004} &  78.78\%$\pm$0.41\% &  85.78\%$\pm$0.69\% & 77.13\%$\pm$0.76\% & 76.17\%$\pm$0.03\% & 0.842$\pm$0.003 & 0.513$\pm$0.005 \\
        Subject 2 FT 1HR AAMax  & \color{gray}{0.112$\pm$0.005} & \color{gray}{0.259$\pm$0.025} &  78.75\%$\pm$1.29\% &  88.04\%$\pm$1.19\% & 78.99\%$\pm$0.97\% & 78.28\%$\pm$0.39\% & 0.824$\pm$0.003 & 0.513$\pm$0.014 \\
        Subject 2 FT Selected 1HR AAMax  & \color{gray}{0.102$\pm$0.003} & \color{gray}{0.267$\pm$0.11} &  \textbf{81.41\%$\pm$0.27\%} &  \textbf{89.55\%$\pm$0.26\%} & \textbf{80.43\%$\pm$0.13\%} & \textbf{80.33\%$\pm$0.01\%} & \textbf{0.823$\pm$0.002} & \textbf{0.496$\pm$0.002} \\
        \hline
        Subject 5 FT 1HR Baseline       & \color{gray}{0.082$\pm$0.002} & \color{gray}{0.267$\pm$0.002} &  75.05\%$\pm$0.08\% &  83.28\%$\pm$0.27\% & 76.04\%$\pm$1.12\% & 76.40\%$\pm$0.64\% & 0.855$\pm$0.002 & 0.526$\pm$0.006 \\
        Subject 5 FT Selected 1HR Baseline & \color{gray}{0.098$\pm$0.001} & \color{gray}{0.286$\pm$0.003} &  77.07\%$\pm$1.0\% &  86.31\%$\pm$0.1\% & 78.09\%$\pm$0.30\% & 79.55\%$\pm$0.06\% & 0.827$\pm$0.004 & 0.499$\pm$0.003 \\
        Subject 5 FT 1HR AAMax  & \color{gray}{0.105$\pm$0.001} & \color{gray}{0.278$\pm$0.005} &  79.90\%$\pm$0.34\% &  88.68\%$\pm$0.07\% & 80.63\%$\pm$0.31\% & 82.20\%$\pm$0.38\% & 0.808$\pm$0.004 & 0.493$\pm$0.02 \\
        Subject 5 FT Selected 1HR AAMax  & \color{gray}{0.117$\pm$0.004} & \color{gray}{0.285$\pm$0.009} &  \textbf{81.05\%$\pm$0.24\%} &  \textbf{89.88\%$\pm$0.78\%} & \textbf{81.92\%$\pm$0.41\%} & \textbf{82.47\%$\pm$0.43\%} & \textbf{0.799$\pm$0.001} & \textbf{0.480$\pm$0.004} \\
        \hline
        Subject 7 FT 1HR Baseline   & \color{gray}{0.079$\pm$0.006} & \color{gray}{0.252$\pm$0.15} &  72.98\%$\pm$3.03\% &  79.98\%$\pm$1.39\% & 71.57\%$\pm$1.04\% & 72.07\%$\pm$0.86\% & 0.871$\pm$0.009 & 0.539$\pm$0.014 \\
        Subject 7 FT Selected 1HR Baseline  & \color{gray}{0.099$\pm$0.004} & \color{gray}{0.286$\pm$0.005} &  78.41\%$\pm$0.08\% &  84.77\%$\pm$1.13\% & 74.91\%$\pm$1.52\% & 75.79\%$\pm$0.21\% & 0.844$\pm$0.006 & 0.515$\pm$0.001 \\
        Subject 7 FT 1HR AAMax  & \color{gray}{0.104$\pm$0.002} & \color{gray}{0.285$\pm$0.001} &  80.04\%$\pm$0.22\% &  87.47\%$\pm$0.78\% & 79.17\%$\pm$0.38\% & 78.93\%$\pm$0.53\% & 0.826$\pm$0.001 & 0.509$\pm$0.001 \\
        Subject 7 FT Selected 1HR AAMax  & \color{gray}{0.123$\pm$0.003} & \color{gray}{0.298$\pm$0.019} &  \textbf{81.24\%$\pm$1.36\%} &  \textbf{88.34\%$\pm$0.30\%} & \textbf{79.94\%$\pm$1.00\%} & \textbf{79.81\%$\pm$0.91\%} & \textbf{0.820$\pm$0.001} & \textbf{0.495$\pm$0.007} \\
        \bottomrule
    \end{tabular}
    }
    \label{tab:main_limiteddata_table}
\end{table*}

\subsection{Tests on Limited Data}\label{sec:limited_tests}
A key objective of multi-subject pipelines is to achieve accurate models of new subjects with limited data. Because acquiring fMRI data is costly, reducing the data requirements makes reconstruction pipelines more practical and broadly applicable.

In our experiments, we designate Subject 1 as the reference and fine-tune Subjects 2, 5, and 7 using the shared images. For consistency, we use the same reference subject checkpoint across all experiments and report averaged results over three runs with fixed random splits. Table~\ref{tab:main_limiteddata_table} presents the unified results for fine-tuning with 1 hour of data (250 images), comparing baseline training and AAMax with both random sampling and image selection (IS). Across all subjects, AAMax significantly outperforms the baseline, while IS consistently improves performance regardless of the fine-tuning strategy. As expected, the strongest results are obtained when combining AAMax with IS, showing that the two approaches are complementary. With all other settings identical, this isolates the contributions of alignment and IS (further ablations in Appendices~\ref{app:train_strat} and~\ref{app:epoch_ablations}). For IS, we select images based on the top 40 eigenvector dimensions, with ablations on this choice provided in Appendix~\ref{app:increase_evdim}. We run experiments with scaling the amount of training data in the following section. Ablations with 2 and 4 hours of data are included in Appendix~\ref{app:more_limited}, where the same trend holds. We also run experiments on convergence in Appendix \ref{app:more_limited} and find that AA followed by just 1 epoch of End-to-End fine-tuning is good enough to outperform the baseline. Appendix~\ref{app:increase_evdim} compares IS with standard selection baselines, and Appendix~\ref{app:retrieval} shows that AAMax also improves pre-diffusion retrieval, so its gains are not due to the generative prior. 

We further investigate the efficiency of our approach by testing how far data requirements can be reduced while maintaining reconstruction performance. In this setup, Subjects 2 and 5 were fine-tuned from Subject 1 using 100, 150, and 200 training images. At 100 images, both AAMax and AAMax+IS already match the baseline trained with 1 hour of data, and at 150 images they clearly surpass it. As shown in Figure~\ref{fig:low_data_push}, AAMax+IS consistently achieves the strongest results, though even AAMax alone outperforms the baseline in these low-data regimes. These findings highlight the robustness of adapter alignment and demonstrate that effective reconstructions can be achieved with as little as 40\% of the data previously considered necessary.

\begin{figure}[!h]
    \centering
    \scalebox{0.85}{
    \begin{subfigure}{0.43\columnwidth}
        \centering
        \includegraphics[width=\columnwidth]{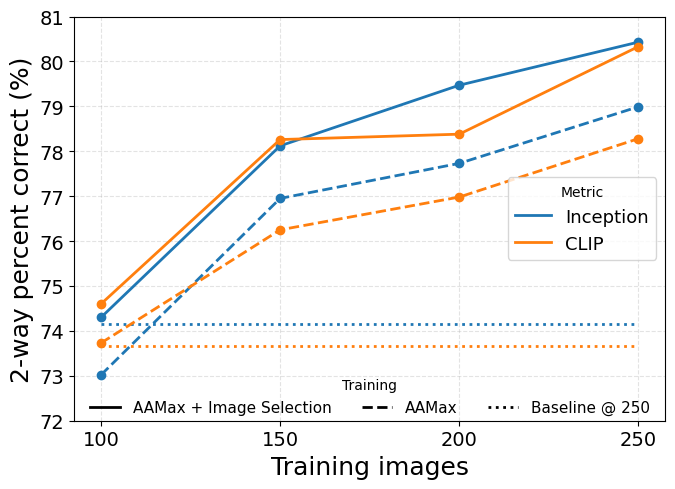}
    \end{subfigure}
    \begin{subfigure}{0.43\columnwidth}
        \centering
        \includegraphics[width=\columnwidth]{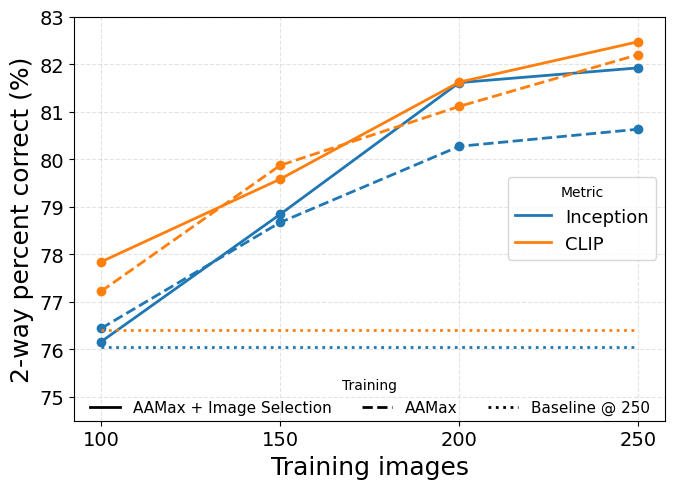}
    \end{subfigure}\hfill
    }
    \caption{Pushing the limits of data efficiency with AAMax+Image Selection. Left: Subject 2 fine-tuned from Subject 1. Right: Subject 5 fine-tuned from Subject 1. In both cases, AAMax and AAMax+Image Selection outperform the baseline trained with 250 images, even when using only 150 images (a 40\% reduction) and match performance with 100 images (a 60\% reduction). Image selection in these experiments uses the top 20 eigenvector dimensions (Appendix~\ref{app:increase_evdim}).}
    \label{fig:low_data_push}
\end{figure}

\subsection{Integrating with Existing Alignment Methods} \label{sec:train_strat}

\begin{table}[h]
    \centering
    \caption{\textbf{Comparison and combination of MindAligner and Adapter Alignment.}
    All methods use MindEye1 with pretraining on Subject~1 and fine-tuning on Subject~2 for 1 hour of data.
    MindAligner follows the original paper, using categorical stimulus pairs.
    AA uses our 250-image Image Selection subset that spans the breadth of the principal eigenspace.
    The combined method applies MindAligner for coarse voxel-space alignment followed by AAMax + Image Selection for fine-grained common-space alignment. The combined approach achieves the strongest overall performance.}

    \footnotesize
    \resizebox{0.8\columnwidth}{!}{  
    \begin{tabular}{lccccccccc}
        \toprule
        & \multicolumn{4}{c}{\textbf{Low-Level}} & \multicolumn{4}{c}{\textbf{High-Level}} \\
        \cmidrule(lr){2-5} \cmidrule(lr){6-9}
        \textbf{Method} & PixCorr $\uparrow$ & SSIM $\uparrow$ & Alex(2) $\uparrow$ & Alex(5) $\uparrow$ & Incep $\uparrow$ & CLIP $\uparrow$ & Eff $\downarrow$ & SwAV $\downarrow$\\
        \midrule
        MindAligner (MA) & \color{gray}{0.116} & \color{gray}{0.280} &  82.11\% &  88.68\% & 81.36\% & 79.69\% & 0.809 & 0.483 \\
        AAMax + IS & \color{gray}{0.118} & \color{gray}{0.285} &  81.78\% &  89.61\% & \textbf{81.49\%} & 79.55\% & 0.810 & 0.491 \\
        MA + AAMax + IS & \color{gray}{0.124} & \color{gray}{0.284} &  \textbf{82.24\%} &  \textbf{89.95\%} & \textbf{81.49\%} & \textbf{80.99\%} & \textbf{0.795} & \textbf{0.480} \\
        \bottomrule
    \end{tabular}
    }
    \label{tab:mindalign_results}
\end{table}

AA combined with IS can also be integrated with existing alignment methods, enabling a modular training strategy. This is possible when the methods operate at different stages of the decoding pipeline and thus address complementary sources of cross-subject variability.

For alignment methods that operate at earlier stages of the pipeline, such as voxel-space or categorical alignment, AAMax+IS can be applied downstream without interference. MindAligner \citep{dai2025mindaligner} performs coarse alignment directly in fMRI voxel space using category-level supervision. In contrast, AAMax operates in the learned common representation space after the subject-specific adapter, enforcing fine-grained geometric alignment. Because these methods act on distinct representations, they can be combined sequentially: coarse alignment reduces subject variability at the signal level, while AAMax+IS refines alignment in the shared space. As shown in Table \ref{tab:mindalign_results}, combining MindAligner with AAMax+IS yields the strongest overall performance, demonstrating that Adapter Alignment is not a competing alternative but a complementary refinement stage. A detailed breakdown of how we integrate MindAligner and AAMax+IS is provided in Appendix \ref{app:mind_align}.

In contrast, methods such as MindBridge \citep{wang2024mindbridge} and \citet{ferrante2024through} perform alignment at the same level of the pipeline as AA, namely through subject-specific adapters mapping into a shared space. For these methods, direct integration of the alignment method is not possible; instead, the relevant comparison is between alignment training strategies. To isolate this effect, we compare three regimes under identical architectural and data conditions in Appendix \ref{app:train_strat}. We find that AAMax consistently outperforms both alternatives across low-level and high-level metrics. This demonstrates that, when alignment is performed at the same stage of the pipeline, the training paradigm is the key differentiator, and AAMax is the most effective of the strategies compared.

Finally, we emphasize that AAMax+IS is compatible with many additional components proposed in prior work. Techniques such as synthetic fMRI generation, auxiliary regularization, or improved diffusion backbones are orthogonal to our contribution and can be incorporated alongside AAMax for further gains. Taken together, these results position Adapter Alignment not as a replacement for existing alignment methods, but as a general-purpose alignment module that can either (i) be combined with upstream alignment approaches or (ii) outperform alternative strategies when operating at the same representational level.

\begin{figure*}[ht]
    \centering
    \begin{subfigure}{0.75\textwidth}
        \centering
        \includegraphics[width=\columnwidth]{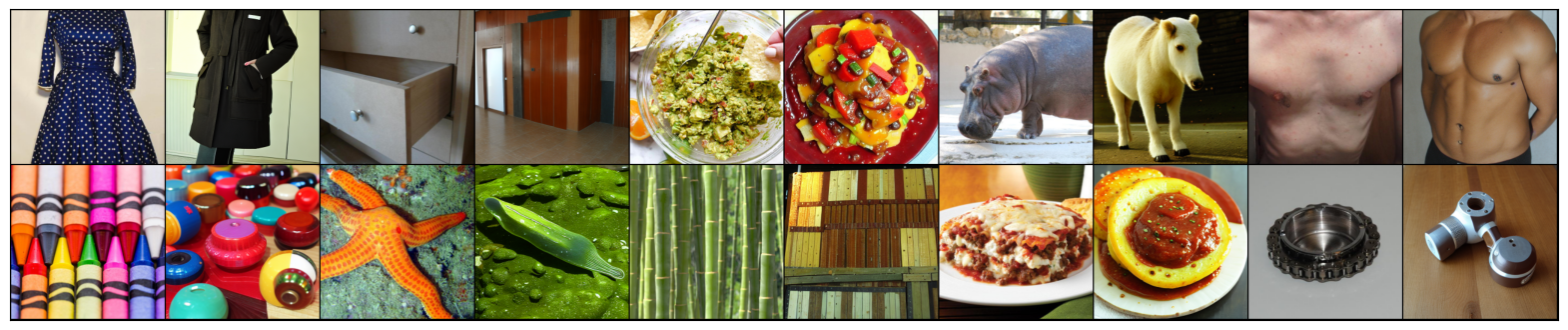}
    \end{subfigure}\hfill
    \caption{Examples of reconstructed images from Subject 1 of the THINGS dataset after fine-tuning. Each pair depicts the ground truth (left) and the reconstructed image (right).}
    \label{fig:things_recons}
\end{figure*}
\section{Transfer to Other Datasets}\label{sec:things_dataset}
To test whether a reference model trained on one dataset can be used to adapt subjects from another, we conducted experiments on the THINGS dataset. We used a model pre-trained on Subject 1 from the NSD dataset as the reference and fine-tuned it on data from Subject 1 of the THINGS dataset. Since no images are shared between THINGS and NSD, we construct an approximate common set: we compute the CLIP similarity between each THINGS training image and the 1{,}000 NSD shared images and retain pairs with cosine similarity above 0.90, yielding approximately 150 semantically matched pairs (Appendix~\ref{app:bold_5k}). As shown in Table \ref{tab:things_comparison}, AAMax again outperforms conventional fine-tuning, although correspondence quality matters (Appendix~\ref{app:bold_5k}).

Several factors may account for the performance differences between the two datasets. Recent fMRI reconstruction reports have largely relied on the NSD, which provides high-SNR 7T fMRI acquired at 1.8 mm resolution, but this level of data quality is often out of reach for many neuroscience labs. THINGS uses a more widely available 3T acquisition at 2 mm resolution. The datasets also differ in SNR (7T vs 3T), stimulus presentation protocol, stimulus set (NSD: MS-COCO, THINGS: curated object images), and task performed by participants (NSD: recognition memory on each image; THINGS: object or not?). These differences make variation in decoding and reconstruction performance expected, rather than evidence that models fail to transfer across datasets.

We further evaluated cross-dataset transfer using BOLD5K \citep{chang2019bold5000}, with results reported in Appendix \ref{app:bold_5k}. Because the COCO images in BOLD5K overlap with the NSD shared images, these experiments use exact image correspondences, and they compare reference datasets under an identical fine-tuning protocol. Consistent with our design, NSD-trained models provide substantially stronger references, improving fine-tuning on both NSD and BOLD5K target subjects. By contrast, BOLD5K-trained references yield weaker transfer, even for BOLD5K targets. These findings support our central premise that high-SNR reference datasets induce more stable and informative shared representations, making them better suited for efficient cross-subject and cross-dataset alignment.

\begin{table}[h]
    \centering
    \caption{Quantitative results on fine-tuning THINGS-fMRI Subject 1 using NSD Subject 1 as the reference subject. Only the high level pipeline is trained and therefore, pixelwise metrics are excluded. AAMax outperforms the baseline even when an explicit data mapping is not available.}
    \footnotesize  
    \resizebox{0.7\columnwidth}{!}{  
    \begin{tabular}{lccccccc}
        \toprule
        & \multicolumn{2}{c}{\textbf{Low-Level}} & \multicolumn{4}{c}{\textbf{High-Level}} \\
        \cmidrule(lr){2-3} \cmidrule(lr){4-7}
        \textbf{Method} & Alex(2) $\uparrow$ & Alex(5) $\uparrow$ & Incep $\uparrow$ & CLIP $\uparrow$ & Eff $\downarrow$ & SwAV $\downarrow$\\
        \midrule
        
        NSD Subj1 Pre-train & 92.80\% & 96.88\% & 94.40\% & 90.02\% & 0.692 & 0.399 \\ 
        \hline
        THINGS Subj1 Individual & 68.02\% & 72.6\% & 58.8\% & 64.29\% & 0.947 & 0.649 \\ 
        \cline{2-7}
        THINGS Subj1 Baseline Fine-Tune & 77.46\% & 81.07\% & 71.06\% & 73.16\% & 0.911 & 0.619 \\  
        THINGS Subj1 AAMax Fine-Tune & \textbf{78.27\%} & \textbf{85.55\%} & \textbf{74.10\%} & \textbf{81.37\%} & \textbf{0.902} & \textbf{0.595} \\  
        \bottomrule
    \end{tabular}
    }

    \label{tab:things_comparison}
\end{table}

\section{Conclusions \& Discussion}\label{sec:conclusion}
In this work, we provide evidence that fMRI signals from different subjects can be structurally aligned within a shared representation space. Building on this, we introduce Adapter Alignment (AA), a staged training procedure for low-data adaptation of new subjects that consistently outperforms traditional end-to-end pipelines in low-data regimes while achieving comparable performance in high-data settings. We further propose a greedy image selection algorithm that identifies representative training images, together reducing the data required for fine-tuning by up to 60\% while matching baseline performance, and by 40\% while significantly outperforming it. These contributions demonstrate that accurate visual reconstructions can be achieved with substantially less data than previously thought feasible, greatly improving the efficiency of fMRI-based reconstruction.

Our experiments show that AA is robust across reference--target subject pairs, architectures, and datasets. As a few-shot method, it requires a small calibration set shared with or matched to a reference subject, its performance depends on the quality of that reference and correspondence, and our main evaluation uses a modified NSD split (Appendix~\ref{app:limit}).

By leveraging the shared representation space, AA makes fMRI-based visual reconstruction more efficient and accessible, and it can be combined with other alignment strategies for further gains. Future work could investigate richer adapter designs and extend these strategies to other brain signal modalities, advancing data-efficient brain-to-image reconstruction.


\bibliography{iclr2027_conference}
\bibliographystyle{iclr2027_conference}

\newpage
\appendix
\onecolumn

\section*{Appendix Contents}
\addcontentsline{toc}{section}{Appendix Contents}

\begin{itemize}
    \item \hyperref[app:related_work]{\textbf{A} Related Work}
    \item \hyperref[app:standard_split]{\textbf{B} Evaluation Protocol and Comparison with the Standard NSD Split}
    \item \hyperref[app:me2_diffsub]{\textbf{C} Fine-tuning on Different Subjects with MindEye2}
    \item \hyperref[app:diff_adapters]{\textbf{D} Generalizing to Different Adapter Architectures}
    \item \hyperref[app:diff_refsub]{\textbf{E} Generalizing to Different Reference Subjects}
    \item \hyperref[app:increase_evdim]{\textbf{F} Image Selection Algorithm Ablations and Baselines}
    \item \hyperref[app:train_strat]{\textbf{G} Comparing Different Alignment Training Strategies}
    \item \hyperref[app:more_limited]{\textbf{H} Extended Limited Data Experiments}
    \item \hyperref[app:epoch_ablations]{\textbf{I} Ablation Study on the two phases of Adapter Alignment Fine-Tuning}
    \item \hyperref[app:scale]{\textbf{J} Additional Scaling Analysis}
    \item \hyperref[app:common_space]{\textbf{K} Common Space and Neural Correspondence Analysis}
    \item \hyperref[app:retrieval]{\textbf{L} Pre-Diffusion Retrieval}
    \item \hyperref[app:common_concepts]{\textbf{M} Emergence of Common Concepts: Additional Details}
    \item \hyperref[app:nphardness]{\textbf{N} Image Selection Algorithm: Formulations and Proofs}
    \item \hyperref[app:tr_time]{\textbf{O} Training Time and Computational Cost}
    \item \hyperref[app:hlp_breakdown]{\textbf{P} High-Level Pipeline Training Process}
    \item \hyperref[app:llp_breakdown]{\textbf{Q} Low-Level Pipeline Training Process}
    \item \hyperref[app:vis_comp]{\textbf{R} Visual Comparison Results among Different Methods}
    \item \hyperref[app:bold_5k]{\textbf{S} Cross-Dataset Transfer and Correspondence Quality}
    \item \hyperref[app:mind_align]{\textbf{T} Combining Adapter Alignment and Image Selection with MindAligner}
    \item \hyperref[app:sum_stat]{\textbf{U} Summary Statistics}
    \item \hyperref[app:scaling]{\textbf{V} Scaling With Training Data}
    \item \hyperref[app:diff_archs]{\textbf{W} Generalizing To Different Architectures}
    \item \hyperref[app:limit]{\textbf{X} Limitations}
    \item \hyperref[app:recons]{\textbf{Y} Sample Reconstructions}
\end{itemize}

\section{Related Work}\label{app:related_work}

Recent works have achieved impressive results by mapping fMRI data to latent diffusion model (LDM) spaces ~\citep{takagi2023high,scotti2024reconstructing,lu2023minddiffuser,xia2024dream}, while simultaneously integrating multiple modalities. Despite this progress, these methods have not been thoroughly tested for their generalization performance across a larger population of subjects. In other words, constructing a model for a  new subject still requires retraining the model from scratch to ensure optimal performance. 

Approaches like \citet{ferrante2024through} have demonstrated cross-subject brain decoding using limited data samples by aligning multiple subjects. Although the authors, like us, leveraged visual stimuli commonly used for multiple subjects, the fidelity of their reconstructed images lags behind current advancements. \citet{scotti2024mindeye2}, produced significant results in image reconstruction under limited data conditions but did not fully exploit the properties of the common representation space. Adapter Alignment builds on this shared-space design and explicitly aligns each new subject to it, which enables faster convergence and improved performance in the limited-data setting. \citet{wang2024mindbridge} introduced a cross-subject brain-decoding framework that unifies fMRI dimensions through adaptive aggregation and learns subject-invariant representations via cyclic reconstruction. Its joint optimization of alignment and reconstruction resembles the second stage of our procedure; under matched conditions, a MindBridge-style joint training variant is outperformed by explicit alignment followed by end-to-end fine-tuning (Appendix~\ref{app:train_strat}). Finally, \citet{dai2025mindaligner} manages to align subjects directly in voxel space using categorical priors and many-to-many mappings, motivated by the absence of shared stimuli, whereas our strategy operates within the learned common representation space after the subject-specific adapter. Because the two methods act on different representations, they can be combined (Section~\ref{sec:train_strat}, Appendix~\ref{app:mind_align}).

\citet{qian2023fmri} refers to a 'common space'; however, their representation is primarily anatomical and only implicitly semantic.
\citet{thual2022aligning} focuses on aligning the entire cortical surfaces of different individuals in a way that preserves both anatomical structure and functional activation patterns. However many downstream tasks like mapping fMRI signals to semantic embeddings for image reconstruction benefit from representations that capture high-level semantic content. Therefore the suggested approach doesn’t necessarily capture abstract, semantic features (such as object categories, scenes, or other high-level concepts) that are crucial for decoding complex stimuli.
\citet{thual2023aligning} is an extension of \citet{thual2022aligning}. They apply functional alignment on top of the anatomical alignment. The method is very efficient even in limited data settings but it’s less integrated, meaning that it appears more as a preprocessing step. 
Piecewise shared response modeling \cite{bazeille2021empirical} provides significant performance in terms of aligning neural data from different subjects. However, the independent alignment of segments might not integrate seamlessly into a coherent global representation.
Similarly to MindEye2, but in the context of M/EEG data, \citet{li2024visual} manage to reconstruct visual stimuli using cross subject decoding. They tackle inter-subject variability primarily by incorporating subject-specific tokens into their M/EEG encoder.

\paragraph{Functional alignment in neuroscience.} Aligning functional responses across individuals using shared stimuli has a long history in cognitive neuroscience. Hyperalignment~\citep{haxby2011common} learns subject-specific transformations into a common high-dimensional model of representational space, and the shared response model~\citep{chen2015reduced} learns a reduced-dimensional shared response with subject-specific bases. Adapter Alignment operates within this tradition, but in the learned embedding space of a reconstruction pipeline: the subject-specific adapter absorbs differences in voxel count, anatomy, and cortical organization, and the alignment objective is applied to its output rather than to voxel responses. Treating each subject's fMRI pattern as a vector is therefore an architectural abstraction and not a claim that cortical geometry is irrelevant; anatomically grounded alignment methods~\citep{thual2022aligning,thual2023aligning} operate upstream of this stage and are complementary to it.

\paragraph{Positioning.} Table~\ref{tab:prior_components} summarizes where and how related methods enforce cross-subject alignment. The individual components of Adapter Alignment, namely subject-specific adapters, a reference or shared space, and alignment on shared stimuli, are established. Our contribution is their organization into a staged, architecture-agnostic procedure: an explicit alignment stage followed by end-to-end fine-tuning that retains the alignment loss, together with a selection algorithm that exploits the aligned space to choose calibration images. Appendix~\ref{app:train_strat} shows that this organization outperforms adapter-only alignment and joint end-to-end alignment under matched architecture, data, and hyperparameters.

\begin{table}[h]
    \centering
    \caption{Where and how related methods enforce cross-subject alignment.}
    \vskip 0.1in
    \setlength{\tabcolsep}{3pt}
    \renewcommand{\arraystretch}{1.15}
    \footnotesize
    \resizebox{\columnwidth}{!}{
    \begin{tabular}{llll}
        \toprule
        \textbf{Method} & \textbf{Where alignment acts} & \textbf{Alignment signal} & \textbf{How alignment is trained} \\
        \midrule
        Hyperalignment / SRM~\citep{haxby2011common,chen2015reduced} & Voxel / response space & Shared stimuli & Separate alignment model (no reconstruction) \\
        MindEye2~\citep{scotti2024mindeye2} & Shared space after linear adapters & Implicit, via reconstruction losses & Jointly end-to-end; no explicit alignment \\
        \citet{ferrante2024through} & Adapter output (reference space) & Shared stimuli & Explicit alignment only; no subsequent end-to-end training \\
        MindBridge~\citep{wang2024mindbridge} & Subject-invariant shared space & Cyclic reconstruction & Jointly with the reconstruction objective \\
        MindAligner~\citep{dai2025mindaligner} & Voxel space & Category-level pairs & Explicit voxel-space alignment module \\
        \midrule
        Adapter Alignment (ours) & Adapter output (reference space) & Shared or semantically matched calibration images & Explicit alignment, then end-to-end fine-tuning retaining the alignment loss \\
        \bottomrule
    \end{tabular}
    }
    \label{tab:prior_components}
\end{table}

\section{Evaluation Protocol and Comparison with the Standard NSD Split}\label{app:standard_split}

\paragraph{Modified split.} In the standard NSD protocol, the 1{,}000 images viewed by all subjects are reserved for testing, and each subject is trained on images unique to that subject. Adapter Alignment requires a one-to-one correspondence between the reference and the new subject during alignment, and it targets a practical setting in which a small calibration set, shared with a pre-trained reference subject, is collected for every new participant. We therefore move the shared images into the training set and replace them in the test set with an equal number of previously unique training images (Section~\ref{sec:datasets}). Each subject is thus evaluated on its own set of unseen, subject-specific images, and the test split is constructed to maintain a roughly balanced category distribution. Calibration images, reconstruction-training images, and evaluation images remain disjoint. All baselines and AAMax are retrained on exactly the same modified split, and results are averaged over three independently generated splits.

\paragraph{Calibration against the standard split.} To quantify how the modified split affects evaluation difficulty, we reran the MindEye2 baseline ourselves on the standard split under a matched architectural setup. Across the standard-split and modified-split runs, everything was kept the same, including the architecture, the training and fine-tuning procedure, and the hyperparameters, for which we used the setup recommended in the MindEye2 manuscript~\citep{scotti2024mindeye2}; only the dataset split was changed. Table~\ref{tab:standard_split} reports the results, averaged over Subjects 2, 5, and 7 fine-tuned from Subject 1 with 1 hour of data.

\begin{table}[h]
    \centering
    \caption{Effect of the dataset split on MindEye2 fine-tuning (1 hour of data, Subject 1 as the reference). The standard-split baseline was rerun with the hyperparameters recommended by MindEye2; only the split differs between the standard and modified rows. All rows are averaged over Subjects 2, 5, and 7; per-subject modified-split results are given in Appendices~\ref{app:diff_archs} and~\ref{app:me2_diffsub}.}
    \vskip 0.1in
    \setlength{\tabcolsep}{4pt}  
    \renewcommand{\arraystretch}{1.1}  
    \footnotesize  
    \begin{tabular}{llccccc}
        \toprule
        \textbf{Method} & \textbf{Split} & Alex(2) $\uparrow$ & Alex(5) $\uparrow$ & Incep $\uparrow$ & CLIP $\uparrow$ & Eff $\downarrow$ \\
        \midrule
        MindEye2 Baseline FT & Standard & \textbf{82.94\%} & \textbf{89.68\%} & 80.36\% & 78.74\% & \textbf{0.814} \\
        \midrule
        MindEye2 Baseline FT & Modified & 79.54\% & 86.46\% & 78.50\% & 75.13\% & 0.847 \\
        MindEye2 AAMax FT & Modified & 81.23\% & 88.80\% & \textbf{82.79\%} & \textbf{79.08\%} & 0.819 \\
        \bottomrule
    \end{tabular}
    \label{tab:standard_split}
\end{table}

\paragraph{Discussion.} Under identical training conditions, the baseline performs worse on every metric under the modified split than under the standard split, indicating that the subject-specific unique-image test set is a more difficult benchmark than the shared-image test set. Despite this harder protocol, AAMax recovers most of the low-level gap and exceeds the standard-split baseline on Inception and CLIP. Because all within-paper comparisons use identical modified splits, the gains of AAMax over the baseline are attributable to the training procedure rather than to the split redesign.

We note two limitations of this comparison. First, the baselines use the hyperparameter configuration recommended in the original MindEye1 and MindEye2 papers and were not independently retuned for the modified split; AAMax inherits exactly the same backbone hyperparameters. This controlled practice isolates the effect of the training procedure but does not establish that each baseline is optimally tuned for the modified protocol. Second, our goal is not a leaderboard comparison under the standard NSD benchmark, but a practical protocol for future data collection and low-data adaptation to new participants, in which a small shared calibration set can be collected by design.

\section{Fine-tuning on Different Subjects with MindEye2}\label{app:me2_diffsub}
The main text presents results on fine-tuning Subject 2 using Subject 1 as reference. Here we present the results for Subject 5 and Subject 7 too with 1 hour of training data. In both cases, AAMax outperforms the baseline on most metrics.

\begin{table}[h]
    \centering
    \caption{Fine-tuning on Subjects 5 and 7 using MindEye2's architecture. AAMax outperforms the baseline on most metrics. The same reference subject checkpoint was used in all the runs and all models were trained with the exact same hyperparameter setup.}

    \vskip 0.1in
    \setlength{\tabcolsep}{2pt}  
    \renewcommand{\arraystretch}{1.1}  
    \footnotesize  
    \resizebox{0.8\columnwidth}{!}{  
    \begin{tabular}{lccccccccc}
        \toprule
        & \multicolumn{4}{c}{\textbf{Low-Level}} & \multicolumn{4}{c}{\textbf{High-Level}} \\
        \cmidrule(lr){2-5} \cmidrule(lr){6-9}
        \textbf{Method} & PixCorr $\uparrow$ & SSIM $\uparrow$ & Alex(2) $\uparrow$ & Alex(5) $\uparrow$ & Incep $\uparrow$ & CLIP $\uparrow$ & Eff $\downarrow$ & SwAV $\downarrow$\\
        \midrule
        Subject 5 FT 1HR Baseline  & 0.102 & \textbf{0.358} &  78.40\% &  86.61\% & 79.85\% & 77.23\% & 0.839 & 0.494 \\
        Subject 5 FT 1HR AAMax  & \textbf{0.121} & 0.354 &  \textbf{80.68\%} & \textbf{88.86\%} & \textbf{85.12\%} & \textbf{81.89\%} & \textbf{0.799} & \textbf{0.469} \\
        \hline
        Subject 7 FT 1HR Baseline  & 0.103 & 0.356 & 78.50\% & 84.70\% & 76.81\% & 74.02\% & 0.858 & \textbf{0.506} \\
        Subject 7 FT 1HR AAMax  & \textbf{0.131} & \textbf{0.390} & \textbf{80.80\%} & \textbf{87.86\%} & \textbf{80.29\%} & \textbf{76.55\%} & \textbf{0.836} & 0.508 \\
        \bottomrule
    \end{tabular}
    }
    \label{tab:me2diffsub}
\end{table}

\section{Generalizing to Different Adapter Architectures}\label{app:diff_adapters}
To test the sensitivity of Adapter Alignment to the choice of adapter, we compared several lightweight variants. As shown in Table~\ref{tab:adapterablation}, the choice of adding or removing a non-linearity, or substituting different activation functions (e.g., GELU vs. ReLU), does not significantly affect alignment or reconstruction performance. In practice, simpler architectures tend to perform slightly better, with a single-layer adapter outperforming deeper variants. Our adapter choice in the main experiments was guided by the underlying backbone: we use a nonlinear (linear + GELU) adapter with MindEye1 and retain the linear adapter design of MindEye2 for consistency with the original paper. Overall, these results indicate that the specific adapter design does not affect the validity of Adapter Alignment, and performance is robust across lightweight variants.

\begin{table}[h]
    \centering
    \caption{Ablation on adapter variants for Subject 2 fine-tuned with 1 hour of data under AAMax training. We compare a nonlinear adapter with GELU, a nonlinear adapter with ReLU, and a two-layer linear adapter. Results show that adapter choice has little impact on performance, with all variants yielding similar alignment and reconstruction quality.}
    \vskip 0.1in
    \setlength{\tabcolsep}{2pt}  
    \renewcommand{\arraystretch}{1.1}  
    \footnotesize  
    \resizebox{0.8\columnwidth}{!}{  
    \begin{tabular}{lccccccccc}
        \toprule
        & \multicolumn{4}{c}{\textbf{Low-Level}} & \multicolumn{4}{c}{\textbf{High-Level}} \\
        \cmidrule(lr){2-5} \cmidrule(lr){6-9}
        \textbf{Method} & PixCorr $\uparrow$ & SSIM $\uparrow$ & Alex(2) $\uparrow$ & Alex(5) $\uparrow$ & Incep $\uparrow$ & CLIP $\uparrow$ & Eff $\downarrow$ & SwAV $\downarrow$\\
        \midrule
        Subject 2 FT 1HR AAMAX linear    & \color{gray}{0.110} & \color{gray}{0.277} &  79.70\% &  88.61\% & 80.35\% & 79.88\% & 0.815 & 0.486 \\
        Subject 2 FT 1HR AAMAX GELU      & \color{gray}{0.112} & \color{gray}{0.259} &  78.75\% &  88.04\% & 78.99\% & 78.28\% & 0.824 & 0.513 \\ 
        Subject 2 FT 1HR AAMAX ReLU      & \color{gray}{0.101} & \color{gray}{0.273} &  80.44\% &  88.18\% & 78.72\% & 78.17\% & 0.827 & 0.508 \\
        Subject 2 FT 1HR AAMAX 2 linear  & \color{gray}{0.107} & \color{gray}{0.268} &  78.28\% &  87.12\% & 78.78\% & 78.62\% & 0.827 & 0.504 \\
        \bottomrule
    \end{tabular}
    }
    \label{tab:adapterablation}
\end{table}

\section{Generalizing to Different Reference Subjects}\label{app:diff_refsub}
We conducted additional experiments (Table~\ref{tab:changeref}) to examine how the choice of reference subject influences fine-tuning performance. In these experiments, Subjects 2, 5, and 7 were used as reference subjects, and Subject 1 was fine-tuned with one hour of data. The results show that the reference subject does affect reconstruction quality: Subject 5, which typically achieves the strongest individual reconstructions, yields the best fine-tuning performance, while Subject 7 produces the weakest, consistent with its lower single-subject performance. Nonetheless, Adapter Alignment remains effective in all cases, as AAMax consistently outperforms baseline fine-tuning.

\begin{table}[h]
    \centering
    \caption{Effect of reference subject choice on limited-data fine-tuning. We compare AAMax and normal fine-tuning (160 epochs) when Subjects 2, 5, and 7 are used as reference subjects and Subject 1 is fine-tuned. AAMax consistently outperforms normal fine-tuning regardless of the reference, though overall performance reflects the quality of the chosen reference subject.}

    \vskip 0.1in
    \setlength{\tabcolsep}{2pt}  
    \renewcommand{\arraystretch}{1.1}  
    \footnotesize  
    \resizebox{\columnwidth}{!}{  
    \begin{tabular}{lccccccccc}
        \toprule
        & \multicolumn{4}{c}{\textbf{Low-Level}} & \multicolumn{4}{c}{\textbf{High-Level}} \\
        \cmidrule(lr){2-5} \cmidrule(lr){6-9}
        \textbf{Method} & PixCorr $\uparrow$ & SSIM $\uparrow$ & Alex(2) $\uparrow$ & Alex(5) $\uparrow$ & Incep $\uparrow$ & CLIP $\uparrow$ & Eff $\downarrow$ & SwAV $\downarrow$\\
        \midrule
        Subject 1 FT From Subject 2 1HR Baseline & {\color{gray}0.007} & {\color{gray}0.268} & 75.55\% & 84.00\% & 75.20\% & 74.60\% & 0.862 & 0.520 \\
        Subject 1 FT From Subject 2 1HR AAMAX & {\color{gray}0.060} & {\color{gray}0.277} & \textbf{81.13\%} & \textbf{88.99\%} & \textbf{79.71\%} & \textbf{78.89\%} & \textbf{0.824} & \textbf{0.494} \\
        \hline
        Subject 1 FT From Subject 5 1HR Baseline & {\color{gray}0.094} & {\color{gray}0.239}  & {71.54\%}  & {80.10\%}         & 73.11\%        & 74.02\% & 0.875 & 0.537 \\ 
        Subject 1 FT From Subject 5 1HR AAMAX  & {\color{gray}0.105}& {\color{gray}0.259} & \textbf{76.77\%}  & \textbf{86.20\%} & \textbf{80.08\%} & \textbf{80.14\%} & \textbf{0.828} & \textbf{0.496} \\  
        \hline
        Subject 1 FT From Subject 7 1HR Baseline & {\color{gray}0.010} & {\color{gray}0.266} & {74.96\%} & {81.55\%} & 71.71\% & 73.72\% & 0.874 & 0.528 \\  
        Subject 1 FT From Subject 7 1HR AAMAX & {\color{gray}0.012} & {\color{gray}0.283} & \textbf{79.58\%} & \textbf{87.78\%} & \textbf{77.48\%} & \textbf{77.87\%} & \textbf{0.835} & \textbf{0.502} \\
        \bottomrule
    \end{tabular}
    }
    \label{tab:changeref}
\end{table}

\paragraph{Additional reference--target combinations.} To examine whether these trends extend beyond Subject 1 as the target, we additionally fine-tuned Subjects 2, 5, and 7 with 1 hour of data from each of the other two subjects as the reference (Table~\ref{tab:changeref_all}). Together with Table~\ref{tab:changeref} and the Subject 1-reference results in Table~\ref{tab:main_limiteddata_table}, this covers every ordered reference--target pair among Subjects 1, 2, 5, and 7. We report these AAMax results for reference and to show how reconstruction quality depends on the reference--target pairing, rather than as a comparison against baseline fine-tuning. Performance varies with the quality and compatibility of the pair, but AAMax remains effective for weaker targets such as Subject 7, and no single reference is best for every target and metric.

\begin{table}[h]
    \centering
    \caption{AAMax fine-tuning (1 hour of data) for all remaining reference--target combinations among Subjects 2, 5, and 7. Results are shown for reference; performance depends on the pairing of reference and target subject. Low-level metrics are shown in gray because they are not optimized in these runs.}
    \vskip 0.1in
    \setlength{\tabcolsep}{2pt}  
    \renewcommand{\arraystretch}{1.1}  
    \footnotesize  
    \resizebox{\columnwidth}{!}{  
    \begin{tabular}{lccccccccc}
        \toprule
        & \multicolumn{4}{c}{\textbf{Low-Level}} & \multicolumn{4}{c}{\textbf{High-Level}} \\
        \cmidrule(lr){2-5} \cmidrule(lr){6-9}
        \textbf{Method} & PixCorr $\uparrow$ & SSIM $\uparrow$ & Alex(2) $\uparrow$ & Alex(5) $\uparrow$ & Incep $\uparrow$ & CLIP $\uparrow$ & Eff $\downarrow$ & SwAV $\downarrow$\\
        \midrule
        Subject 5 FT From Subject 2 1HR AAMAX & {\color{gray}0.097} & {\color{gray}0.275} & 77.31\% & 86.98\% & 81.41\% & 80.06\% & 0.829 & 0.488 \\
        Subject 7 FT From Subject 2 1HR AAMAX & {\color{gray}0.112} & {\color{gray}0.272} & 77.05\% & 84.75\% & 78.39\% & 76.95\% & 0.843 & 0.502 \\
        \hline
        Subject 2 FT From Subject 5 1HR AAMAX & {\color{gray}0.092} & {\color{gray}0.264} & 76.74\% & 86.27\% & 78.80\% & 78.24\% & 0.823 & 0.489 \\
        Subject 7 FT From Subject 5 1HR AAMAX & {\color{gray}0.093} & {\color{gray}0.252} & 76.49\% & 85.31\% & 79.21\% & 77.99\% & 0.834 & 0.504 \\
        \hline
        Subject 2 FT From Subject 7 1HR AAMAX & {\color{gray}0.098} & {\color{gray}0.264} & 75.73\% & 84.33\% & 77.69\% & 77.36\% & 0.838 & 0.494 \\
        Subject 5 FT From Subject 7 1HR AAMAX & {\color{gray}0.091} & {\color{gray}0.264} & 76.70\% & 85.75\% & 81.72\% & 79.86\% & 0.823 & 0.487 \\
        \bottomrule
    \end{tabular}
    }
    \label{tab:changeref_all}
\end{table}

\paragraph{Multiple reference subjects.} We also tested whether pre-training the reference model on more than one subject improves fine-tuning (Table~\ref{tab:multiref}). Pre-training on Subjects 1 and 2 jointly and fine-tuning Subject 5 with 1 hour of data does not consistently improve over using Subject 1 alone as the reference. This is consistent with MindEye2 (Appendix A.11 of \citet{scotti2024mindeye2}), where adding pre-training subjects yielded limited gains. We hypothesize that once a strong reference space is learned, additional subjects contribute mostly redundant rather than complementary structure. In practice, this suggests that a single high-quality reference subject suffices, while a weak reference, such as a lower-SNR dataset (Appendix~\ref{app:bold_5k}), degrades fine-tuning performance.

\begin{table}[h]
    \centering
    \caption{Single- versus multi-subject reference pre-training. Subject 5 is fine-tuned with AAMax on 1 hour of data. The single-reference row corresponds to Subject 5 FT 1HR AAMax in Table~\ref{tab:main_limiteddata_table}.}
    \vskip 0.1in
    \setlength{\tabcolsep}{2pt}  
    \renewcommand{\arraystretch}{1.1}  
    \footnotesize  
    \resizebox{\columnwidth}{!}{  
    \begin{tabular}{lccccccccc}
        \toprule
        & \multicolumn{4}{c}{\textbf{Low-Level}} & \multicolumn{4}{c}{\textbf{High-Level}} \\
        \cmidrule(lr){2-5} \cmidrule(lr){6-9}
        \textbf{Method} & PixCorr $\uparrow$ & SSIM $\uparrow$ & Alex(2) $\uparrow$ & Alex(5) $\uparrow$ & Incep $\uparrow$ & CLIP $\uparrow$ & Eff $\downarrow$ & SwAV $\downarrow$\\
        \midrule
        Pre-train Subjects 1+2 $\rightarrow$ FT Subject 5 1HR AAMAX & {\color{gray}0.102} & {\color{gray}0.263} & 78.79\% & 88.45\% & \textbf{81.08\%} & 81.41\% & \textbf{0.804} & 0.495 \\
        Pre-train Subject 1 $\rightarrow$ FT Subject 5 1HR AAMAX & {\color{gray}0.105} & {\color{gray}0.278} & \textbf{79.90\%} & \textbf{88.68\%} & 80.63\% & \textbf{82.20\%} & 0.808 & \textbf{0.493} \\
        \bottomrule
    \end{tabular}
    }
    \label{tab:multiref}
\end{table}

\section{Image Selection Algorithm Ablations and Baselines}\label{app:increase_evdim}
\subsection{Algorithm Coverage}
We next evaluate how well the image selection algorithm covers the common space compared to random sampling. Using the top 20 eigenvectors, we discretize each dimension into bins and, for a subset of $K=250$ images, count the number of empty bins (bins without a selected image). Random selection leaves, on average, around 200 empty bins with a standard deviation of about 10 across 1,000 trials. In contrast, our algorithm consistently reduces this number to 15 empty bins. A permutation test yields a p-value of 0.0002, confirming that the improved coverage is statistically significant. This demonstrates that our algorithm selects images that span the representation space more uniformly along high-variance directions. It is important to note that the number of empty bins in our algorithm remain relatively similar when changing the number of top eigenvectors.

\begin{figure}[t]
  \centering
  \includegraphics[width=0.9\columnwidth]{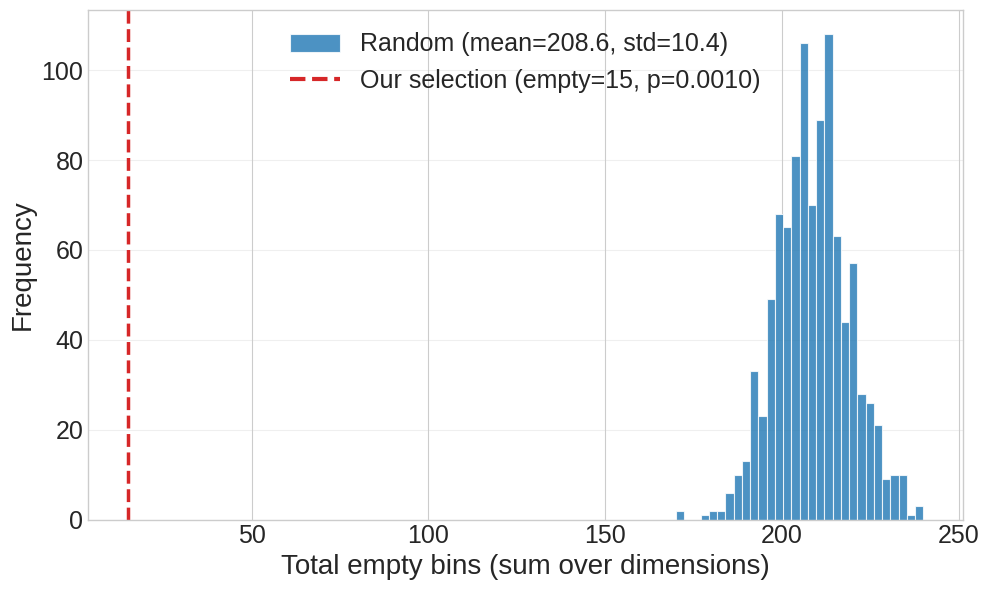}
  \caption{\textbf{Coverage of the image selection algorithm compared to random sampling.}
  We project embeddings onto the top 20 eigenvectors and count empty bins across 250 selected images. 
  The histogram shows empty-bin counts for 1{,}000 random subsets, while the dashed line marks our algorithm’s selection. 
  Our method leaves only 15 bins empty, whereas random sampling leaves $\approx 200 \pm 10$ on average, indicating significantly more uniform coverage of the representation space ($p=0.0002$).}
  \label{fig:is_coverage}
\end{figure}

\subsection{Varying the Eigenvector Dimension}
We conducted an ablation study on the image selection algorithm to evaluate the effect of varying the number of principal eigenvector dimensions used in the binning process. Specifically, we compared subsets selected using the top 10, 20, 30, and 40 eigenvectors. As shown in Table~\ref{tab:isevablation}, performance is sensitive to this choice. For Subject 2, 20 and 40 eigenvectors are competitive and clearly outperform 10 and 30 dimensions, with 40 eigenvectors achieving the best score on most metrics. More generally, we observed that using a larger number of eigenvector dimensions tended to yield better fine-tuning performance. We therefore adopt 40 eigenvectors for the image-selection experiments in Tables~\ref{tab:main_limiteddata_table} and~\ref{tab:mindalign_results}; the reduced-data experiments in Figure~\ref{fig:low_data_push} were run with 20 eigenvectors.

\begin{table}[h]
    \centering
    \caption{Ablation on the number of eigenvector dimensions used for image selection. We compare random image selection and Adapter Alignment with images selected using our algorithm across different eigenvector settings. Results are reported for Subject 1 as the reference and Subject 2 fine-tuned with 250 images under AAMax training, averaged over 3 runs. 20 and 40 eigenvectors are competitive, and 40 eigenvectors are used for the image-selection results in Tables~\ref{tab:main_limiteddata_table} and~\ref{tab:mindalign_results}.}

    \vskip 0.1in
    \setlength{\tabcolsep}{2pt}  
    \renewcommand{\arraystretch}{1.1}  
    \footnotesize  
    \resizebox{0.8\columnwidth}{!}{  
    \begin{tabular}{lccccccccc}
        \toprule
        & \multicolumn{4}{c}{\textbf{Low-Level}} & \multicolumn{4}{c}{\textbf{High-Level}} \\
        \cmidrule(lr){2-5} \cmidrule(lr){6-9}
        \textbf{Method} & PixCorr $\uparrow$ & SSIM $\uparrow$ & Alex(2) $\uparrow$ & Alex(5) $\uparrow$ & Incep $\uparrow$ & CLIP $\uparrow$ & Eff $\downarrow$ & SwAV $\downarrow$\\
        \midrule
        Random Selection & \color{gray}{0.112} & \color{gray}{0.259} &  78.75\% &  88.04\% & 78.99\% & 78.28\% & 0.824 & 0.513 \\
        10 EV Dimensions & \color{gray}{0.114} & \color{gray}{0.303} &  78.99\% &  86.53\% & 77.30\% & 75.23\% & 0.850 & 0.541 \\
        20 EV Dimensions & \color{gray}{0.102} & \color{gray}{0.267} &  81.41\% &  89.55\% & 80.43\% & 80.33\% & 0.823 & 0.496 \\
        30 EV Dimensions & \color{gray}{0.112} & \color{gray}{0.275} &  79.60\% &  88.03\% & 78.58\% & 76.89\% & 0.850 & 0.537 \\
        40 EV Dimensions & \color{gray}{0.118} & \color{gray}{0.285} &  81.78\% &  89.61\% & 81.49\% & 79.55\% & 0.810 & 0.491 \\
        
        \bottomrule
    \end{tabular}
    }
    \label{tab:isevablation}
\end{table}

\subsection{Comparison with Standard Selection Baselines}\label{app:is_baselines}
Our greedy coverage objective resembles coreset and diversity-based selection over a learned representation. We therefore compare it against four standard selection strategies under the same image budget:
\begin{itemize}
    \item \textbf{K-Center Greedy:} cosine farthest-first traversal over the normalized 4{,}096-dimensional Subject 1 adapter embeddings. At every step it selects the image farthest from its nearest already-selected image, minimizing the largest uncovered radius (minimax coverage) in the reference representation.
    \item \textbf{K-Means:} K-means with 250 clusters in the same normalized Subject 1 adapter space, selecting the real image nearest to each centroid. This favors representative prototypes and density coverage rather than boundary points.
    \item \textbf{DPP Diversity:} greedy maximization of the determinant of a linear kernel over the same normalized Subject 1 adapter embeddings, favoring a jointly diverse and non-redundant subset with large volume.
    \item \textbf{CLIP K-Center:} the same cosine farthest-first rule applied to 768-dimensional CLIP ViT-L/14 image embeddings, targeting visual--semantic diversity in pre-trained image space independently of the fMRI representation.
\end{itemize}
All methods select 250 images (1 hour of data) from the same pool of 1{,}000 shared images. We use the same Subject 1 reference model, Subject 2 fine-tuning budget, MindEye1 architecture with AAMax training, random seed, checkpoint epoch, test set, reconstruction settings, and evaluation pipeline, so only the image-selection rule changes. Our method uses 40 eigenvector dimensions.

\begin{table}[h]
    \centering
    \caption{Comparison of image-selection strategies. Subject 1 is the reference and Subject 2 is fine-tuned with AAMax on 250 selected images; only the selection rule differs between rows. Best results are in bold.}
    \vskip 0.1in
    \setlength{\tabcolsep}{2pt}  
    \renewcommand{\arraystretch}{1.1}  
    \footnotesize  
    \resizebox{0.8\columnwidth}{!}{  
    \begin{tabular}{lcccccccc}
        \toprule
        & \multicolumn{4}{c}{\textbf{Low-Level}} & \multicolumn{4}{c}{\textbf{High-Level}} \\
        \cmidrule(lr){2-5} \cmidrule(lr){6-9}
        \textbf{Method} & PixCorr $\uparrow$ & SSIM $\uparrow$ & Alex(2) $\uparrow$ & Alex(5) $\uparrow$ & Incep $\uparrow$ & CLIP $\uparrow$ & Eff $\downarrow$ & SwAV $\downarrow$\\
        \midrule
        Ours & \textbf{0.118} & 0.285 & 81.78\% & \textbf{89.61\%} & \textbf{81.49\%} & 79.55\% & \textbf{0.810} & 0.491 \\
        K-Means & 0.110 & 0.281 & \textbf{82.40\%} & 89.14\% & 80.51\% & \textbf{80.26\%} & 0.815 & 0.490 \\
        K-Center Greedy & 0.114 & 0.271 & 81.93\% & 89.54\% & 80.10\% & 79.35\% & 0.817 & \textbf{0.489} \\
        DPP Diversity & 0.115 & 0.270 & 80.80\% & 89.22\% & 80.21\% & 78.66\% & 0.820 & 0.493 \\
        CLIP K-Center & 0.112 & \textbf{0.290} & 81.15\% & 88.73\% & 78.56\% & 79.92\% & 0.831 & 0.499 \\
        \bottomrule
    \end{tabular}
    }
    \label{tab:is_baselines}
\end{table}

As shown in Table~\ref{tab:is_baselines}, the baselines are competitive, and no single method dominates across all metrics. Our algorithm achieves the best PixCorr, AlexNet(5), Inception, and EfficientNet scores (4 of 8 metrics), K-Means is best on AlexNet(2) and CLIP, K-Center Greedy is best on SwAV, and CLIP K-Center is best on SSIM. This pattern indicates that the outcome is not explained by a single generic notion of diversity: both the selection geometry and the representation in which diversity is measured affect different aspects of reconstruction quality. The selectors computed in the aligned reference space (ours, K-Means, K-Center Greedy, DPP) perform best on the high-level metrics, whereas CLIP K-Center, which ignores the fMRI representation, performs worst on Inception, EfficientNet, and SwAV. Our greedy algorithm is thus one effective instantiation of coverage-based selection in the aligned space, with the approximation guarantees of Appendix~\ref{app:nphardness}, rather than the only viable choice.

\subsection{Overlap Between Selected Subsets}
To verify that these methods do not recover the same calibration subset, we computed the pairwise overlap between their selected images (Table~\ref{tab:is_overlap}). Most cross-method overlaps range from 17.2\% to 30.8\%. The larger overlap between DPP and K-Center Greedy (52.4\%) is expected, since both operate on the same Subject 1 adapter representation, although they optimize different subset objectives. Only one image is selected by all five methods. The baselines therefore explore materially different calibration subsets rather than minor variants of the same selection.

\begin{table}[h]
    \centering
    \caption{Pairwise overlap between the 250-image subsets selected by each method. Entries are the number of shared images, with percentages in parentheses.}
    \vskip 0.1in
    \setlength{\tabcolsep}{4pt}
    \renewcommand{\arraystretch}{1.1}
    \footnotesize
    \begin{tabular}{lccccc}
        \toprule
        & Ours & K-Center & K-Means & DPP & CLIP K-Center \\
        \midrule
        Ours & 250 (100.0\%) & 73 (29.2\%) & 74 (29.6\%) & 65 (26.0\%) & 58 (23.2\%) \\
        K-Center & 73 (29.2\%) & 250 (100.0\%) & 43 (17.2\%) & 131 (52.4\%) & 77 (30.8\%) \\
        K-Means & 74 (29.6\%) & 43 (17.2\%) & 250 (100.0\%) & 43 (17.2\%) & 60 (24.0\%) \\
        DPP & 65 (26.0\%) & 131 (52.4\%) & 43 (17.2\%) & 250 (100.0\%) & 64 (25.6\%) \\
        CLIP K-Center & 58 (23.2\%) & 77 (30.8\%) & 60 (24.0\%) & 64 (25.6\%) & 250 (100.0\%) \\
        \bottomrule
    \end{tabular}
    \label{tab:is_overlap}
\end{table}

\section{Comparing Different Alignment Training Strategies}\label{app:train_strat}
To directly compare our approach against related methods, we conducted an experiment fine-tuning Subject 2 with 250 training images under three different training regimes (Table~\ref{tab:train_strat}). First, we evaluated AA Step 1 only, which corresponds to the setting of Ferrante et al.~\citep{ferrante2024through}, where adapters are overfitted to the common space without any subsequent end-to-end training. Second, we implemented a variant inspired by MindBridge~\citep{wang2024mindbridge}, where the adapter remains unfrozen while the MLP and prior are frozen, and the full end-to-end training is carried out using a combined MSE, CLIP, and prior loss. Finally, we compared these baselines to our full AAMax training regime, which explicitly aligns adapters in the first phase and then fine-tunes the entire pipeline end-to-end.

The results show that AAMax achieves the strongest performance across both low-level and high-level metrics, outperforming both AA Step 1 and the MindBridge-style end-to-end setup. Importantly, this experiment was run under identical architectural settings and without auxiliary enhancements such as additional synthetic data or architectural add-ons. This isolates the effect of the training paradigm itself, demonstrating that when compared purely on training objectives, AAMax provides the most effective strategy for fine-tuning new subjects.

\begin{table}[h]
    \centering
    \caption{Comparison of different training regimes for fine-tuning Subject 2 with 250 training images. We evaluate (1) AAMax (full Adapter Alignment with staged alignment + end-to-end fine-tuning), (2) AA Step 1 only (adapter overfitting without subsequent end-to-end training), and (3) End-to-End with MSE+Prior+CLIP Loss (MLP Frozen) (reset-tuning). AAMax achieves the best performance across both low-level and high-level metrics, demonstrating that explicit staged alignment followed by end-to-end fine-tuning is more effective than either alignment-only or combined end-to-end training when compared under identical architectural conditions.}

    \vskip 0.1in
    \setlength{\tabcolsep}{2pt}  
    \renewcommand{\arraystretch}{1.1}  
    \small  
    \resizebox{0.9\columnwidth}{!}{  
    \begin{tabular}{lccccccccc}
        \toprule
        & \multicolumn{4}{c}{\textbf{Low-Level}} & \multicolumn{4}{c}{\textbf{High-Level}} \\
        \cmidrule(lr){2-5} \cmidrule(lr){6-9}
        \textbf{Method} & PixCorr $\uparrow$ & SSIM $\uparrow$ & Alex(2) $\uparrow$ & Alex(5) $\uparrow$ & Incep $\uparrow$ & CLIP $\uparrow$ & Eff $\downarrow$ & SwAV $\downarrow$\\
        \midrule
        AAMAX      & \color{gray}{0.112} & \color{gray}{0.259} &  \textbf{78.75\%} &  \textbf{88.04\%} & \textbf{78.99\%} & \textbf{78.28\%} & \textbf{0.824} & 0.513 \\
        AA Step 1 only & \color{gray}{0.008} & \color{gray}{0.240} &  77.14\% &  87.04\% & 77.00\% & 75.43\% & 0.857 & 0.510 \\ 
        End-to-End with MSE+Prior+CLIP Loss (MLP Frozen) & \color{gray}{0.007} & \color{gray}{0.238} &  76.40\% &  86.42\% & 77.72\% & 75.04\% & 0.852 & \textbf{0.507} \\
        \bottomrule
    \end{tabular}
    }
    \label{tab:train_strat}
\end{table}

\section{Extended Limited Data Experiments}\label{app:more_limited}
To further validate the efficiency of Adapter Alignment, we repeated our limited-data experiments using Subject 5 as the reference and Subject 1 for fine-tuning. We fine-tune with 1, 2, and 4 hours of data and compare baseline end-to-end training to AAMax. As always, AAMax first performs Step 1 adapter alignment by overfitting the adapter, followed by end-to-end fine-tuning. We track performance after 1 epoch of fine-tuning and after 160 epochs of fine-tuning. As shown in Table \ref{tab:app_limitedcomparison}, even after a single epoch, AAMax consistently outperforms baseline models trained for 160–200 epochs.

We also track performance over fine-tuning epochs for all three data conditions (1, 2, and 4 hours). As shown in Figure~\ref{fig:limited_data_convergence}, Adapter Alignment converges rapidly and maintains a consistent advantage throughout training, while baseline training requires hundreds of epochs to reach comparable accuracy. These results reinforce that AAMax substantially reduces both the data and computational requirements for fine-tuning new subjects.

\begin{table}[h]
    \centering
    \caption{Quantitative results on fine-tuning a new subject with limited data. FT=Fine tuning. The first section shows the result of fine-tuning with 250 (1 hour) images after 1 epoch and after 160 epochs on the baseline and using AAMax. The second section compares performance on 500 (2 hours) images and the third section compares performance on 1000 (4 hours) images. In all cases AAMax, significantly outperforms the baseline. The greyed out low level metrics are not optimized for in these experiments.}

    \vskip 0.1in
    \setlength{\tabcolsep}{2pt}  
    \renewcommand{\arraystretch}{1.1}  
    \footnotesize  
    \resizebox{0.8\columnwidth}{!}{  
    \begin{tabular}{lccccccccc}
        \toprule
        & \multicolumn{4}{c}{\textbf{Low-Level}} & \multicolumn{4}{c}{\textbf{High-Level}} \\
        \cmidrule(lr){2-5} \cmidrule(lr){6-9}
        \textbf{Method} & PixCorr $\uparrow$ & SSIM $\uparrow$ & Alex(2) $\uparrow$ & Alex(5) $\uparrow$ & Incep $\uparrow$ & CLIP $\uparrow$ & Eff $\downarrow$ & SwAV $\downarrow$\\
        \midrule
        Subject 1 FT 1HR Baseline E1       & {\color{gray}0.031}           & {\color{gray}0.234}       & 58.31\%         & 62.56\%         & 58.51\%       & 58.71\% & 0.947 & 0.611 \\          
        Subject 1 FT 1HR AAMAX E1   & {\color{gray}0.095}           & {\color{gray}0.254}       & \textbf{77.20\%} & 85.98\%         & 78.80\%       & 79.39\% & 0.833 & 0.503 \\  
        Subject 1 FT 1HR Baseline E160     & {\color{gray}0.094}           & {\color{gray}0.239}       & 71.54\%         & 80.10\%         & 73.11\%        & 74.02\% & 0.875 & 0.537 \\ 
        Subject 1 FT 1HR AAMAX E160    & \textbf{{\color{gray}0.105}}   & \textbf{{\color{gray}0.259}} & 76.77\%       & \textbf{86.20\%} & \textbf{80.08\%} & \textbf{80.14\%} & \textbf{0.828} & \textbf{0.496} \\  
        \hline
        Subject 1 FT 2HR Baseline E1       & {\color{gray}0.049} & {\color{gray}0.242} & 62.41\% & 68.06\% & 62.64\% & 64.80\% & 0.925 & 0.583 \\  
        Subject 1 FT 2HR AAMAX E1  & {\color{gray}0.112} & {\color{gray}0.275} & 79.46\% & 88.29\% & \textbf{83.28\%} & 82.61\% & 0.803 & 0.485 \\  
        Subject 1 FT 2HR Baseline E160     & {\color{gray}0.110} & {\color{gray}0.270} & 75.32\% & 84.62\% & 76.66\% & 78.55\% & 0.845 & 0.516 \\  
        Subject 1 FT 2HR AAMAX E160  & \textbf{{\color{gray}0.114}} & \textbf{{\color{gray}0.285}} & \textbf{80.00\%} & \textbf{88.50\%} & 83.19\% & \textbf{83.05\%} & \textbf{0.799} & \textbf{0.480} \\  
        \hline
        Subject 1 FT 4HR Baseline E1   & {\color{gray}0.079} & {\color{gray}0.241} & 65.84\% & 73.04\% & 64.19\% & 68.10\% & 0.914 & 0.574 \\  
        Subject 1 FT 4HR AAMAX E1  & \textbf{{\color{gray}0.129}} & {\color{gray}0.265} & 80.83\% & 90.93\% & \textbf{86.66\%} & 86.17\% & 0.767 & 0.460 \\  
        Subject 1 FT 4HR Baseline E160  & {\color{gray}0.124} & {\color{gray}0.253} & 77.86\% & 89.18\% & 83.10\% & 83.26\% & 0.801 & 0.480 \\       
        Subject 1 FT 4HR AAMAX E160  & {\color{gray}0.127} & \textbf{{\color{gray}0.273}} & \textbf{81.73\%} & \textbf{91.29\%} & 86.23\% & \textbf{86.26\%} & \textbf{0.763} & \textbf{0.453} \\  
        \bottomrule
    \end{tabular}
    }
    \label{tab:app_limitedcomparison}
\end{table}

\begin{figure*}[h]
    \centering
    \begin{subfigure}{0.32\textwidth}
        \centering
        \includegraphics[width=\columnwidth]{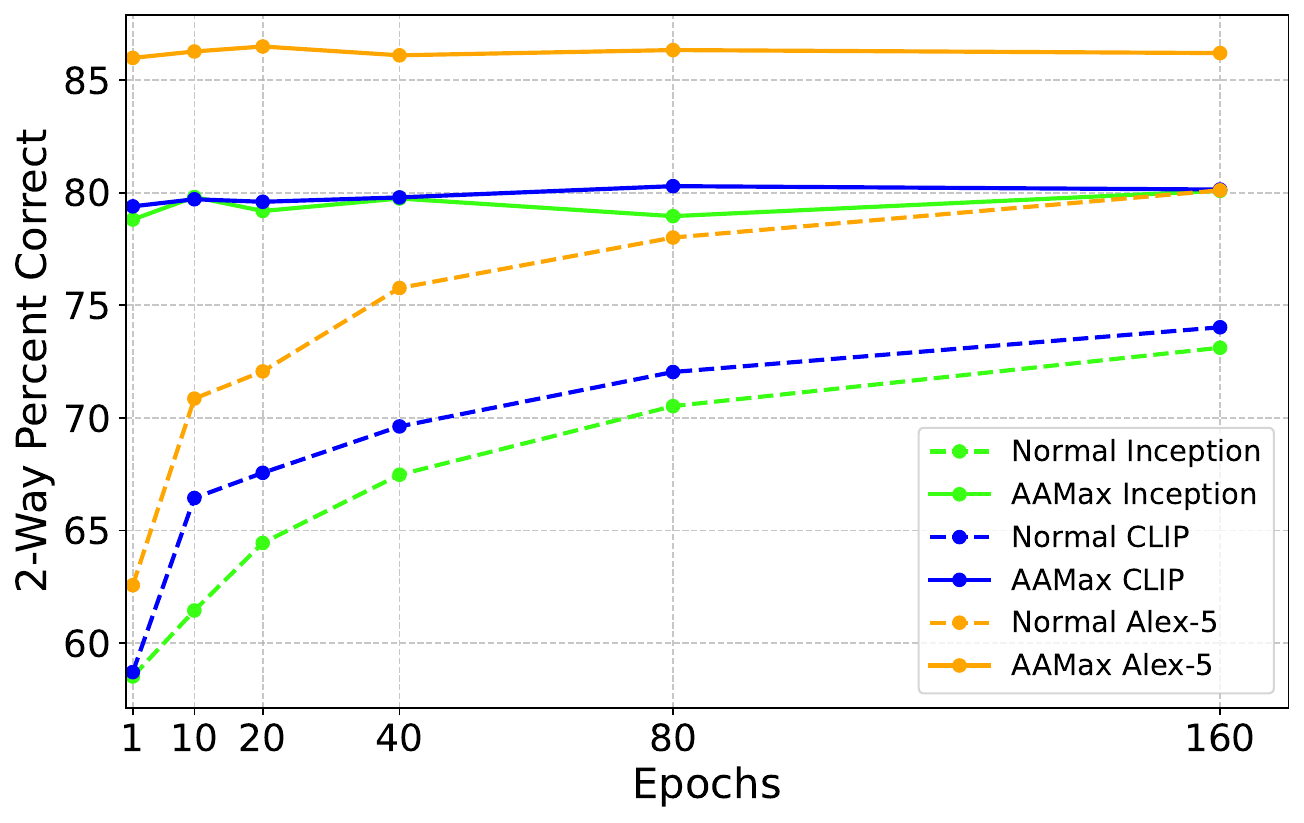}
    \end{subfigure}\hfill
    \begin{subfigure}{0.32\textwidth}
        \centering
        \includegraphics[width=\columnwidth]{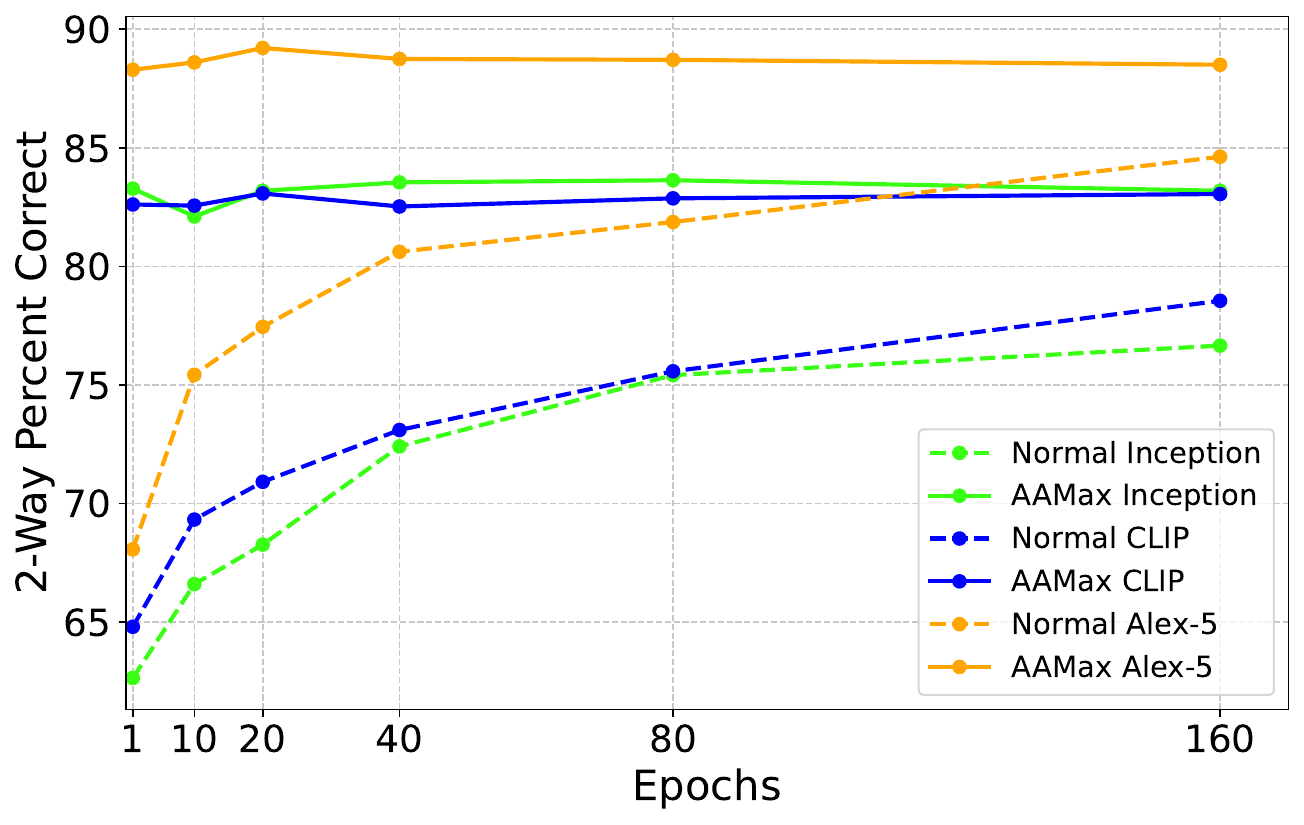}
    \end{subfigure}\hfill
    \begin{subfigure}{0.32\textwidth}
        \centering
        \includegraphics[width=\columnwidth]{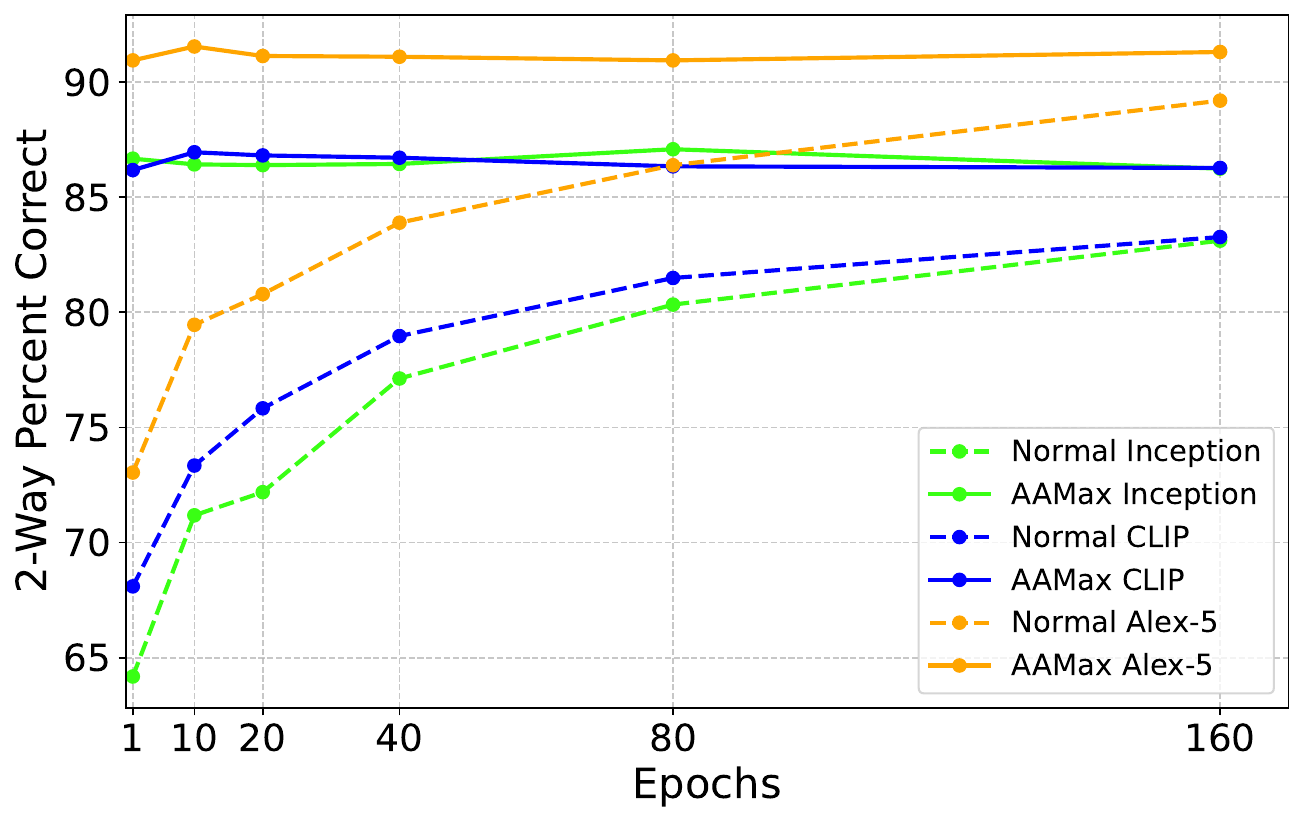}
    \end{subfigure}
    \caption{Comparing finetuning reconstruction performance over epochs with limited data. From left to right: (a) Finetuning Subject 1 with 250 images or 1 hour of data (b) Finetuning Subject 1 with 500 images or 2 hours of data (c) Finetuning Subject 1 with 1000 images or 4 hours of data. AAMax not only outperforms normal end-to-end training, but the performance at Epoch 1 is better than the baseline's performance at Epoch 160.}
    \label{fig:limited_data_convergence}
\end{figure*}

\section{Ablation Study on the Two Phases of Adapter Alignment Fine-Tuning}\label{app:epoch_ablations}
We analyze the role of Step 1 (adapter alignment) relative to full end-to-end fine-tuning. In this experiment, we fine-tuned Subject 2 with 1 hour of data using Subject 1 as the reference. Figure~\ref{fig:aa_vs_e2e} illustrates the reconstruction performance across epochs for both phases.

During Step 1, where only the adapter is optimized, performance improves steadily as the adapter overfits to the shared images. This phase alone provides a substantial boost over the unaligned baseline, and after sufficient epochs the adapter can achieve strong alignment across subjects. However, the improvement plateaus, as Step 1 only modifies the adapter and does not allow deeper components of the pipeline to adapt.

Introducing end-to-end fine-tuning after adapter alignment leads to an immediate jump in performance. Even after a single epoch of end-to-end updates, the model surpasses what hundreds of adapter-only epochs can achieve. We stop at around 200 epochs as the relative gain plateaus again. While the improvement in performance by a few percent from end-to-end fine-tuning might not seem significant when compared to the gain from Step 1, it can make the difference between coherent images and gibberish when working with limited data and when reconstruction fidelity is low.

\begin{figure}[t]
  \centering
  \includegraphics[width=0.9\columnwidth]{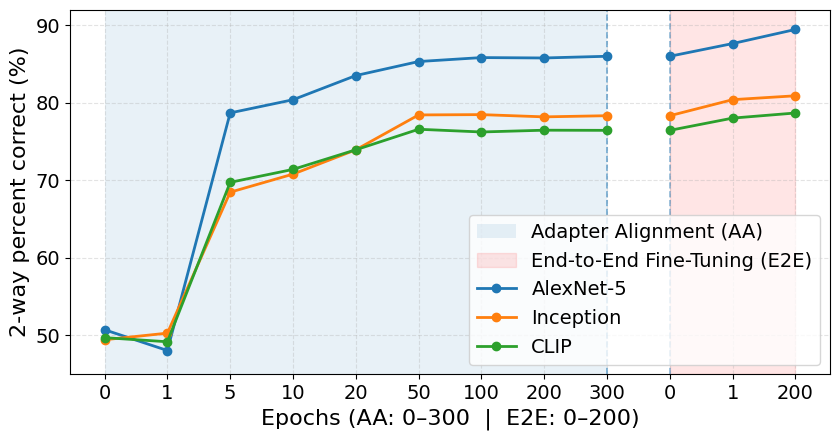}
  \caption{Reconstruction performance across epochs for Subject 2 fine-tuned with 1 hour of data using Subject 1 as the reference. The shaded regions denote Step 1 (adapter alignment) and Step 2 (end-to-end fine-tuning). Performance steadily increases during adapter alignment but plateaus short of optimal results. Transitioning to end-to-end fine-tuning yields an immediate jump in accuracy and ultimately outperforms adapter-only training. The dashed lines indicate the transition from AA to E2E. Performance at Epoch 300 of AA is the same as Epoch 0 of E2E, since no further fine-tuning has been done yet.}
  \label{fig:aa_vs_e2e}
\end{figure}

\section{Additional Scaling Analysis}\label{app:scale}

 The architectural setup used by MindEye2 leverages a 4096-dimensional shared space with BigCLIP and Stable Diffusion XL (SDXL), requiring substantially greater computational resources than our main experiments, which employ Versatile Diffusion. To evaluate whether the benefits of Adapter Alignment (AAMax) persist across different model capacities, we conducted a scaling law analysis varying both the dimensionality of the shared space and the choice of generative backbone.
 

\subsection{Empirical Observations}
\begin{enumerate}
    \item \textbf{Architectural Improvements at 1024 Dimensions:} Using a 1024-dimensional shared space, we find that BigCLIP + SDXL outperforms Versatile Diffusion, confirming that a higher-capacity CLIP space and diffusion model contribute to improved reconstruction performance.
    \item \textbf{Scaling Behavior with Versatile Diffusion:} We vary the shared space dimensionality from 512 → 1024 → 4096, using adapter alignment (AAMax) fine-tuning and observe consistent improvements in reconstruction quality as dimensionality increases. 
\end{enumerate}

These findings indicate that AAMax is robust to architectural choice and continues to provide benefits regardless of the CLIP encoder or diffusion model employed.

\subsection{Projection to 4096 Dimensions with Stable Diffusion XL}
The key insight from these experiments is that \textbf{adapter alignment is a training paradigm, not an architectural modification}. Our results demonstrate that AA is agnostic to the representational scale and generative backbone, and scales effectively with larger models. While large-scale pipelines like MindEye2 set upper bounds on reconstruction quality under high compute budgets, adapter alignment provides a complementary framework that improves efficiency and accessibility. Importantly, its benefits persist even when extended to high-capacity architectures, making it a broadly applicable strategy for fMRI-based reconstruction.

\begin{figure}[h]
    \centering
    \begin{subfigure}{0.4\textwidth}
        \centering
        \includegraphics[width=\textwidth]{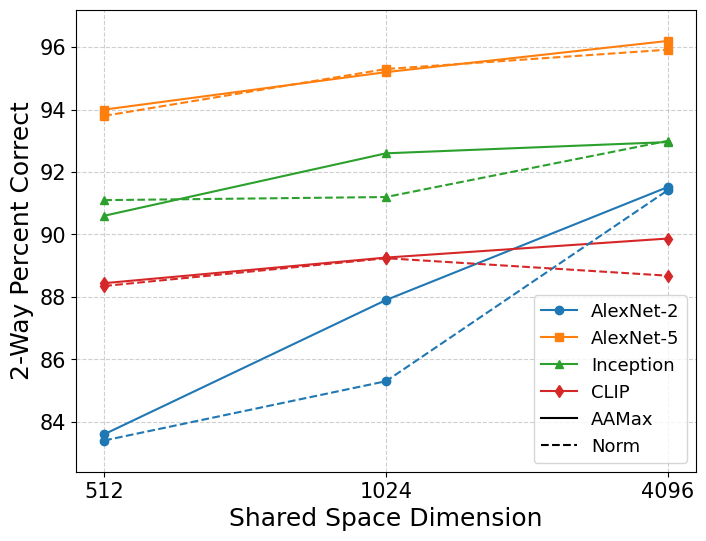}
    \end{subfigure}
    \begin{subfigure}{0.4\textwidth}
        \centering
        \includegraphics[width=\textwidth]{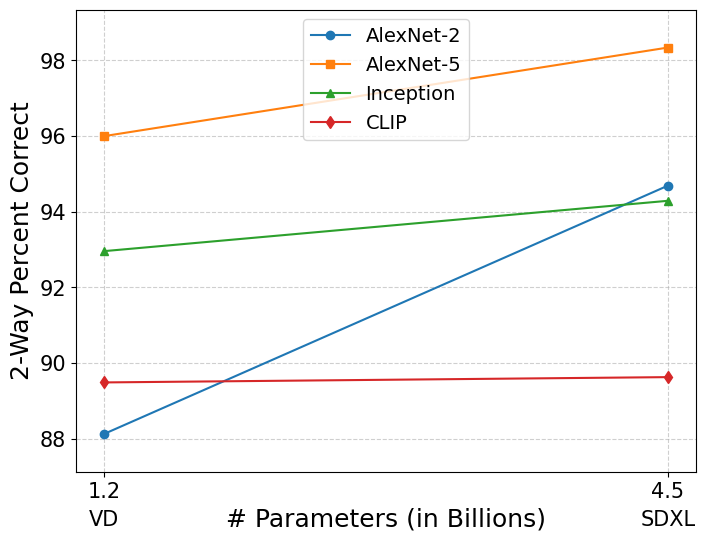}
    \end{subfigure}\hfill
    \caption{Comparing performance with scale for AAMax. \textbf{Left:} AAMax vs regular finetuning performance over increasing shared space size. AAMax scales with shared space dimensionality. \textbf{Right:} AAMax finetuning performance at 1024 shared space with Versatile Diffusion and Stable Diffusion XL. All metrics improve when moving to a larger diffusion model.}
    \label{fig:ablation}
\end{figure}

\section{Common Space and Neural Correspondence Analysis}\label{app:common_space}
We run several empirical tests to validate the effectiveness of Adapter Alignment in aligning the common space across subjects. We distinguish two kinds of evidence. Analyses on the shared images that are used for alignment (Sections~\ref{sec:mse} and~\ref{sec:knn}) verify that the alignment objective is achieved; because Step 1 explicitly fits the new subject's adapter to the reference embeddings of these images, high agreement on them is expected by construction. Analyses on images that are not used for alignment, namely each subject's held-out test images (Section~\ref{sec:eigen}) and shared images excluded from the calibration set (Section~\ref{app:neural_corr}), test whether the induced cross-subject structure generalizes beyond the calibration pairs. These analyses provide empirical evidence of transferable functional representational structure; they do not by themselves establish an intrinsic shared neural space or anatomical correspondence. 

\subsection{Eigenvector Analysis}\label{sec:eigen}
To first investigate the structural alignment of embeddings across subjects, we analyzed the subject-wise cosine similarity of the principal eigenvectors in the embedding spaces produced by two techniques: the standard end-to-end training (Baseline) and our proposed approach (AAMax). For both techniques, we extracted the top five principal eigenvectors from the test data for each subject and computed the cosine similarity between the eigenvectors from different subjects.
\begin{figure}[h]
    \centering
    \begin{subfigure}{\textwidth}
        \centering
        \includegraphics[width=\textwidth]{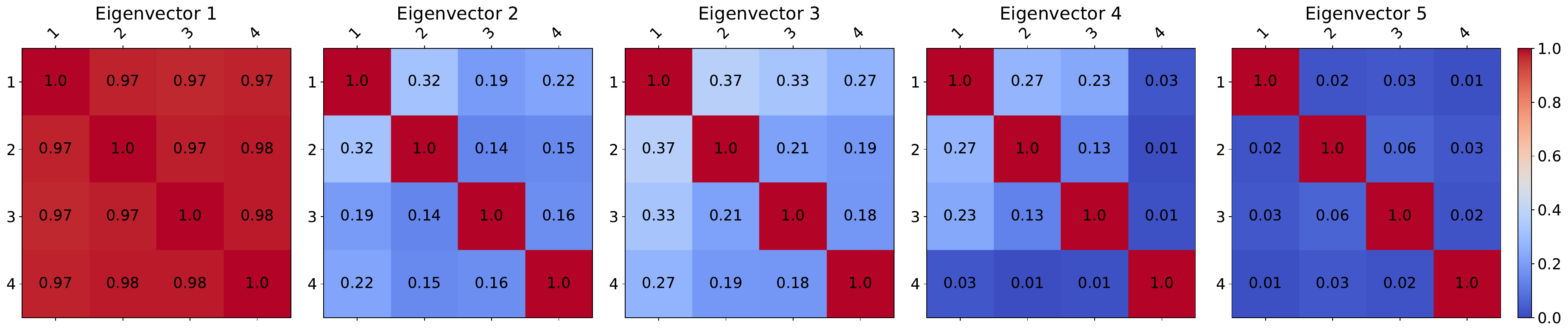}
    \end{subfigure}\hfill
    \begin{subfigure}{\textwidth}
        \centering
        \includegraphics[width=\textwidth]{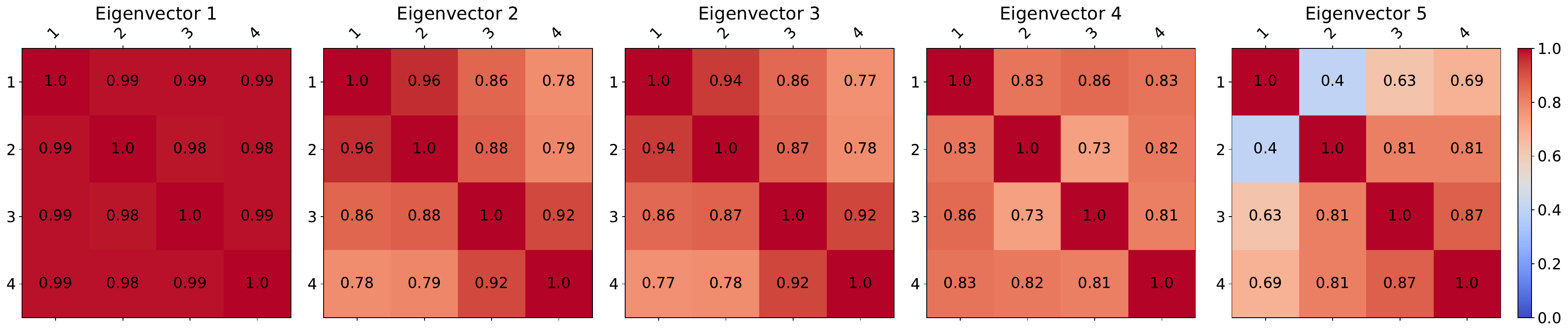}
    \end{subfigure}\hfill
    \caption{Comparing eigenvector cosine similarity for the test data across subjects in both training techniques. We use the first 5 principal eigenvectors. Row 1 shows the results for normal end-to-end training. Except for the first eigenvector, cosine similarity is very low across subjects. Row 2 shows the results for AAMax training. Cosine similarity is high even at the 5th principal eigenvector.}
    \label{fig:eigenspace_common_aamax}
\end{figure}
In normal end-to-end training, we observed consistently low cosine similarity between all but the first principal eigenvector of different subjects, suggesting that the primary directions of variance captured by the embeddings were highly subject-specific and lacked alignment across subjects. This indicates that the embedding space learned in this approach does not generalize well to a common structure, resulting in subject-specific representations that diverge significantly in terms of principal variance.

\begin{figure}[h]
    \centering
    \begin{subfigure}{0.7\textwidth}
        \centering
        \includegraphics[width=\textwidth]{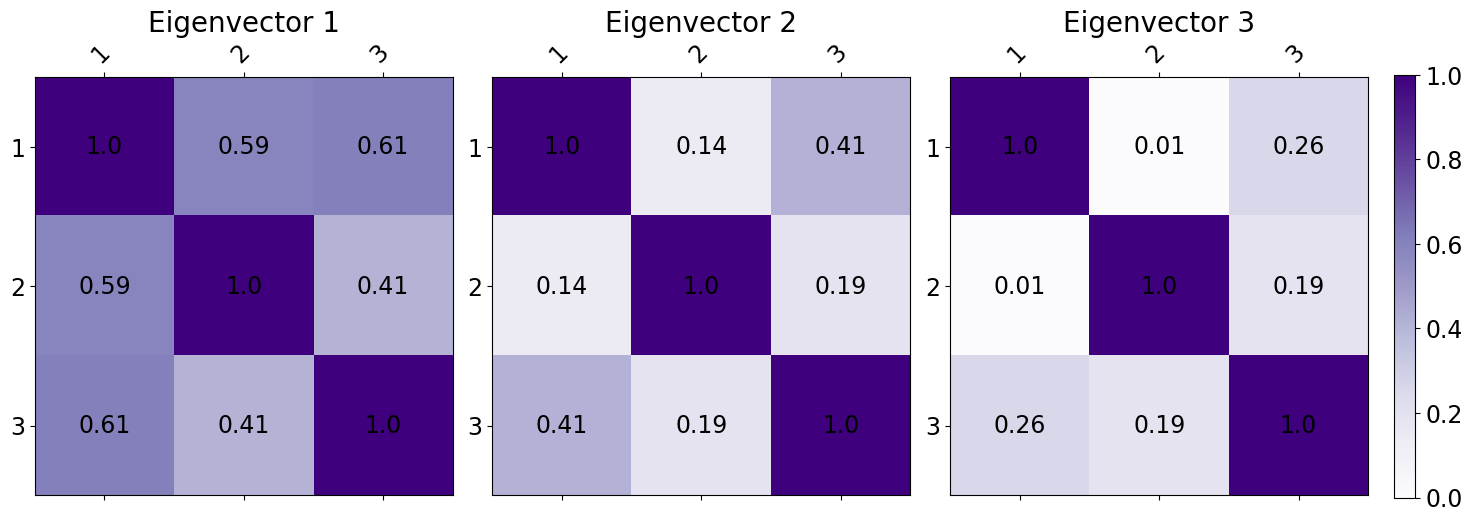}
    \end{subfigure}\hfill
    \caption{Comparing eigenvector cosine similarity for the test data across datasets in both training techniques. 1. Reference Subject 1 NSD 2. Normal Finetune Subject 1 THINGS 3. AAMax Finetune Subject 1 THINGS. AAMax training generates better aligned embedding spaces.}
    \label{fig:eigenspace_things_common_aamax}
\end{figure}

In contrast, AAMax training showed significantly higher cosine similarity between the principal eigenvectors of different subjects. This indicates that AAMax effectively aligns the embeddings across subjects in terms of the primary directions of variance, achieving a shared representation space that better captures common structural features. The high similarity across subjects indicates that AAMax encodes the principal semantic and structural patterns consistently. Because these eigenvectors are computed on each subject's held-out test images, which are subject-specific, unseen during alignment, and have no one-to-one correspondence across subjects, this shows that the aligned geometry generalizes beyond the image pairs on which alignment is explicitly enforced.

We also perform a similar analysis across datasets in Figure \ref{fig:eigenspace_things_common_aamax}. We compare the alignment between Subject 1 from NSD and Subject 1 from THINGS when the latter has been finetuned on a model that was pre-trained with the former. We present the difference between both training techniques. AAMax training generates better aligned embedding spaces.

\subsection{Mean Squared Error}\label{sec:mse}

We compare the end-to-end training approach with AAMax by evaluating the MSE between subjects using shared images after training as shown in Figure \ref{fig:mse_common_aamax}. In the end-to-end model, we observe relatively high MSE between subjects, indicating poor alignment in the common space. In contrast, the MSE for AAMax is an order of magnitude lower, demonstrating a much stronger alignment of subject-specific embeddings. Because these shared images include the calibration images on which the adapter-level MSE is optimized, this result primarily confirms that the alignment objective is achieved. Section~\ref{app:neural_corr} shows that AAMax also reduces cross-subject MSE on shared images excluded from the calibration set.

\begin{figure}[h]
    \centering
    \begin{subfigure}{0.5\textwidth}
        \centering
        \includegraphics[width=\textwidth]{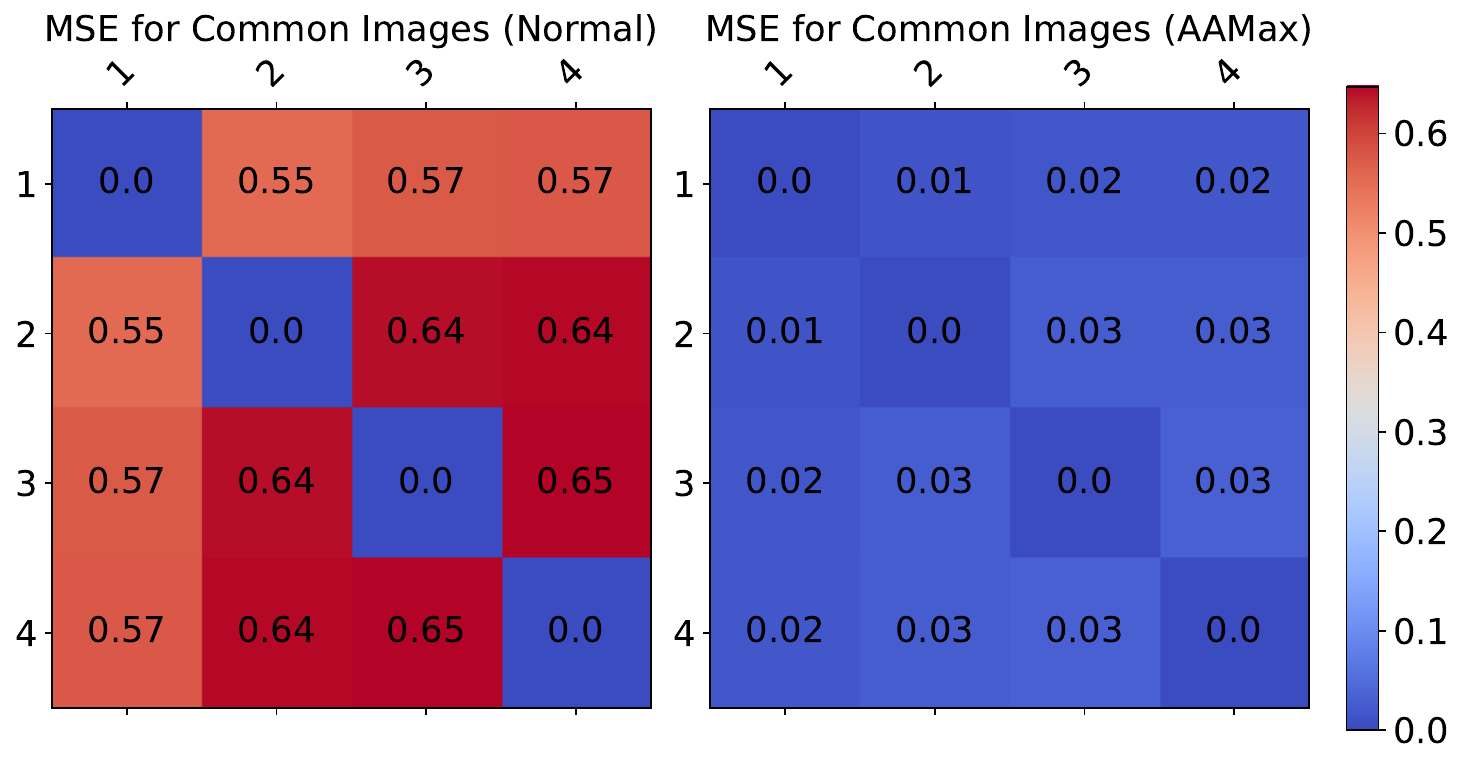}
    \end{subfigure}\hfill
    \caption{Comparing Mean Squared Error for the shared images across subjects in both training techniques. The MSE is compared after training is completed in both techniques. MSE for AAMax is an order of magnitude lower than the baseline}
    \label{fig:mse_common_aamax}
\end{figure}

\subsection{k-Nearest Neighbor Consistency}\label{sec:knn}
We also evaluate the k-nearest neighbors (k-NN) consistency (Figure \ref{fig:knn_common_aamax}) between subjects for shared images after training, using 50 nearest neighbors as the metric. In the end-to-end model, the k-NN consistency is relatively low, indicating that the nearest neighbors of embeddings from one subject do not align well with those from another subject. However, with AAMax, the k-NN consistency significantly improves, demonstrating much better alignment of subject-specific embeddings. AAMax consistently outperforms the normal approach. As with the MSE analysis, this evaluation includes the calibration images; the corresponding result on held-out shared images is reported in Section~\ref{app:neural_corr}.

\begin{figure}[h]
    \centering
    \begin{subfigure}{0.5\textwidth}
        \centering
        \includegraphics[width=\textwidth]{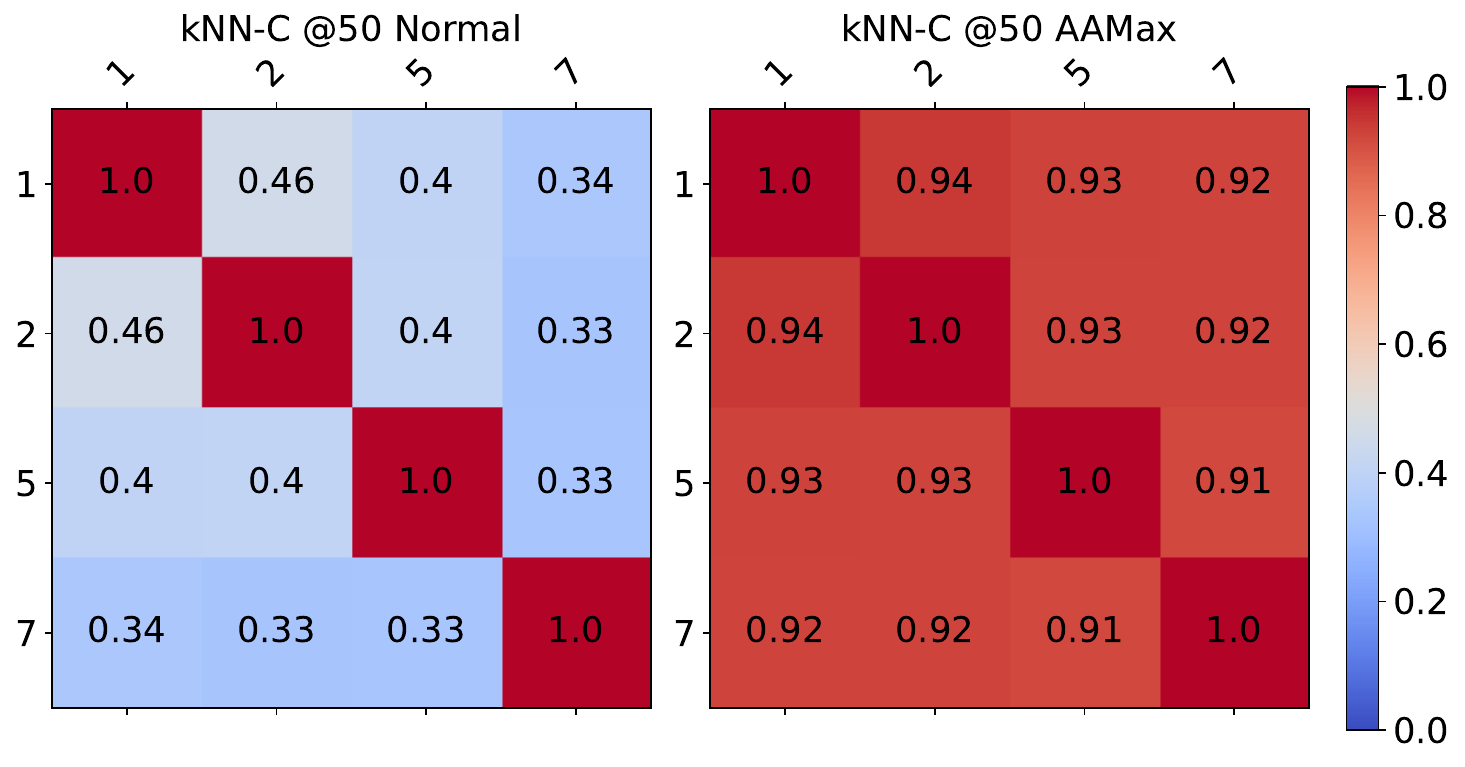}
    \end{subfigure}\hfill
    \caption{Comparing kNN consistency @ 50 for the shared images across subjects in both training techniques. The kNN-C is compared after training is completed in both techniques. AAMax significantly outperforms the baseline.}
    \label{fig:knn_common_aamax}
\end{figure}

\subsection{Neural Correspondence: pRF- and fLOC-Resolved Alignment}\label{app:neural_corr}
The analyses above treat each subject's \texttt{nsdgeneral} voxel vector as a single input. Here we test whether AAMax improves cross-subject alignment of signals originating from homologous, independently defined functional regions of visual cortex: retinotopic areas and eccentricity bands from population receptive field (pRF) mapping, and category-selective regions from the functional localizer (fLOC) experiment. The subject-specific pRF and fLOC labels are distributed with NSD~\citep{allen2022massive} and were estimated independently of our model.

\paragraph{Setup.} We compare baseline fine-tuning and AAMax using MindEye2, with Subject 1 (pre-trained on its full NSD data) as the reference and Subject 2 fine-tuned on 250 shared images for 200 epochs. The two conditions differ only in the training procedure: the baseline performs standard end-to-end fine-tuning, and AAMax adds the adapter alignment stage and the adapter-level MSE loss. The Subject 1 adapter weights and bias are bit-for-bit identical in the two checkpoints, so the reference representation is fixed and any difference reflects how Subject 2 is mapped toward it. All analyses use the 750 shared images that are \emph{not} part of the 250-image calibration set, averaging the three trial repetitions of each image. We analyze the 4{,}096-dimensional output of the subject-specific adapter, before the shared residual network, the CLIP projection, the diffusion prior, or the image generator.

\paragraph{Region-resolved decomposition.} Applying the same \texttt{nsdgeneral} selection to the pRF and fLOC label volumes assigns every model input voxel its subject-specific labels. Individual voxels are not matched across subjects; instead, we compare contributions from homologous regions. Because the MindEye2 adapter is linear, $z_s = W_s x_s + b_s$, the bias-free signal decomposes exactly into regional contributions, $W_s x_s = \sum_r W_{s,r} x_{s,r}$. For each region $r$ we compute the paired cross-subject cosine similarity $\cos(W_{1,r} x_{1,r}, W_{2,r} x_{2,r})$ for each image, excluding the bias so that a shared offset cannot inflate the similarity of weak regional signals. We also report 750-way cross-subject identification (whether Subject 1's embedding of an image is closest to Subject 2's embedding of the same image; chance $=0.13\%$), a leave-one-region-out analysis that removes a region's exact contribution from the full embedding, paired-bootstrap 95\% confidence intervals over 2{,}000 image resamples, and a shuffled-correspondence null with 1{,}000 permutations.

\begin{table}[h]
    \centering
    \caption{Cross-subject alignment of the full adapter output on 750 held-out shared images (Subject 1 $\rightarrow$ Subject 2, MindEye2). Identification is 750-way.}
    \vskip 0.1in
    \setlength{\tabcolsep}{4pt}
    \renewcommand{\arraystretch}{1.1}
    \footnotesize
    \begin{tabular}{lcc}
        \toprule
        \textbf{Metric} & \textbf{Baseline} & \textbf{AAMax} \\
        \midrule
        Paired cosine $\uparrow$ (95\% CI) & 0.366 [0.360, 0.373] & \textbf{0.576} [0.569, 0.582] \\
        Paired MSE $\downarrow$ & 8.42 & \textbf{6.30} \\
        kNN consistency @50 $\uparrow$ & 0.339 & \textbf{0.427} \\
        Subject 1 $\rightarrow$ 2 identification, top-1 / top-5 $\uparrow$ & 50.80\% / 76.27\% & \textbf{92.40\% / 98.93\%} \\
        Subject 2 $\rightarrow$ 1 identification, top-1 / top-5 $\uparrow$ & 49.47\% / 78.13\% & \textbf{87.73\% / 97.47\%} \\
        \bottomrule
    \end{tabular}
    \label{tab:neural_full}
\end{table}

\paragraph{Full-space results.} On held-out shared images, AAMax improves every cross-subject alignment measure (Table~\ref{tab:neural_full}), and the paired cosine increases for 99.9\% of the images. In both conditions, correctly paired embeddings are far more similar than shuffled pairs (shuffled cosine $\approx 0.005$--$0.007$; $p \le 0.001$), so the alignment is stimulus-specific rather than a shared average direction.

\begin{table}[h]
    \centering
    \caption{Region-resolved paired cosine similarity on 750 held-out shared images. Each entry compares only the contribution of the given pRF-defined area, eccentricity band, or fLOC-defined category system. 95\% paired-bootstrap confidence intervals of the difference are shown for the fLOC systems.}
    \vskip 0.1in
    \setlength{\tabcolsep}{4pt}
    \renewcommand{\arraystretch}{1.1}
    \footnotesize
    \begin{tabular}{llccc}
        \toprule
        & \textbf{Region} & \textbf{Baseline} & \textbf{AAMax} & \textbf{$\Delta$} \\
        \midrule
        \multirow{5}{*}{pRF areas} & All pRF voxels & 0.344 & \textbf{0.505} & +0.160 \\
        & V1 & 0.286 & \textbf{0.394} & +0.108 \\
        & V2 & 0.241 & \textbf{0.335} & +0.094 \\
        & V3 & 0.186 & \textbf{0.266} & +0.080 \\
        & hV4 & 0.171 & \textbf{0.250} & +0.079 \\
        \midrule
        \multirow{5}{*}{Eccentricity} & $0.5^\circ$ & 0.102 & \textbf{0.180} & +0.078 \\
        & $1^\circ$ & 0.094 & \textbf{0.194} & +0.100 \\
        & $2^\circ$ & 0.243 & \textbf{0.360} & +0.117 \\
        & $4^\circ$ & 0.270 & \textbf{0.372} & +0.102 \\
        & $>4^\circ$ & 0.339 & \textbf{0.441} & +0.102 \\
        \midrule
        \multirow{5}{*}{fLOC systems} & All fLOC voxels & 0.342 & \textbf{0.488} & +0.146 [0.142, 0.151] \\
        & Faces & 0.217 & \textbf{0.303} & +0.085 [0.080, 0.090] \\
        & Places & 0.278 & \textbf{0.357} & +0.079 [0.075, 0.084] \\
        & Bodies & 0.305 & \textbf{0.437} & +0.132 [0.126, 0.137] \\
        & Words & 0.101 & \textbf{0.168} & +0.067 [0.062, 0.073] \\
        \bottomrule
    \end{tabular}
    \label{tab:neural_regions}
\end{table}

\paragraph{pRF-resolved results.} AAMax increases the paired cosine of the contributions from every retinotopic area (V1, V2, V3, and hV4) and every eccentricity band (Table~\ref{tab:neural_regions}), and it does so in all 20 area $\times$ eccentricity cells. Using only pRF-defined voxels, 750-way Subject 1 $\rightarrow$ 2 identification improves from 28.93\% to 64.53\%. The improvement is therefore not confined to an undifferentiated whole-brain vector: it is present within independently defined early and intermediate visual areas across the visual field. Absolute values should not be compared across regions, since regions differ in voxel count, pRF reliability, signal-to-noise ratio, and cortical magnification; the informative comparison is AAMax versus the baseline within the same region.

\begin{table}[h]
    \centering
    \caption{Paired cosine similarity for individual bilateral fLOC subregions on 750 held-out shared images, with 95\% paired-bootstrap confidence intervals of the difference. Subregions absent in either subject (mTL-faces, aTL-faces, RSC, FBA-1, mTL-bodies, mTL-words) are excluded.}
    \vskip 0.1in
    \setlength{\tabcolsep}{4pt}
    \renewcommand{\arraystretch}{1.1}
    \footnotesize
    \begin{tabular}{lccc}
        \toprule
        \textbf{Subregion} & \textbf{Baseline} & \textbf{AAMax} & \textbf{$\Delta$ [95\% CI]} \\
        \midrule
        OFA & 0.121 & \textbf{0.133} & +0.011 [0.006, 0.016] \\
        FFA-1 & 0.120 & \textbf{0.186} & +0.066 [0.061, 0.070] \\
        FFA-2 & 0.156 & \textbf{0.236} & +0.079 [0.074, 0.086] \\
        OPA & 0.099 & \textbf{0.173} & +0.074 [0.069, 0.080] \\
        PPA & 0.292 & \textbf{0.359} & +0.068 [0.064, 0.072] \\
        EBA & 0.271 & \textbf{0.395} & +0.124 [0.118, 0.130] \\
        FBA-2 & 0.122 & \textbf{0.193} & +0.071 [0.065, 0.077] \\
        OWFA & 0.086 & \textbf{0.121} & +0.035 [0.031, 0.040] \\
        VWFA-1 & 0.065 & \textbf{0.123} & +0.057 [0.052, 0.063] \\
        VWFA-2 & 0.007 & \textbf{0.016} & +0.009 [0.005, 0.014] \\
        mfs-words & 0.054 & \textbf{0.069} & +0.015 [0.011, 0.019] \\
        \bottomrule
    \end{tabular}
    \label{tab:neural_subregions}
\end{table}

\paragraph{fLOC-resolved results.} AAMax improves the paired cosine in all four category-selective systems and in every retained bilateral subregion, and every paired-bootstrap confidence interval excludes zero (Tables~\ref{tab:neural_regions} and~\ref{tab:neural_subregions}). For every system and both conditions, correctly paired contributions exceed all 1{,}000 shuffled-correspondence permutations ($p \le 0.001$). Using only fLOC-defined voxels, 750-way identification improves from 14.53\% to 31.73\%. The gains are largest in EBA, FFA-2, OPA, and FBA-2, and small in OFA, VWFA-2, and mfs-words. Not every measure improves uniformly: for word-selective voxels, kNN consistency is essentially unchanged (0.164 vs.\ 0.163) despite the improved paired cosine and identification.

\begin{table}[h]
    \centering
    \caption{Leave-one-region-out analysis: decrease in the full paired cosine similarity when a region's exact contribution is removed from the adapter output. Larger values indicate that the region contributes more to the cross-subject alignment.}
    \vskip 0.1in
    \setlength{\tabcolsep}{4pt}
    \renewcommand{\arraystretch}{1.1}
    \footnotesize
    \begin{tabular}{llcc}
        \toprule
        & \textbf{Region removed} & \textbf{Baseline} & \textbf{AAMax} \\
        \midrule
        \multirow{5}{*}{pRF} & V1 & 0.020 & 0.026 \\
        & V2 & 0.011 & 0.018 \\
        & V3 & 0.008 & 0.019 \\
        & hV4 & 0.005 & 0.010 \\
        & All pRF voxels & 0.029 & 0.056 \\
        \midrule
        \multirow{5}{*}{fLOC} & Faces & 0.006 & 0.011 \\
        & Places & 0.023 & 0.037 \\
        & Bodies & 0.033 & 0.047 \\
        & Words & $-0.006$ & 0.000 \\
        & All fLOC voxels & 0.051 & 0.081 \\
        \bottomrule
    \end{tabular}
    \label{tab:neural_leaveout}
\end{table}

\paragraph{Leave-one-region-out.} Removing the contributions of retinotopic areas or of face-, place-, or body-selective regions reduces cross-subject alignment more under AAMax than under the baseline (Table~\ref{tab:neural_leaveout}), indicating that these homologous signals contribute more consistently to the aligned representation. Removing word-selective signals has essentially no effect, suggesting that they carry stimulus-specific correspondence in isolation but contribute little to the full cross-subject similarity.

\paragraph{Scope.} These results use independently obtained functional atlases to decompose and validate the learned representation. They show that AAMax improves the stimulus-specific functional alignment of signals arising from homologous retinotopic and category-selective regions, on images not used for alignment. They do not show that AAMax establishes cortical vertex correspondence, and AAMax does not use cortical mesh geometry during training: anatomical alignment operates upstream of the reconstruction pipeline, whereas AAMax operates after the subject-specific adapter has absorbed anatomical variability. The regional decomposition is exact for the linear adapter, but the downstream nonlinear network may combine regional contributions non-additively.

\section{Pre-Diffusion Retrieval}\label{app:retrieval}
Reconstruction metrics depend on both the fMRI-derived conditioning signal and the pre-trained generative model, so they cannot on their own show how much image-specific information is recovered from the fMRI signal. To measure this directly, we evaluate retrieval on the fMRI-derived embeddings before any diffusion prior or image generator is applied.

\paragraph{Setup.} We fine-tune Subject 2 with 1 hour of data (250 images) using Subject 1 as the reference, for both MindEye1 (CLIP ViT-L/14) and MindEye2 (OpenCLIP ViT-bigG/14), and compare baseline fine-tuning with AAMax using the epoch-200 checkpoints. Retrieval uses the contrastive embedding produced immediately after the subject-specific adapter and the shared voxel-to-CLIP network, and is evaluated on the 1{,}000-image Subject 2 test set. The baseline uses the same architecture, subject-specific adapter, voxel-to-CLIP network, contrastive objective, data, and checkpoint epoch as AAMax, so the comparison isolates the effect of the training procedure. We report two protocols:
\begin{itemize}
    \item \textbf{Global retrieval:} top-1, top-5, and top-10 accuracy among all 1{,}000 test images.
    \item \textbf{MindEye protocol}~\citep{scotti2024reconstructing,scotti2024mindeye2}: top-1 accuracy among 300 candidates (the correct match and 299 randomly sampled distractors), repeated 30 times with different candidate samples.
\end{itemize}
Image retrieval denotes fMRI$\rightarrow$image matching, and brain retrieval denotes the reverse direction.

\begin{table}[h]
    \centering
    \caption{Pre-diffusion retrieval for Subject 2 fine-tuned with 1 hour of data from Subject 1. Global retrieval reports top-1/5/10 accuracy (\%) among all 1{,}000 test images. The 300-way protocol reports mean $\pm$ standard deviation (\%) of top-1 accuracy across 30 random candidate samplings of a single trained model.}
    \vskip 0.1in
    \setlength{\tabcolsep}{4pt}
    \renewcommand{\arraystretch}{1.15}
    \footnotesize
    \resizebox{\columnwidth}{!}{
    \begin{tabular}{llcccc}
        \toprule
        & & \multicolumn{2}{c}{\textbf{1{,}000-way (T1 / T5 / T10)}} & \multicolumn{2}{c}{\textbf{300-way (top-1)}} \\
        \cmidrule(lr){3-4} \cmidrule(lr){5-6}
        \textbf{Model} & \textbf{Training} & Image retrieval $\uparrow$ & Brain retrieval $\uparrow$ & Image retrieval $\uparrow$ & Brain retrieval $\uparrow$ \\
        \midrule
        MindEye1 & Baseline & 21.7 / 45.2 / 55.1 & 25.4 / 50.1 / 61.7 & 34.78 $\pm$ 0.76 & 39.24 $\pm$ 0.99 \\
        MindEye1 & AAMax & \textbf{29.0 / 55.9 / 66.6} & \textbf{36.1 / 64.8 / 74.2} & \textbf{44.20 $\pm$ 0.80} & \textbf{52.08 $\pm$ 0.92} \\
        \midrule
        MindEye2 & Baseline & 33.2 / 63.8 / 75.0 & 39.8 / 70.0 / 80.2 & 50.87 $\pm$ 0.88 & 56.18 $\pm$ 1.23 \\
        MindEye2 & AAMax & \textbf{39.6 / 69.6 / 78.6} & \textbf{45.1 / 75.2 / 85.5} & \textbf{56.60 $\pm$ 0.79} & \textbf{63.00 $\pm$ 0.74} \\
        \bottomrule
    \end{tabular}
    }
    \label{tab:retrieval}
\end{table}

\paragraph{Results.} AAMax improves retrieval consistently for both architectures, in both retrieval directions, and under both candidate-pool protocols (Table~\ref{tab:retrieval}). Under the 300-way protocol, AAMax improves image retrieval by 9.4 and 5.7 percentage points for MindEye1 and MindEye2, respectively, and brain retrieval by 12.8 and 6.8 points. Because these measurements are taken before diffusion sampling, they show that the improvements of AAMax reflect additional image-specific information in the fMRI-derived embeddings, rather than a stronger reliance on the generative prior. The reported standard deviations capture variability across candidate samplings of a single trained model per condition, not variability across training runs.

\section{Emergence of Common Concepts: Additional Details}\label{app:common_concepts}


The existence of a semantic common representation space is suggested from the emergence of semantic categories in this space. For this, we extracted the most significant left singular vectors from the ridge regression weight matrix. This matrix was extracted from the pre-trained MindEye2 \citep{scotti2024mindeye2} model which was trained using the 1024-dimensional common space. Subsequently, we projected the 1024-dimensional embeddings of each image onto these vectors and sampled from the two extremes (set 1 and set 3) and the middle (set 2). In order to obtain an objective classification, we presented these sets of images to GPT-4o \citep{achiam2023gpt} and asked it to find the semantic categories that are well represented in set 1, somewhat represented in set 2 and not represented in set 3. Results for the first five singular vectors are shown in Table~\ref{tab:singular-concepts}. 
As is evident, images in the principal dimensions of the representation space have a common theme. 

\begin{figure}[!h] 
    \centering
    \includegraphics[width=0.5\textwidth]{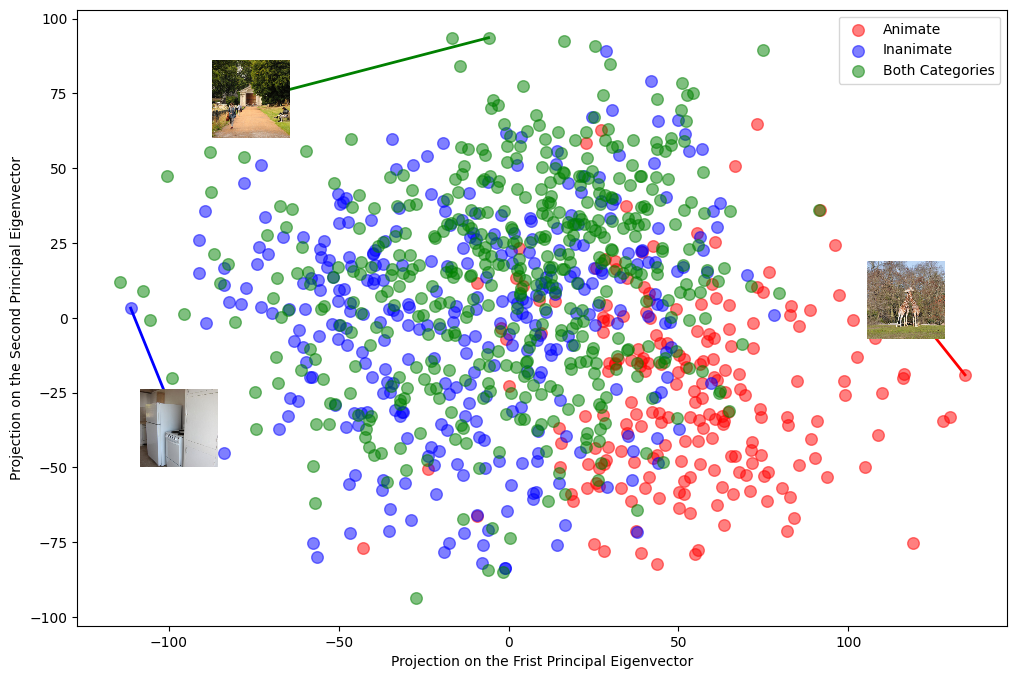} 
    \caption{Visualization of the Animate/Inanimate semantic category on the two principal singular vectors.  Each point (embedding) is colored based on whether its corresponding image contains `objects' that fall into the Animate category (red), Inanimate category (blue) or both categories (green).}
    \label{fig:embs}
\end{figure}

Previous studies, such as those by \cite{huth2012continuous}, have indicated that animal categories, including humans, are represented distinctly from non-animal categories within the brain. In our study, we observed that the common representational space also preserves these semantic properties. As demonstrated in Figure \ref{fig:embs}, after projecting the fMRI signals of Subject 2 into the common space and onto the first two singular vectors, we can see the animal vs. non-animal distinction, consistent with previous hypotheses. Extending this analysis across additional dimensions, we observed similar separations between other hypothesized semantic categories of "Social interaction," "man-made vs nature," and "biological vs non-biological."

\begin{table}[h]
\vskip 0.15in
\begin{center}
\caption{Concept analysis of images along principal dimensions.}
\vskip 0.1in
\begin{small}
\begin{sc}
\resizebox{\columnwidth}{!}{%
    \begin{tabular}{|c|c|}
    \hline
    Dimension 1 & Animals (giraffes, zebras, and elephants), nature, and outdoor scenes \\ \hline
    Dimension 2 & Cats and food \\ \hline
    Dimension 3 & Kite flying, airplanes, and flying objects \\ \hline
    Dimension 4 & Surfing, water sports, beach, and ocean scenes \\ \hline
    Dimension 5 & Airplanes, flying objects, outdoor activities and events, transportation scenes \\ \hline
    \end{tabular}
}
\end{sc}
\end{small}

\label{tab:singular-concepts}
\end{center}
\vskip -0.1in
\end{table}

To determine if dominant semantic categories were present within these sets, we employed ChatGPT with the query: ``Is there a dominant semantic category to which the majority of images in a set belong?'' The methodology for defining semantic categories was consistent with that used in \cite{huth2012continuous}. Table \ref{tab:singular-concepts} illustrates our findings.


After projecting the fMRI signals of individual images from Subject 2 into the common space and onto the first singular vectors, we observe that certain concepts are strongly represented along these dimensions defined by the singular vectors, aligning with previous hypotheses (Figures \ref{fig:first_ev}, \ref{fig:fourth_ev}).

\begin{figure}[h]
    \centering
    \begin{subfigure}{0.45\textwidth}
        \centering
        \includegraphics[width=\textwidth]{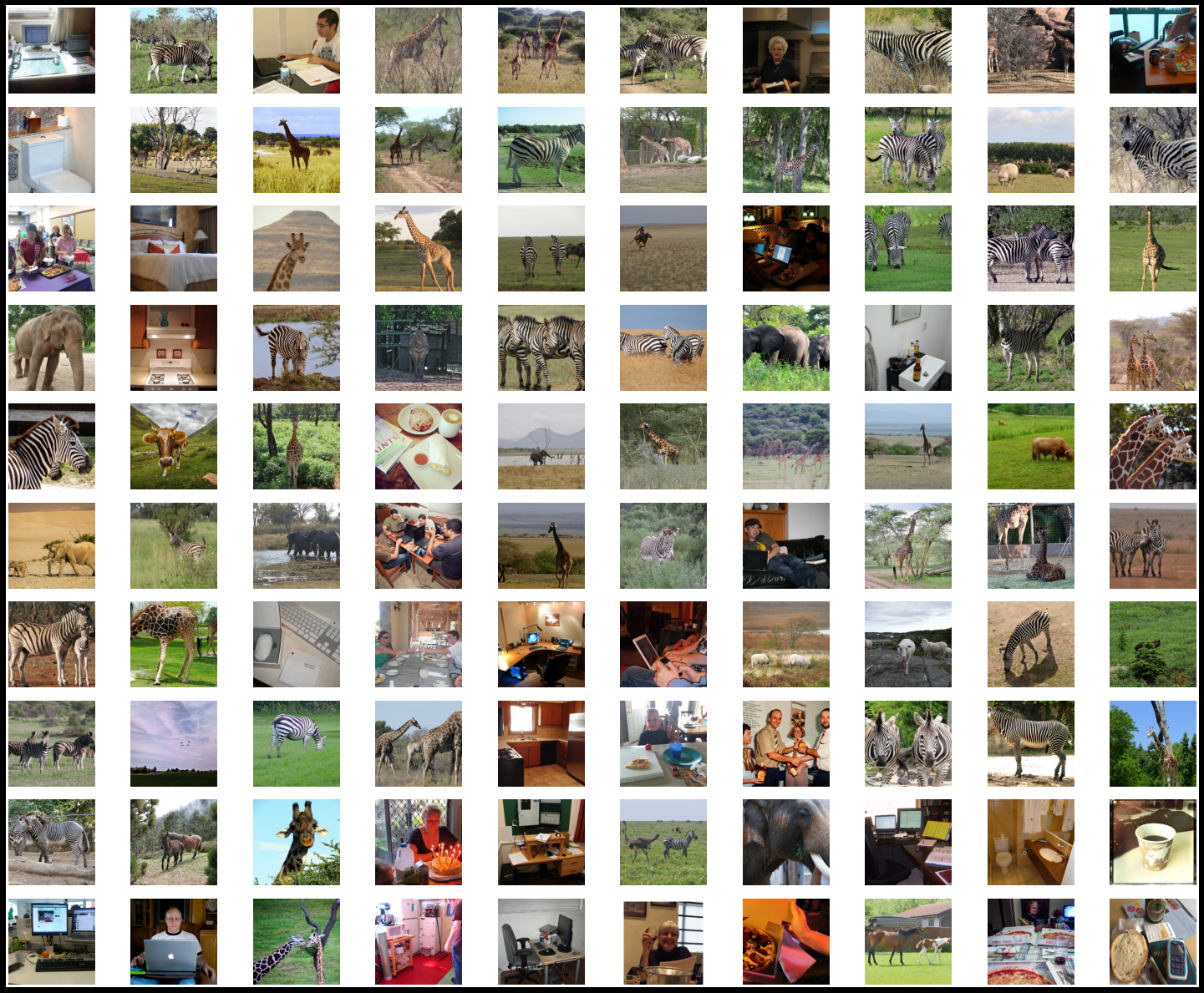}
    \end{subfigure}\hfill
    \begin{subfigure}{0.45\textwidth}
        \centering
        \includegraphics[width=\textwidth]{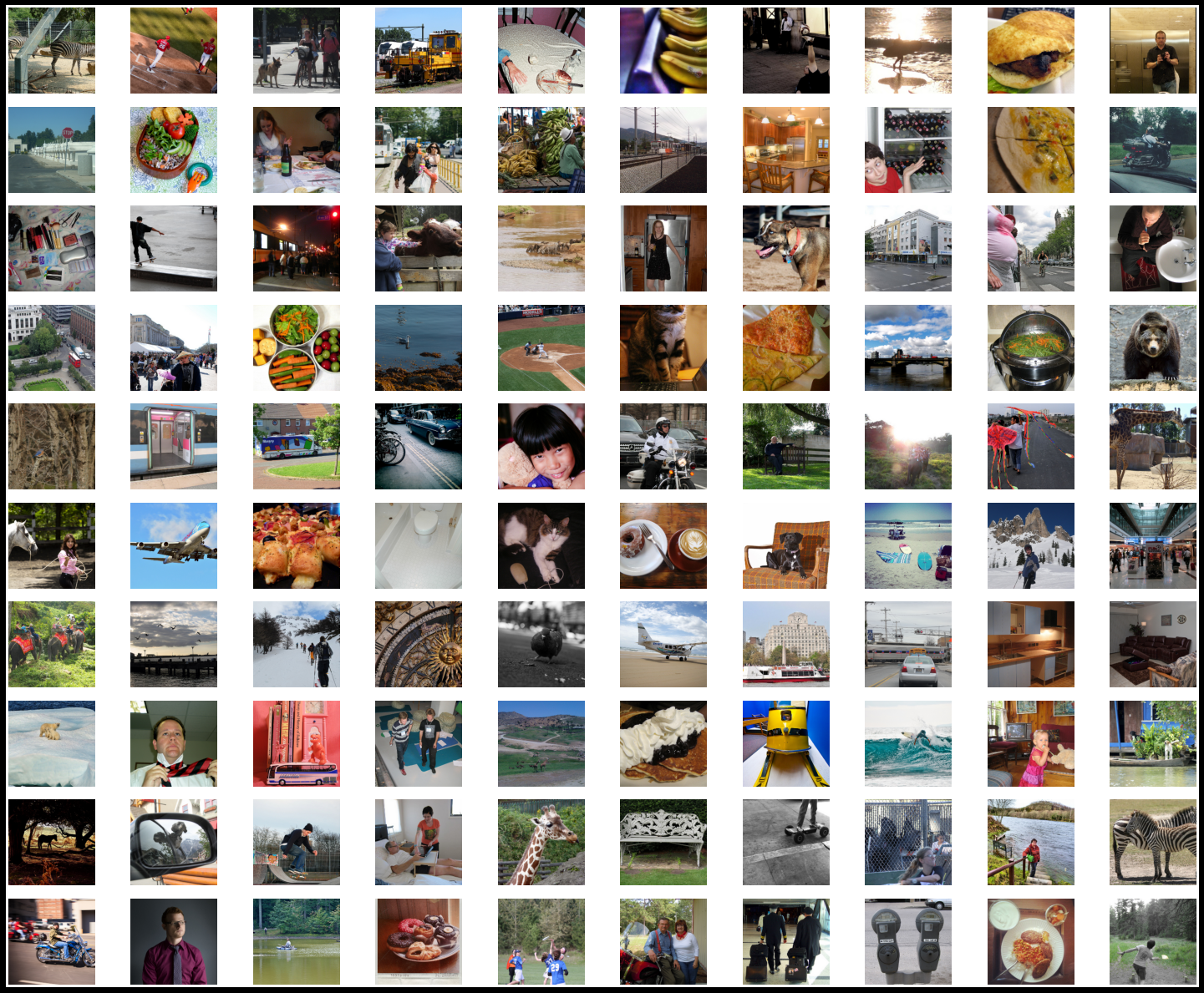}
    \end{subfigure}\hfill
    \caption{Comparison of images from Subject 2 with the highest and lowest projection values along the first principal left eigenvector. (a) Images whose corresponding 1024-dimensional embeddings yielded the highest projection values. (b) Images whose corresponding 1024-dimensional embeddings yielded the lowest projection values. The presence of the animate concept (e.g., giraffes, zebras, elephants) is prominent in the first set, while it fades in the second.}
    \label{fig:first_ev}
\end{figure}

\begin{figure}[h]
    \centering
    \begin{subfigure}{0.45\textwidth}
        \centering
        \includegraphics[width=\textwidth]{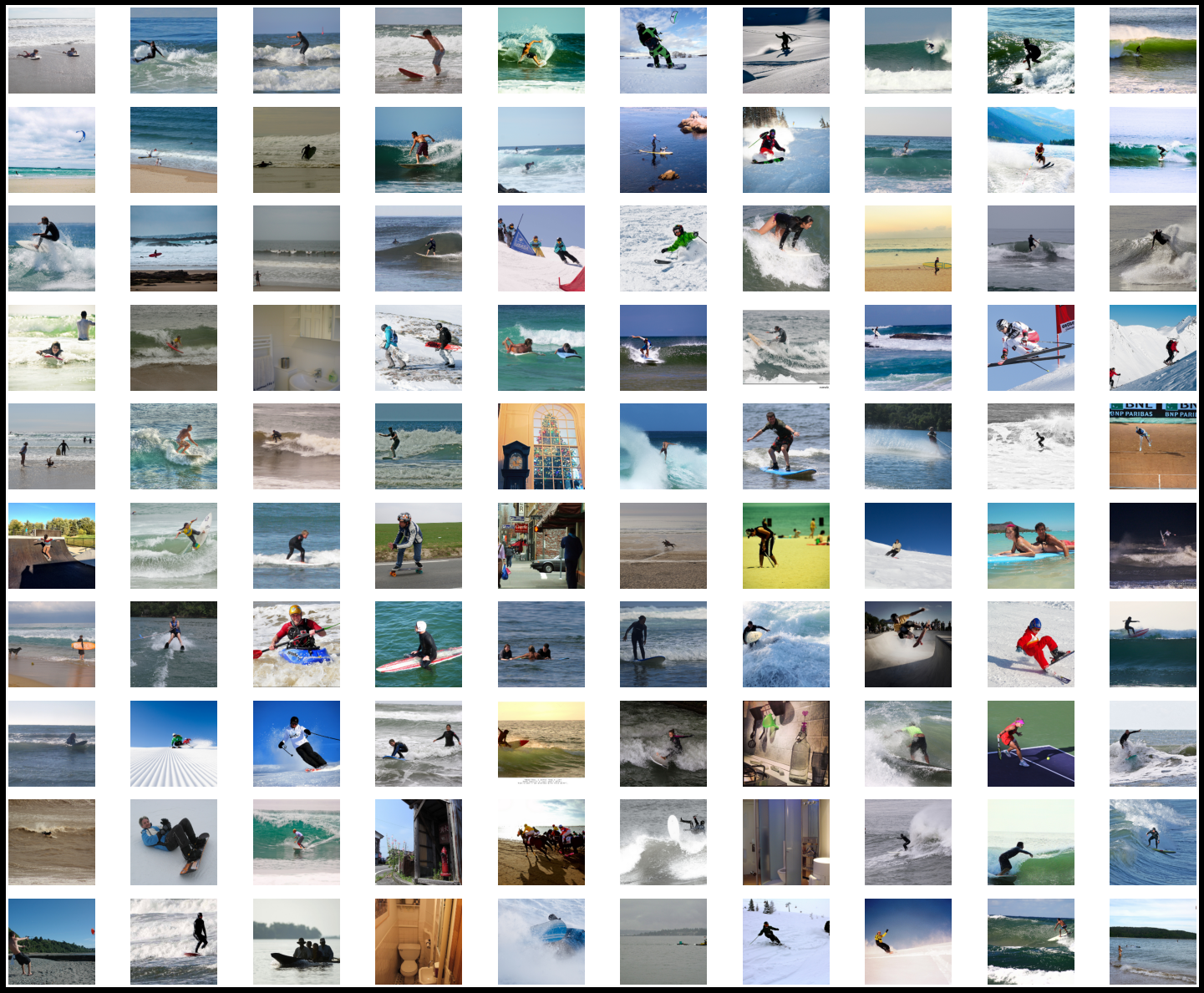}
    \end{subfigure}\hfill
    \begin{subfigure}{0.45\textwidth}
        \centering
        \includegraphics[width=\textwidth]{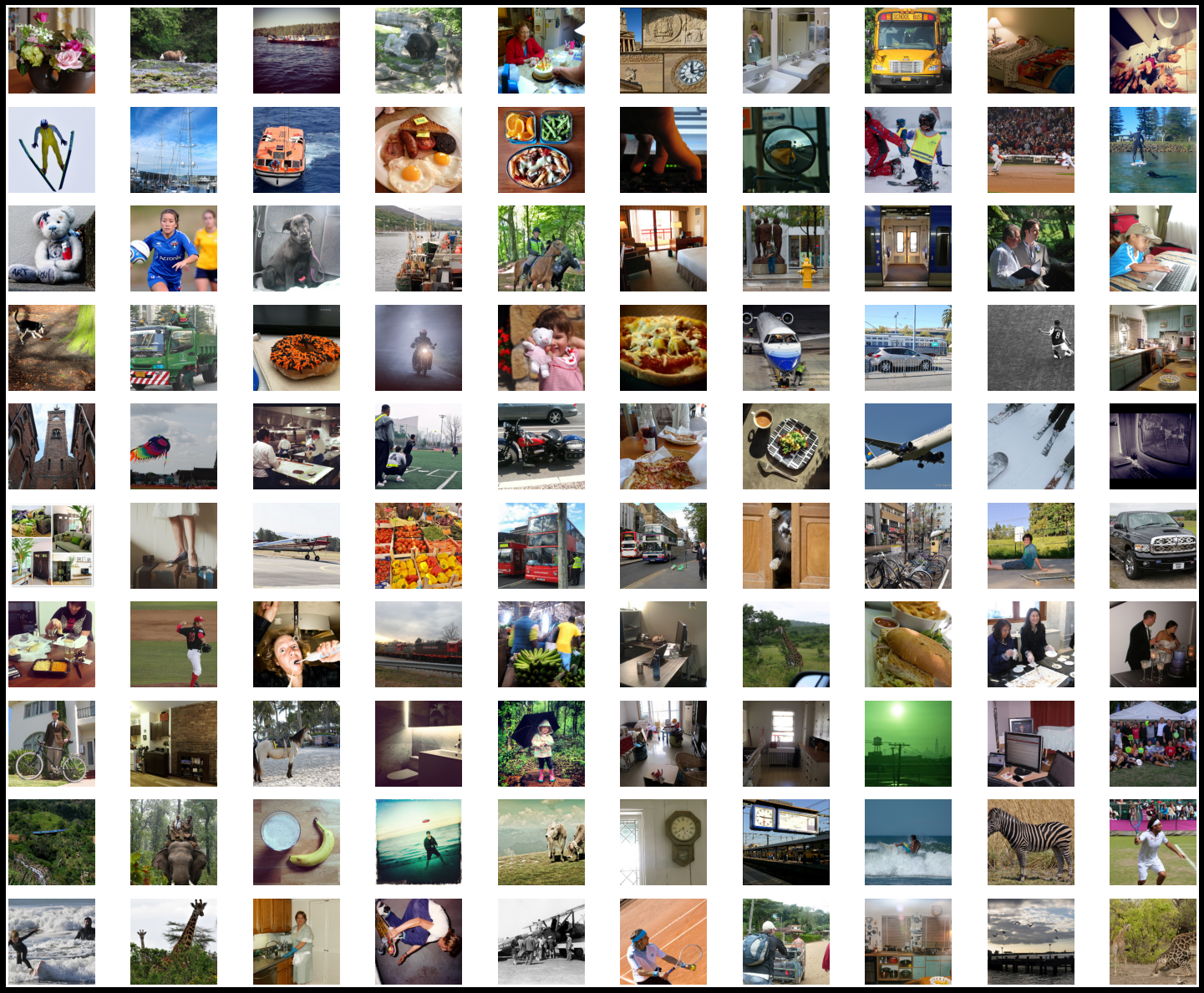}
    \end{subfigure}\hfill
    \caption{Comparison of images from Subject 2 with the highest and lowest projection values along the 4$^{th}$ principal left eigenvector. (a) Images whose corresponding 1024-dimensional embeddings yielded the highest projection values. (b) Images whose corresponding 1024-dimensional embeddings yielded the lowest projection values. Again the presence of common concepts (e.g. surfing, water sports, beach, and ocean scenes) is prominent in the first set, while it fades in the second.}
    \label{fig:fourth_ev}
\end{figure}

\section{Image Selection Algorithm: Formulations and Proofs}\label{app:nphardness}
This appendix formalizes the image selection problem of Section~\ref{sec:imageselection}, shows that it is NP-hard, and states the approximation guarantees of the greedy procedure (Algorithm~\ref{algo:greedyalgo}). The same greedy procedure supports two formulations that differ only in their stopping criterion: a \emph{set cover} formulation, which seeks the smallest subset spanning the discretized eigenspace, and a \emph{maximum coverage} formulation, which selects a fixed budget of images to maximize coverage.

\subsection{Setup and Formulations}\label{app:is_formulations}
Let $\mathcal{I}$ denote the pool of candidate images (the 1{,}000 NSD shared images). The reference subject's shared-space embeddings of these images are projected onto the top $d$ principal directions of the reference adapter, and the projected range of each dimension $j \in \{1, \dots, d\}$ is partitioned into $B_j = \lfloor w \cdot \lambda_j / \lambda_1 \rfloor$ bins, where $\lambda_j$ is the $j$-th singular value and $w$ is a scaling parameter. The universe of bins is
\[
U = \bigcup_{j=1}^{d} \{(j, b) : b \in [B_j]\}, \qquad |U| = \sum_{j=1}^{d} B_j .
\]
Each image $i \in \mathcal{I}$ falls into exactly one bin in every dimension, so it covers the set $C_i = \{(j, \mathrm{bin}_j(i)) : j = 1, \dots, d\} \subseteq U$ with $|C_i| = d$. For a subset $S \subseteq \mathcal{I}$, the coverage function and the gap function of Section~\ref{sec:imageselection} are
\[
f(S) = \Big| \bigcup_{i \in S} C_i \Big|, \qquad \mathrm{Gap}(S) = \sum_{j=1}^{d} \mathrm{Gap}(S, j) = |U| - f(S),
\]
where $\mathrm{Gap}(S, j)$ is the number of bins in dimension $j$ that contain no image of $S$.

\paragraph{Formulation 1 (Set Cover).} Find the smallest subset that covers every bin:
\[
\min_{S \subseteq \mathcal{I}} |S| \quad \text{s.t.} \quad f(S) = |U| .
\]

\paragraph{Formulation 2 (Maximum Coverage).} Given a budget $k$, find the subset of at most $k$ images that covers the most bins:
\[
\max_{S \subseteq \mathcal{I},\ |S| \le k} f(S) .
\]

Algorithm~\ref{algo:greedyalgo} solves both formulations with the same greedy rule: at every step it adds the image that minimizes $\mathrm{Gap}(S \cup \{i\})$, which is equivalent to maximizing the marginal coverage gain $f(S \cup \{i\}) - f(S)$. Running until $\mathrm{Gap}(S) = 0$ yields a set cover (Formulation 1); stopping after $k$ images yields a maximum coverage solution (Formulation 2). Because the greedy procedure is deterministic, the $k$-image solution is exactly the first $k$ images of the full-coverage sequence. In our experiments we use the fixed-budget variant with $k = 250$ images (one hour of data). With $d = 40$ and $w = 105$, a full cover requires 255 images, so the two formulations select nearly identical subsets in practice.

\subsection{NP-Hardness}
\begin{definition}[Bin Covering Problem]
Given a set of vectors, each of which places one element into one of the bins of every dimension, and an integer $k$, decide whether there exists a subset of at most $k$ vectors that covers every bin of every dimension.
\end{definition}

\begin{definition}[Set Cover Problem]
Given a finite universe $\mathcal{U}$, a collection $\mathcal{S} = \{S_1, \dots, S_l\}$ of subsets of $\mathcal{U}$ with $\bigcup_i S_i = \mathcal{U}$, and an integer $k$, decide whether there exists a subcollection $\mathcal{S}' \subseteq \mathcal{S}$ with $|\mathcal{S}'| \le k$ and $\bigcup_{S \in \mathcal{S}'} S = \mathcal{U}$.
\end{definition}

\begin{theorem}
The bin covering problem is NP-hard, even when every dimension has exactly three bins.
\end{theorem}

\begin{proof}
We reduce from Set Cover. Given an instance $(\mathcal{U}, \mathcal{S}, k)$ with $\mathcal{U} \neq \emptyset$, we create one dimension for every element $e \in \mathcal{U}$ and two additional dimensions $p$ and $q$, each with three bins. We construct one vector per input set and two special vectors:
\begin{itemize}
    \item For each $S_i \in \mathcal{S}$, the vector $V_i$ selects bin~1 in dimension $e$ if $e \in S_i$ and bin~2 otherwise, and selects bin~2 in dimensions $p$ and $q$.
    \item The vector $N_p$ selects bin~2 in every element dimension, bin~1 in dimension $p$, and bin~3 in dimension $q$.
    \item The vector $N_q$ selects bin~3 in every element dimension, bin~3 in dimension $p$, and bin~1 in dimension $q$.
\end{itemize}
The construction is polynomial in $|\mathcal{U}|$ and $l$. By construction, bin~1 of dimension $p$ is covered only by $N_p$, bin~1 of dimension $q$ is covered only by $N_q$, and bin~1 of an element dimension $e$ is covered only by the vectors $V_i$ with $e \in S_i$. We show that $\mathcal{U}$ has a set cover of size at most $k$ if and only if there is a bin cover of size at most $k + 2$.

($\Rightarrow$) Let $\mathcal{S}'$ be a set cover with $|\mathcal{S}'| \le k$, and take the vectors $\{V_i : S_i \in \mathcal{S}'\} \cup \{N_p, N_q\}$. In every element dimension, bin~1 is covered because $\mathcal{S}'$ covers $\mathcal{U}$, bin~2 is covered by $N_p$, and bin~3 is covered by $N_q$. In dimensions $p$ and $q$, bin~1 is covered by $N_p$ and $N_q$ respectively, bin~2 is covered by any $V_i$ in the selection (which is non-empty because $\mathcal{U} \neq \emptyset$), and bin~3 is covered by $N_q$ and $N_p$ respectively. The selection therefore covers every bin and has size at most $k + 2$.

($\Leftarrow$) Let $\mathcal{V}'$ be a bin cover with $|\mathcal{V}'| \le k + 2$. Because bin~1 of dimensions $p$ and $q$ is covered only by $N_p$ and $N_q$, both special vectors belong to $\mathcal{V}'$, so the remaining vectors are at most $k$ vectors of the form $V_i$. Neither special vector covers bin~1 of any element dimension, so for every $e \in \mathcal{U}$ some remaining $V_i$ has $e \in S_i$. The corresponding sets therefore form a set cover of size at most $k$.

Since Set Cover is NP-hard, the bin covering problem is NP-hard. The general problem, with an arbitrary number of bins $B_j$ in each dimension, contains this three-bin case and is therefore NP-hard as well. Deciding whether $k$ images can cover all of $U$ is exactly this decision problem, so the optimization versions of both Formulation 1 and Formulation 2 are NP-hard.
\end{proof}

\subsection{Monotonicity and Submodularity of the Coverage Function}\label{app:approxproof}
A set function $f$ is monotone if $f(A) \le f(B)$ for all $A \subseteq B$, and submodular if it has diminishing returns~\citep{schrijver2003combinatorial}:
\[
f(A \cup \{x\}) - f(A) \geq f(B \cup \{x\}) - f(B), \qquad \forall A \subseteq B,\ x \notin B .
\]

\begin{theorem}
The coverage function $f(S) = |\bigcup_{i \in S} C_i|$ is monotone and submodular. Equivalently, $\mathrm{Gap}(S) = |U| - f(S)$ is non-increasing and supermodular.
\end{theorem}

\begin{proof}
Write $f(S) = \sum_{j=1}^{d} f_j(S)$, where $f_j(S)$ is the number of covered bins in dimension $j$. Adding an image $x$ to a subset $A$ increases $f_j$ by
\[
\Delta_A(x, j) = f_j(A \cup \{x\}) - f_j(A) =
\begin{cases}
1 & \text{if no image of } A \text{ lies in bin } \mathrm{bin}_j(x), \\
0 & \text{otherwise.}
\end{cases}
\]
Since $\Delta_A(x, j) \ge 0$, $f$ is monotone. If $A \subseteq B$, every bin covered by $A$ is also covered by $B$, so $\Delta_A(x, j) \ge \Delta_B(x, j)$ for every dimension $j$. Summing over dimensions gives
\[
f(A \cup \{x\}) - f(A) = \sum_{j=1}^{d} \Delta_A(x, j) \;\ge\; \sum_{j=1}^{d} \Delta_B(x, j) = f(B \cup \{x\}) - f(B),
\]
so $f$ is submodular. Because $\mathrm{Gap}(S) = |U| - f(S)$, the marginal change of the gap is $\mathrm{Gap}(A \cup \{x\}) - \mathrm{Gap}(A) = -\sum_j \Delta_A(x, j) \le 0$, and the inequality above reverses, i.e., $\mathrm{Gap}(A \cup \{x\}) - \mathrm{Gap}(A) \le \mathrm{Gap}(B \cup \{x\}) - \mathrm{Gap}(B)$. The gap function is therefore non-increasing and supermodular, and minimizing $\mathrm{Gap}(S \cup \{i\})$ is the same as maximizing the marginal gain of $f$.
\end{proof}

\subsection{Approximation Guarantees}
\begin{theorem}[Set Cover]
Run until $\mathrm{Gap}(S) = 0$, Algorithm~\ref{algo:greedyalgo} returns a cover of size at most $H(d) \cdot \mathrm{OPT}$, where $\mathrm{OPT}$ is the size of a minimum cover and $H(d) = \sum_{m=1}^{d} 1/m \le \ln d + 1$.
\end{theorem}
This is the classical greedy guarantee for set cover, which is $H(s)$ in the size $s$ of the largest input set~\citep{johnson1974approximation,lovasz1975ratio}. In our setting every image covers exactly one bin per dimension, so $s = |C_i| = d$. For $d = 40$, $H(40) \approx 4.28$.

\begin{theorem}[Maximum Coverage]
Stopped after $k$ images, Algorithm~\ref{algo:greedyalgo} returns a subset $S_k$ with
\[
f(S_k) \ \ge\ \Big(1 - \frac{1}{e}\Big) \max_{S \subseteq \mathcal{I},\ |S| \le k} f(S) .
\]
\end{theorem}
This follows from the monotonicity and submodularity of $f$~\citep{nemhauser1978analysis}.

\subsection{Greedy Image Selection Algorithm}\label{app:alg}
Algorithm~\ref{algo:greedyalgo} presents the greedy procedure. Setting $k = \infty$ gives the set cover formulation; a finite $k$ gives the maximum coverage formulation used in our experiments.

\begin{algorithm}
\caption{Greedy Heuristic for Bin Coverage}
\label{algo:greedyalgo}
\begin{algorithmic}[1]
\REQUIRE Set of images $\mathcal{I}$, number of dimensions $d$, scaling parameter $w$, budget $k$ ($k = \infty$ for full coverage)
\ENSURE $S$: selected subset of images
\STATE Initialize $S \gets \emptyset$
\FOR{each dimension $j = 1, 2, \ldots, d$}
    \STATE Compute the number of bins $B_j = \left\lfloor w \cdot \frac{\lambda_j}{\lambda_1} \right\rfloor$
\ENDFOR
\STATE Define $Gap(S, j)$: number of empty bins in dimension $j$
\STATE Define $Gap(S) \gets \sum_{j} Gap(S, j)$
\WHILE{$Gap(S) > 0$ \AND $|S| < k$}
    \STATE Select $i \in \mathcal{I} \setminus S$ minimizing $Gap(S \cup \{i\})$
    \STATE Update $S \gets S \cup \{i\}$
\ENDWHILE
\end{algorithmic}
\end{algorithm}

\section{Training Time and Computational Cost}\label{app:tr_time}
This appendix reports the training time, model size, inference time, and hyperparameter-search scope of AAMax relative to baseline fine-tuning. All experiments were run on a single NVIDIA H200 GPU with 141 GB of memory. We separate the cost of the underlying architecture, which is shared by both methods, from the incremental cost of AAMax.

\subsection{Training Time}
We measured the run time of each step of the training process (Table~\ref{tab:training_times}). The adapters are first allowed to overfit to the reference embeddings by training them for 300 epochs with the MSE loss, which takes around one minute. End-to-end fine-tuning then has essentially the same per-epoch cost for the baseline and for AAMax: apart from the adapter-level MSE loss, both pipelines run the same architecture on the same data. For a complete 160-epoch run, baseline fine-tuning requires 590.4 minutes (9.84 GPU hours) and AAMax requires 597.4 minutes (9.96 GPU hours). AAMax therefore adds approximately 7.0 minutes, or 1.19\%, to the training run: 1.1 minutes for the alignment stage and about 5.9 minutes accumulated from the 2.2-second per-epoch difference. Moreover, AAMax does not require the full schedule: one epoch of end-to-end fine-tuning after alignment already outperforms 160 epochs of baseline training (Appendix~\ref{app:more_limited}).

\begin{table}[ht]
\centering
\caption{Training times for baseline and AAMax fine-tuning, averaged over the three fine-tuned subjects (Subjects 2, 5, and 7). The total time is the time to train for 160 epochs of end-to-end fine-tuning.}
\vskip 0.1in
\begin{tabular}{@{}lcccc@{}}
\toprule
\multirow{2}{*}{Training Method} & \multirow{2}{*}{\shortstack{Adapter Alignment \\ (Total Time)}} & \multicolumn{2}{c}{End-to-End Fine-tuning} & \multirow{2}{*}{\shortstack{Total Time \\ to Train}} \\
\cmidrule(lr){3-4}
                                 &                                                           & Time per Epoch      & Total Time (160 epochs)        &                                          \\
\midrule
Baseline                  & -                                                  & 221.40 seconds            & 590.40 minutes            & 590.40 minutes                             \\
AAMax                     & 67 seconds                                         & 223.61 seconds            & 596.29 minutes            & 597.41 minutes                              \\
\bottomrule
\end{tabular}
\label{tab:training_times}
\end{table}

\subsection{Model Size}
AAMax changes the optimization procedure and does not add a layer or an inference module. In every paired comparison, the baseline and AAMax use the same architecture and the same subject-specific adapters, so they have identical parameter counts (Table~\ref{tab:param_counts}). The counts include the subject-specific adapters, the shared voxel-to-CLIP network, and the diffusion prior, and exclude the frozen image generator (Versatile Diffusion/UnCLIP), which is identical for both methods. If AAMax is applied to a backbone that has no subject-specific adapter, the only additional subject-specific cost is a linear projection plus bias, with $D_s \times 4096 + 4096$ parameters for a subject with $D_s$ voxels (58.5M parameters for Subject 2, with 14{,}278 voxels).

\begin{table}[ht]
\centering
\caption{Learned parameters of the brain-decoding checkpoints, which are identical for baseline fine-tuning and AAMax. Subject adapters are those for Subjects 1 and 2.}
\vskip 0.1in
\footnotesize
\begin{tabular}{@{}lcccc@{}}
\toprule
Architecture & Subject adapters & Shared voxel-to-CLIP & Diffusion prior & Total \\
\midrule
MindEye1 & 122,912,768 & 883,162,112 & 56,056,564 & 1,062,131,444 \\
MindEye2 & 122,896,384 & 1,903,020,028 & 259,866,596 & 2,285,783,008 \\
\bottomrule
\end{tabular}
\label{tab:param_counts}
\end{table}

\subsection{Inference Time}
Because the baseline and AAMax checkpoints share the same inference graph, their inference costs are identical. Measured on the same H200 GPU, high-level MindEye1 reconstruction takes approximately 4.44 seconds per image (1 hour 13 minutes 58 seconds for 1{,}000 test images) with batch size 1, eight reconstruction candidates per sample, 100 diffusion-prior timesteps, and 20 UnCLIP inference steps. Most of this time is spent in diffusion-prior sampling and image generation rather than in the subject-specific adapter.

\subsection{Hyperparameter-Search Scope}
AAMax inherits the backbone fine-tuning hyperparameters of the baseline it improves, including the optimizer, learning-rate schedule, batch construction, number of fine-tuning epochs, and loss weights. We did not perform a separate per-subject search over these hyperparameters for either method; the baselines use the configurations recommended in the original MindEye1 and MindEye2 papers. The AAMax-specific settings are held fixed across all paired experiments: the alignment stage uses AdamW with a learning rate of $3 \times 10^{-4}$ for 300 epochs, and end-to-end fine-tuning uses an adapter-level MSE weight of 10. The adapter-architecture comparisons in Appendix~\ref{app:diff_adapters} are robustness ablations rather than a hyperparameter search for the reported results.

\subsection{Amortized Cost}
The reference model and the diffusion stack are part of the underlying architecture and are required equally by baseline fine-tuning and AAMax. Training the reference subject is a one-time cost that is amortized over every subsequently adapted subject, whereas the reductions in scan time and adaptation time achieved by AAMax apply to each new subject.

\section{High-Level Pipeline Training Process}\label{app:hlp_breakdown}

The high-level pipeline aims to map input fMRI signals to the CLIP embedding space. The training process differs based on whether the subject is the reference subject or a new subject being aligned to the reference space.

\subsection{Training the Reference Subject}

For the reference subject, the loss function combines Diffusion Prior Loss and CLIP Loss. The overall loss for the reference subject is:
\[
\mathcal{L}_{\text{ref}} = \lambda_1 \mathcal{L}_{\text{Prior}} + \lambda_2 \mathcal{L}_{\text{CLIP}}
\]
where \(\lambda_1\) and \(\lambda_2\) are weighting coefficients for the respective loss components.

\subsection{Training a New Subject}

For a new subject, the training involves two stages:

\paragraph{Stage 1: Adapter-Level Alignment.}
Initially, the adapter is trained using the Mean Squared Error (MSE) loss:
\[
\mathcal{L}_{\text{adapter}} = \mathcal{L}_{\text{MSE}}(\mathbf{z}_{\text{adapter, new}}, \mathbf{z}_{\text{adapter, ref}})
\]
where \(\mathbf{z}_{\text{adapter, new}}\) and \(\mathbf{z}_{\text{adapter, ref}}\) are the adapter-level embeddings for the new and reference subjects, respectively, on the shared common images.

\paragraph{Stage 2: End-to-End Training.}
After aligning the adapter, the entire pipeline is trained end-to-end. During this stage:
\begin{itemize}
    \item For non-common images, Prior Loss and CLIP Loss are applied to align the new subject’s output embeddings to the CLIP space.
    \item For common images, the MSE Loss is applied at the adapter level to retain the alignment to the reference subject as much as possible.
\end{itemize}

The overall loss function in this stage is:
\[
\mathcal{L}_{\text{new}} = 
\begin{cases} 
\lambda_1 \mathcal{L}_{\text{Prior}} + \lambda_2 \mathcal{L}_{\text{CLIP}}, & \text{for non-common images} \\
\lambda_1 \mathcal{L}_{\text{Prior}} + \lambda_2 \mathcal{L}_{\text{CLIP}} + \lambda_3 \mathcal{L}_{\text{MSE}}(\mathbf{z}_{\text{adapter,new}}, \mathbf{z}_{\text{adapter,ref}}), & \text{for common images}
\end{cases}
\]
where \(\mathbf{z}_{\text{adapter, new}}\) and \(\mathbf{z}_{\text{adapter, ref}}\) are the adapter-level embeddings for the new and reference subjects, respectively, on the shared common images.

\section{Low-Level Pipeline Training Process}\label{app:llp_breakdown}

The low-level pipeline aims to map input fMRI signals to the latent space of an Autoencoder (AE), which is trained to compress the NSD train images into a 1024-dimensional latent representation and reconstruct them. The training process differs based on whether the subject is the reference subject or a new subject being aligned to the reference space.

\subsection{Training the Reference Subject}

For the reference subject, the goal is to map the fMRI signals to the AE latent space using the Mean Squared Error (MSE) loss. The AE latent representations of the corresponding images serve as the ground truth.

The loss function for the reference subject is:
\[
\mathcal{L}_{\text{ref}} = \mathcal{L}_{\text{MSE}}(\mathbf{z}_{\text{fMRI}}, \mathbf{z}_{\text{AE}})
\]
where \(\mathbf{z}_{\text{fMRI}}\) is the mapped output embedding from fMRI signals, and \(\mathbf{z}_{\text{AE}}\) is the AE latent space representation of the corresponding image.

\subsection{Training a New Subject}

For a new subject, the training involves two stages:

\paragraph{Stage 1: Adapter-Level Alignment.}
In this stage, only the adapter is trained using MSE Loss to align the new subject’s adapter-level embeddings with those of the reference subject for shared common images:
\[
\mathcal{L}_{\text{adapter}} = \mathcal{L}_{\text{MSE}}(\mathbf{z}_{\text{adapter, new}}, \mathbf{z}_{\text{adapter, ref}})
\]
where \(\mathbf{z}_{\text{adapter, new}}\) and \(\mathbf{z}_{\text{adapter, ref}}\) are the adapter-level embeddings for the new and reference subjects, respectively.

\paragraph{Stage 2: End-to-End Training.}
In this stage, the entire pipeline is trained end-to-end. MSE Loss is applied at 2 points in the pipeline:
\begin{itemize}
    \item At the adapter level, to maintain alignment between the new subject’s embeddings and the reference subject’s embeddings for common images.
    \item At the mapper output level, to ensure the new subject’s fMRI embeddings align with the AE latent representations for corresponding images.
\end{itemize}

The loss function for this stage is:
\[
\mathcal{L}_{\text{new}} = 
\begin{cases} 
\mathcal{L}_{\text{MSE}}(\mathbf{z}_{\text{fMRI, new}}, \mathbf{z}_{\text{AE}}), & \text{for non-common images} \\
\mathcal{L}_{\text{MSE}}(\mathbf{z}_{\text{fMRI, new}}, \mathbf{z}_{\text{AE}}) + \mathcal{L}_{\text{MSE}}(\mathbf{z}_{\text{adapter, new}}, \mathbf{z}_{\text{adapter, ref}}), & \text{for common images}
\end{cases}
\]
where:
\begin{itemize}
    \item \(\mathbf{z}_{\text{fMRI, new}}\): Output embedding for the new subject.
    \item \(\mathbf{z}_{\text{AE}}\): AE latent space representation of the image.
    \item \(\mathbf{z}_{\text{adapter, new}}\) and \(\mathbf{z}_{\text{adapter, ref}}\): Adapter-level embeddings for the new and reference subjects, respectively.
\end{itemize}
\clearpage
\section{Visual Comparison Results among Different Methods}\label{app:vis_comp}
\begin{figure}[h]
    \centering
    \begin{subfigure}{0.50\textwidth}
        \centering
        \includegraphics[width=\textwidth]{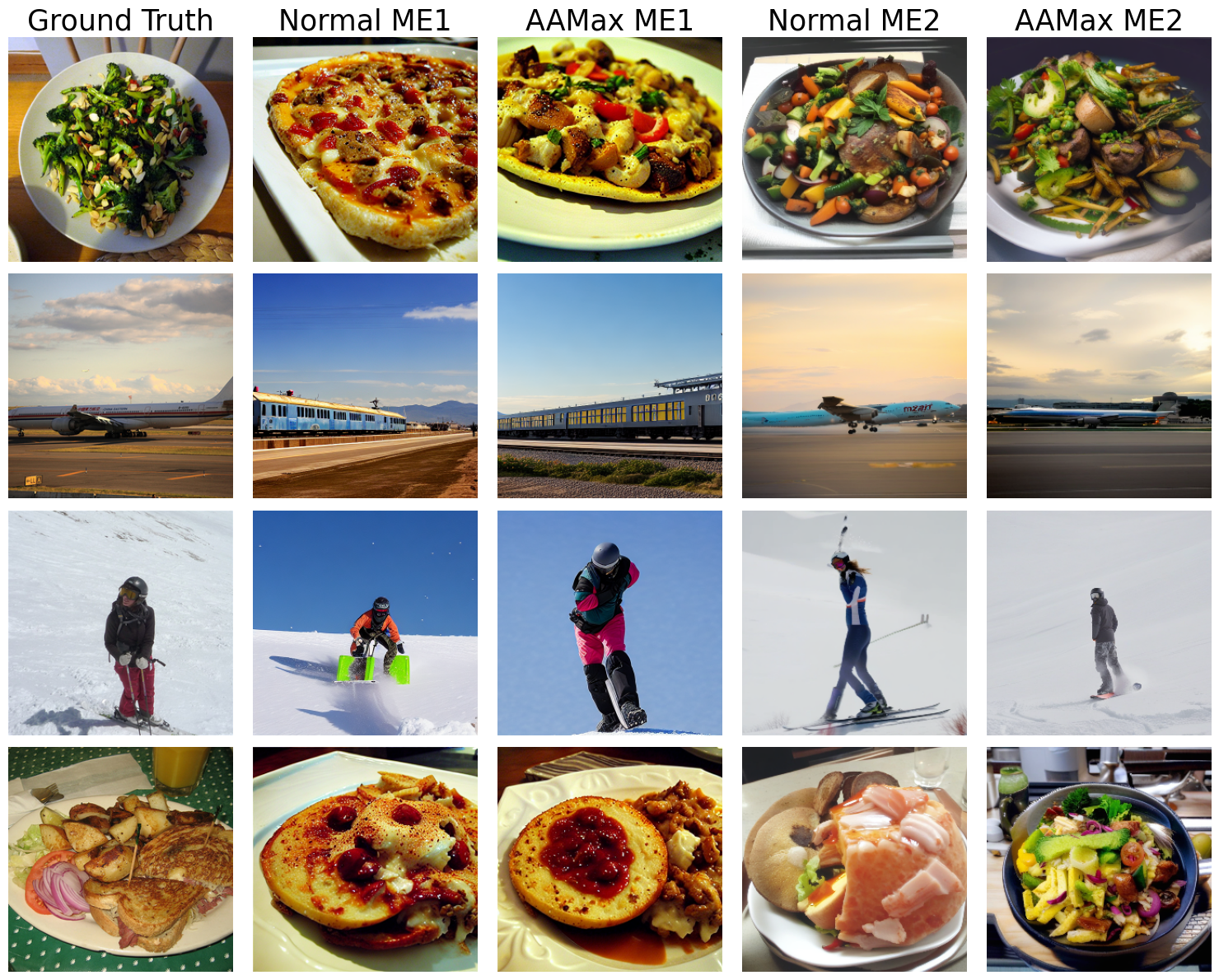}
    \end{subfigure}\hfill
    \caption{Visual comparison among MindEye 1 without AAMax (Normal ME1), MindEye1 with AAMax (AAMax ME1), MindEye 2 without AAMax (Normal ME2) and MindEye2 with AAMax (AAMax ME2).}
    \label{fig:four_mes}
\end{figure}
\section{Cross-Dataset Transfer and Correspondence Quality}\label{app:bold_5k}
Adapter Alignment requires a set of image pairs on which the new subject's adapter is aligned to the reference. Within NSD, these are the shared images viewed by all subjects. Across datasets, exact correspondences may or may not exist. This appendix describes how correspondences are constructed for BOLD5000 and THINGS-fMRI, and quantifies how the quality of the correspondence affects AAMax.

\subsection{Transfer between NSD and BOLD5000}
We conducted additional transfer experiments using the BOLD5000 \cite{chang2019bold5000} dataset. BOLD5000 includes a subset of COCO images that overlaps directly with the NSD shared images, so these experiments use exact image correspondences between the reference and the fine-tuned subject. Table \ref{tab:nsd_b5k_transfer} summarizes the results. These experiments allow us to directly compare how well models trained on NSD vs. BOLD5000 serve as reference subjects for fine-tuning. All fine-tuning runs use AAMax with 500 training samples for the fine-tuned subject on the ME1 architecture; because all rows use AAMax, the comparison isolates the effect of the reference dataset rather than the training procedure. The reference subjects across both datasets were trained on all subject specific available data.
\begin{itemize}

 \item Row 1 (B5K Subj1 Pre-train): Performance of a model trained from scratch on Subject 1 of BOLD5000.

 \item Row 2 (FT B5K Subj2 on B5K Subj1): Fine-tuning Subject 2 of BOLD5000 using Subject 1 of BOLD5000 as the reference.

 \item Row 3 (FT B5K Subj2 on NSD Subj1): Fine-tuning Subject 2 of BOLD5000 using NSD Subject 1 as the reference.

 \item Row 4 (FT NSD Subj2 on NSD Subj1): Fine-tuning NSD Subject 2 using NSD Subject 1 as the reference (the standard NSD→NSD setting).

 \item Row 5 (FT NSD Subj2 on B5K Subj1): Fine-tuning NSD Subject 2 using BOLD5000 Subject 1 as the reference.
\end{itemize}
Across these conditions, we consistently observe that NSD serves as a substantially stronger reference dataset. NSD-trained reference models yield higher-quality fine-tuning outcomes on both NSD and BOLD5000 target subjects. In contrast, BOLD5000-trained reference models lead to noticeably weaker reconstruction performance even when transferring within the same dataset. These results confirm the intuition that NSD’s high-resolution 7T acquisitions give rise to richer, more stable subject-specific representation spaces, which in turn make NSD the most reliable and effective choice for constructing a reference model. This finding aligns with our original rationale for using NSD as the backbone dataset for the common space.

\begin{table}[h]
    \centering
    \caption{Cross-dataset and within-dataset transfer performance between NSD and BOLD5000.}

    \vskip 0.1in
    \setlength{\tabcolsep}{2pt}  
    \renewcommand{\arraystretch}{1.1}  
    \footnotesize
    \resizebox{\columnwidth}{!}{
    \begin{tabular}{lccccccccc}
        \toprule
        & \multicolumn{4}{c}{\textbf{Low-Level}} & \multicolumn{4}{c}{\textbf{High-Level}} \\
        \cmidrule(lr){2-5} \cmidrule(lr){6-9}
        \textbf{Method} & PixCorr $\uparrow$ & SSIM $\uparrow$ & Alex(2) $\uparrow$ & Alex(5) $\uparrow$ & Incep $\uparrow$ & CLIP $\uparrow$ & Eff $\downarrow$ & SwAV $\downarrow$\\
        \midrule
    
        B5k Subj1 Pre-train 
            & {\color{gray}0.041} 
            & {\color{gray}0.211} 
            & 65.71\% 
            & 76.87\%
            & 66.24\% 
            & 74.77\% 
            & 0.928 
            & 0.612 \\
    \cline{2-9}
        FT B5k Subj2 on B5k Subj1 
            & {\color{gray}0.047} 
            & {\color{gray}0.247} 
            & 60.53\% 
            & 66.23\% 
            & \textbf{60.06\%} 
            & 64.63\% 
            & 0.952 
            & \textbf{0.631} \\
    
        FT B5k Subj2 on NSD Subj1 
            & {\color{gray}0.053} 
            & {\color{gray}0.252} 
            & \textbf{63.32\%} 
            & \textbf{69.13\%} 
            & 59.75\% 
            & \textbf{66.07\%} 
            & \textbf{0.948} 
            & 0.634 \\
    \cline{2-9}
        FT NSD Subj2 on B5k Subj1 
            & {\color{gray}0.115} 
            & {\color{gray}0.289} 
            & 79.76\% 
            & 86.73\% 
            & 79.57\% 
            & 79.88\% 
            & 0.825 
            & 0.499 \\
    
        FT NSD Subj2 on NSD Subj1 
            & {\color{gray}0.114} 
            & {\color{gray}0.270} 
            & \textbf{83.70\%} 
            & \textbf{91.76\%} 
            & \textbf{83.71\% }
            & \textbf{82.00\%} 
            & \textbf{0.789} 
            & \textbf{0.469} \\
    
        \bottomrule
    \end{tabular}
    }
    \label{tab:nsd_b5k_transfer}
\end{table}

\subsection{Constructing Approximate Correspondences for THINGS-fMRI}
No images are shared between THINGS-fMRI and NSD. For the experiment in Section~\ref{sec:things_dataset}, we therefore construct an approximate common set. We compute the CLIP cosine similarity between each THINGS training image and each of the 1{,}000 NSD shared images used to define the reference embeddings, and retain the pairs with similarity above 0.90, which yields approximately 150 semantically matched pairs. During adapter alignment, the THINGS subject's response to each retained image is aligned to the NSD reference subject's embedding of its matched NSD image. We chose the 0.90 threshold because Figure~\ref{fig:low_data_push} indicates that roughly 150 images suffice for strong performance, and this threshold yields that budget while keeping the matched pairs highly similar. THINGS images are mostly object-centric and visually simple, so pairs above this threshold typically depict the same object category in similar views, and high CLIP similarity behaves much like category-level matching.

\subsection{Effect of Correspondence Quality}\label{app:approx_pairs}
To quantify the cost of replacing exact correspondences with semantic ones, we repeated the NSD Subject 1 $\rightarrow$ Subject 2 experiment (1 hour of data, MindEye1) using 250 CLIP-matched image pairs instead of 250 identical shared images.
Because NSD images are complex natural scenes, close semantic matches are much rarer than for the object-centric THINGS images: whereas every THINGS pair exceeds the 0.90 threshold by construction, over 90\% of the matched NSD pairs fall below it (median CLIP similarity of approximately 0.79--0.86 across the matching directions used). Table~\ref{tab:approx_pairs} shows that AAMax with approximate pairs still improves over the baseline on every metric, but the gain is much smaller than with exact pairs, which are substantially better on all metrics except SSIM.

\begin{table}[h]
    \centering
    \caption{Effect of correspondence quality on AAMax. Subject 1 is the reference and Subject 2 is fine-tuned with 1 hour of data. Approximate pairs replace the identical shared images with CLIP-matched images. All rows are averaged over three independently generated splits. Low-level metrics are shown in gray because they are not optimized in these runs.}
    \vskip 0.1in
    \setlength{\tabcolsep}{2pt}  
    \renewcommand{\arraystretch}{1.1}  
    \footnotesize  
    \resizebox{0.8\columnwidth}{!}{  
    \begin{tabular}{lcccccccc}
        \toprule
        & \multicolumn{4}{c}{\textbf{Low-Level}} & \multicolumn{4}{c}{\textbf{High-Level}} \\
        \cmidrule(lr){2-5} \cmidrule(lr){6-9}
        \textbf{Method} & PixCorr $\uparrow$ & SSIM $\uparrow$ & Alex(2) $\uparrow$ & Alex(5) $\uparrow$ & Incep $\uparrow$ & CLIP $\uparrow$ & Eff $\downarrow$ & SwAV $\downarrow$\\
        \midrule
        Baseline & {\color{gray}0.078} & {\color{gray}0.264} & 75.33\% & 84.18\% & 74.15\% & 73.66\% & 0.862 & 0.524 \\
        AAMax, approximate pairs & {\color{gray}0.086} & {\color{gray}0.275} & 78.11\% & 85.74\% & 74.79\% & 75.29\% & 0.843 & 0.518 \\
        AAMax, exact pairs & {\color{gray}0.112} & {\color{gray}0.259} & \textbf{78.75\%} & \textbf{88.04\%} & \textbf{78.99\%} & \textbf{78.28\%} & \textbf{0.824} & \textbf{0.513} \\
        \bottomrule
    \end{tabular}
    }
    \label{tab:approx_pairs}
\end{table}

These results support a more precise conclusion than exact matching alone would: shared images are not strictly required, but the quality of the correspondence strongly affects the quality of the alignment, especially for complex natural scenes, where CLIP-level semantic similarity does not guarantee proximity in the learned fMRI representation. We do not include random matching as a baseline, because it would explicitly force unrelated stimuli toward the same reference embedding and thus inject contradictory supervision rather than provide a neutral control. Category-level matching is represented by MindAligner (Section~\ref{sec:train_strat}, Appendix~\ref{app:mind_align}), which aligns subjects in voxel space using category-level pairs; combining it with AAMax gives the strongest results, suggesting that coarse category-level alignment and fine-grained common-space alignment are complementary. When only weak matches are available, a hybrid strategy is therefore likely to be most effective: coarse semantic or voxel-space alignment, followed by Adapter Alignment on the highest-confidence subset of pairs.

\section{Combining Adapter Alignment and Image Selection with MindAligner}\label{app:mind_align}
We conducted a direct comparison between MindAligner and our Adapter Alignment (AA) method using the MindEye1 architecture under the same experimental protocol used throughout the paper. As before, we pretrain the full model on Subject 1 and fine-tune on Subject 2 using 1 hour of data. We evaluate three configurations:
\begin{enumerate}
    \item \textbf{MindAligner fine-tuning:} We applied MindAligner exactly as described in \citep{dai2025mindaligner}, including the generation of categorical stimulus pairs, even though the training data contain exact shared stimuli. This procedure yields significantly more training examples than the 250-image subset used in our low-data experiments, and the model is trained using the original MindAligner voxel-space loss functions.
    \item \textbf{Adapter Alignment + Image Selection:} Using our proposed image-selection algorithm, we generate a 250-image subset that spans the breadth of the principal eigenspace of Subject 1’s shared representation. This ensures coverage across all principal axes of variation, enabling AA to perform exact one-to-one alignment in the shared space. This configuration mirrors the setup used in all other AA experiments.
    \item \textbf{MindAligner + Adapter Alignment + Image Selection:} MindAligner and AA operate at different levels of the pipeline. MindAligner performs coarse categorical voxel-space alignment, and AA performs fine-grained common-space alignment. This lets us build and evaluate a multi-stage pipeline. Training a new subject is decomposed into four steps:
    \begin{enumerate}
        \item fine-tune the BFA (Brain Feature Adapter) to reconstruct category-level fMRI for Subject 1;
        \item apply Adapter Alignment by overfitting the new subject's adapter to Subject 1's shared-space embeddings on the new 1-hour dataset;
        \item perform end-to-end training of the BFA module with the remainder of the architecture frozen;
        \item perform AA end-to-end fine-tuning with the BFA frozen.
    \end{enumerate}
    This sequence allows MindAligner to establish coarse alignment at the voxel level, while AA and Image Selection refine the geometry of the shared space. 
\end{enumerate}

\section{Summary statistics}\label{app:sum_stat}

In this section, we present summary statistics of the results reported in Table \ref{tab:fulldatacomparison}. 

\begin{table}[h]

    \centering
    \caption{Summary statistics of the quantitative results on training the reconstruction models on all 40 hours of data. The fine-tuning results are averaged over subjects 2,5 and 7. AAMax and baseline fine-tuning perform similarly across all metrics.}
    \vskip 0.1in
    \setlength{\tabcolsep}{2pt}  
    \renewcommand{\arraystretch}{1.1}  
    \footnotesize  

    \resizebox{0.99\columnwidth}{!}{  
    
    \begin{tabular}{lccccccccc}
        \toprule
        & \multicolumn{4}{c}{\textbf{Low-Level}} & \multicolumn{4}{c}{\textbf{High-Level}} \\
        \cmidrule(lr){2-5} \cmidrule(lr){6-9}
        \textbf{Method} & PixCorr $\uparrow$ & SSIM $\uparrow$ & Alex(2) $\uparrow$ & Alex(5) $\uparrow$ & Incep $\uparrow$ & CLIP $\uparrow$ & Eff $\downarrow$ & SwAV $\downarrow$\\
        \midrule
        
        Subject 1 Pre-train & 0.345 & 0.346 & 92.80\% & 96.88\% & 94.40\% & 90.02\% & 0.692 & 0.399 \\ 
        \cline{2-9}
        Baseline FT AVG & 0.258$\pm$0.03 & \textbf{0.339$\pm$0.005} & \textbf{89.64\%$\pm$1.8\%} & 95.05\%$\pm$0.87\% & 92.63\%$\pm$1.62\% & \textbf{89.64\%$\pm$1.55\%} & 0.717$\pm$0.017 & 0.427$\pm$0.01 \\  
        AAMax FT AVG & \textbf{0.259$\pm$0.03} & 0.337$\pm$0.00 & 89.20\%$\pm$2.27\% & \textbf{95.16\%$\pm$1.03\%} & \textbf{92.68\%$\pm$1.69\%} & 89.49\%$\pm$1.25\% & \textbf{0.713$\pm$0.02} & \textbf{0.423$\pm$0.02} \\  
        \bottomrule
    \end{tabular}
    }
    \label{tab:fulldatacomparison_sum}
\end{table}

\section{Scaling with Training Data}\label{app:scaling}

We further examine how the benefits of Adapter Alignment scale with training data size. Using Subject 2 as the fine-tuning subject, we progressively increase the training set from 250 to 500, 1000, and 4000 images, before extending to the full dataset. Up to 1000 images, only common images are used for alignment and fine-tuning; for 4000 images, alignment is still performed on the 1000 shared images, while the additional 3000 images are used only for fine-tuning. As shown in Figure~\ref{fig:limited_scaling}, AAMax achieves substantially higher performance than baseline fine-tuning in low-data regimes, with the gap narrowing as training data increases. At full data, both approaches converge, consistent with prior multi-subject pipelines such as MindEye2, where large amounts of subject-specific data allow the model to generalize even without explicit alignment.

\begin{figure*}[h]
    \centering
    \begin{subfigure}{0.29\textwidth}
        \centering
        \includegraphics[width=\columnwidth]{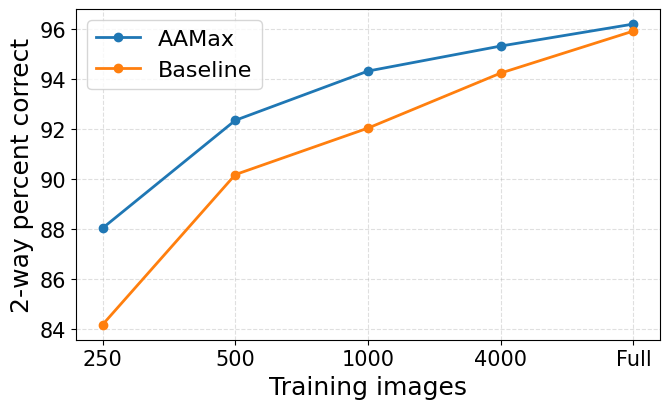}
    \end{subfigure}\hfill
    \begin{subfigure}{0.29\textwidth}
        \centering
        \includegraphics[width=\columnwidth]{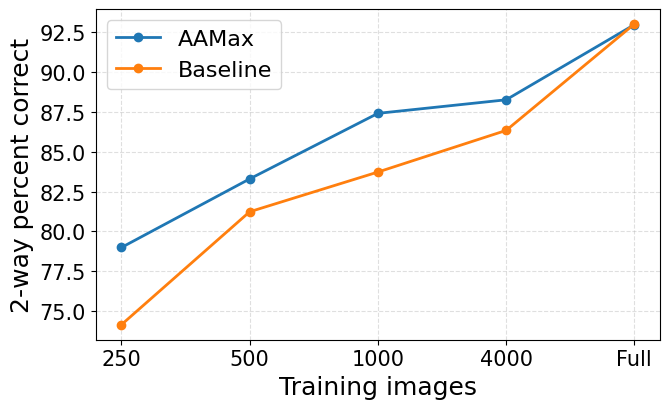}
    \end{subfigure}\hfill
    \begin{subfigure}{0.29\textwidth}
        \centering
        \includegraphics[width=\columnwidth]{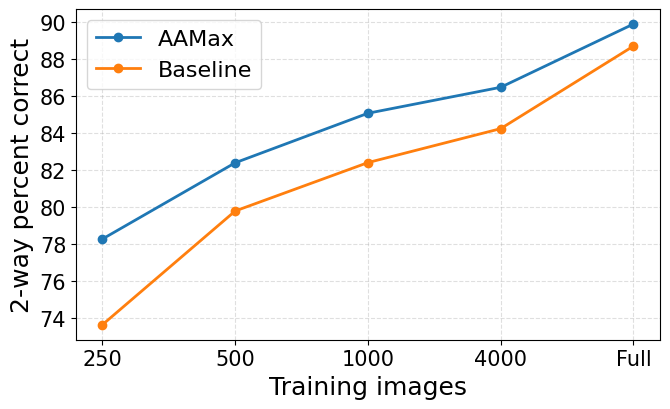}
    \end{subfigure}
    \caption{Performance of AAMax compared to baseline fine-tuning as a function of training data size. Left: AlexNet-5 2-way percent correct. Middle: Inception 2-way percent correct. Right: CLIP 2-way percent correct. Across all three metrics, AAMax achieves higher accuracy in limited-data regimes, with the advantage diminishing as more data becomes available.}
    \label{fig:limited_scaling}
\end{figure*}
\section{Generalizing To Different Architectures}\label{app:diff_archs}
Adapter Alignment is inherently architecture-agnostic: any reconstruction model can be extended with subject-specific adapters, enabling alignment to a common space. For efficiency and comparability with prior work, all of our main experiments are conducted using a slightly modified MindEye1 pipeline. To demonstrate the architecture-agnostic nature of Adapter Alignment, we additionally evaluate a limited-data setting using 1 hour of training data to fine-tune Subject 2. We compare baseline fine-tuning and AAMax under both the MindEye1 and MindEye2 architectures, and observe consistent improvements with AAMax across both models. Additional experiments using MindEye2 are presented in Appendix \ref{app:me2_diffsub}.

\begin{table}[h]
    \centering
    \caption{Evaluating Adapter Alignment across architectures. Subject~1 is trained as the reference, and Subject~2 is fine-tuned with 1 hour of training data. Results are shown for MindEye1 and MindEye2 under baseline and AAMax fine-tuning. Adapter Alignment consistently improves performance and scales with stronger architectures.}
    \vskip 0.1in
    \setlength{\tabcolsep}{3pt}
    \renewcommand{\arraystretch}{1.15}
    \resizebox{0.6\columnwidth}{!}{
    \begin{tabular}{lccccccccc}
        \toprule
        & \multicolumn{4}{c}{\textbf{Low-Level}} & \multicolumn{4}{c}{\textbf{High-Level}} \\
        \cmidrule(lr){2-5} \cmidrule(lr){6-9}
        \textbf{Method}
        & PixCorr $\uparrow$ & SSIM $\uparrow$ & Alex(2) $\uparrow$ & Alex(5) $\uparrow$
        & Incep $\uparrow$ & CLIP $\uparrow$ & Eff $\downarrow$ & SwAV $\downarrow$ \\
        \midrule
        \multicolumn{9}{l}{\textbf{MindEye1}} \\
        \midrule
        Baseline
            & \color{gray}{0.078} & \color{gray}{0.264} & 75.33\% & 84.18\%
            & 74.15\% & 73.66\% & 0.862 & 0.524 \\
        AAMax
            & \color{gray}{0.112} & \color{gray}{0.259} & 78.75\% & 88.04\%
            & 78.99\% & 78.28\% & 0.824 & 0.513 \\
        \midrule
        \multicolumn{9}{l}{\textbf{MindEye2}} \\
        \midrule
        Baseline
            & \color{gray}{0.104} & \color{gray}{0.354} & 81.71\% & 88.06\%
            & 78.84\% & 74.15\% & 0.845 & 0.499 \\
        AAMax
            & \color{gray}{0.140} & \color{gray}{0.355} & \textbf{82.22\%} & \textbf{89.67\%}
            & \textbf{82.97\%} & \textbf{78.79\%} & \textbf{0.821} & \textbf{0.483} \\
        \bottomrule
    \end{tabular}}
    \label{tab:archcomparison}
\end{table}

\section{Limitations}\label{app:limit}

\paragraph{Few-shot adaptation, not zero-shot generalization.} Adapter Alignment reduces the amount of subject-specific data needed to model a new participant; it does not remove the need for such data. Every new subject requires a lightweight subject-specific adapter and a small calibration set viewed by, or semantically matched to, the reference subject. Our evaluations therefore concern low-data adaptation to subjects that were not used to pre-train the reference model, rather than zero-shot transfer to unseen subjects. The comparable reconstruction methods we build on or compare against likewise require subject-specific data for every new subject.

\paragraph{Calibration correspondences.} Explicit alignment requires image pairs between the new subject and the reference. With exact shared images the method is most effective. When only semantically matched images are available, AAMax still improves over baseline fine-tuning, but its advantage shrinks as the quality of the correspondence decreases, particularly for complex natural scenes (Appendix~\ref{app:bold_5k}). The CLIP-based matching used for THINGS-fMRI may impose part of the shared structure through the CLIP representation rather than recovering it from the brain data alone.

\paragraph{Dependence on the reference subject.} AAMax outperforms baseline fine-tuning for every reference subject we evaluated with Subject 1 as the target, but absolute performance depends on the quality and compatibility of the reference (Appendix~\ref{app:diff_refsub}), and a reference model trained on a lower-SNR dataset yields weaker fine-tuning (Appendix~\ref{app:bold_5k}). Pre-training the reference on multiple subjects did not improve over a single strong reference. In practice, this favors building the reference from a high-quality dataset such as NSD; when no strong reference is available for a given scanner or stimulus set, the gains of AAMax are expected to be smaller.

\paragraph{Evaluation protocol.} Our main NSD experiments use a modified split in which the shared images are moved into training. All methods are retrained and evaluated on identical splits, and a matched rerun shows that the modified test set is harder than the standard one (Appendix~\ref{app:standard_split}), but the protocol still limits direct comparison with results reported under the standard NSD benchmark. The baselines use the hyperparameters recommended in their original papers and were not independently retuned for the modified split; this isolates the effect of the training procedure but does not establish that each baseline is optimally tuned for this protocol. Several supporting analyses, namely the image-selection baselines (Appendix~\ref{app:increase_evdim}), pre-diffusion retrieval (Appendix~\ref{app:retrieval}), and the neural correspondence analysis (Appendix~\ref{app:common_space}), use a single training run and a single target subject.

\paragraph{Neural interpretation.} Like most reconstruction pipelines, AA treats each subject's fMRI pattern as a vector, and alignment is learned after the subject-specific adapter has absorbed differences in voxel count, anatomy, and cortical organization. Our pRF- and fLOC-resolved analyses show improved functional alignment of signals from homologous visual regions, but they do not establish anatomical or vertex-level correspondence, and AA does not use cortical mesh geometry during training.

\paragraph{System-level cost.} AA reduces the scan time and computation needed to adapt each new subject. It does not reduce the one-time cost of training the reference model or of the frozen generative stack, which are determined by the underlying architecture (Appendix~\ref{app:tr_time}).

\paragraph{Generative prior and signal quality.} The ultimate image fidelity hinges on the strengths and biases of the pre-trained diffusion model, so any generative shortcomings propagate into our outputs. In addition, variability and noise in the fMRI signals (e.g., scanner differences, low SNR) can further degrade adapter alignment and reconstruction accuracy.

\paragraph{Ethical considerations.} Because fMRI-to-image reconstruction systems process sensitive brain data, their development raises important privacy, consent, and misuse concerns. These methods should not be interpreted as general-purpose ``mind reading,'' and any future deployment, especially in patient or clinical populations, must include strict safeguards for informed consent, data security, transparency about reconstruction uncertainty, and protection against coercive or unauthorized use. This work should therefore be viewed as a step toward more efficient scientific and assistive neural decoding, rather than as a technology ready for unconstrained real-world application.

\clearpage
\section{Reconstruction Results}\label{app:recons}
\begin{figure*}[h]
    \centering
    \begin{subfigure}{0.37\textwidth}
        \centering
        \includegraphics[width=\textwidth]{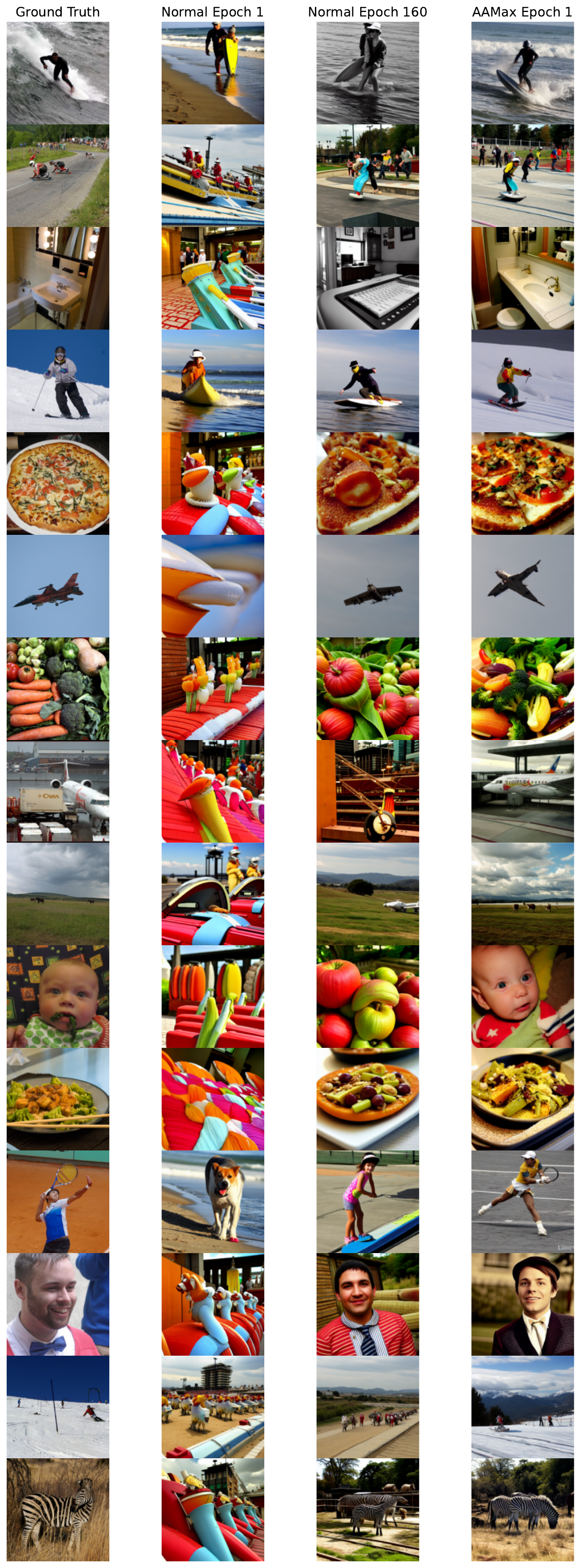}
    \end{subfigure}
    \caption{Reconstruction results after fine-tuning on 250 common images.}
    \label{fig:recons}
\end{figure*}

\end{document}